\documentclass[11pt, a4paper, reqno]{amsart}

\usepackage{amsxtra}
\usepackage{graphicx}
\usepackage{newlfont}
\usepackage{amscd}
\usepackage{hyperref}
\usepackage[german,english]{babel}
\usepackage{cancel}
\usepackage{amsmath,amsthm}
\usepackage{amssymb}
\usepackage{enumerate}
\usepackage{color}
\usepackage[all,cmtip]{xy}
\usepackage{kotex}
\usepackage{tikz,ytableau}
\usepackage{mathrsfs}
\usepackage{setspace}
\usepackage{MnSymbol}
\usepackage{comment}

\DeclareFontFamily{U}{mathb}{\hyphenchar\font45}
\DeclareFontShape{U}{mathb}{m}{n}{
<-6> mathb5 <6-7> mathb6 <7-8> mathb7
<8-9> mathb8 <9-10> mathb9
<10-12> mathb10 <12-> mathb12
}{}
\DeclareSymbolFont{mathb}{U}{mathb}{m}{n}
\DeclareMathSymbol{\llcurly}{\mathrel}{mathb}{"CE}
\DeclareMathSymbol{\ggcurly}{\mathrel}{mathb}{"CF}

\newtheorem{thm}{Theorem}[section]

 \newtheorem{lem}[thm]{Lemma}
 \newtheorem{prop}[thm]{Proposition}

 \theoremstyle{definition}
 \newtheorem{defn}[thm]{Definition}
  
  \newtheorem{defn-thm}[thm]{Definition-Theorem}
   \newtheorem{ex}[thm]{Example}
   \newtheorem{remark}[thm]{Remark} 

 \theoremstyle{remark}

\numberwithin{equation}{section}

\numberwithin{thm}{section}

\numberwithin{table}{section}

\numberwithin{figure}{section}

\ifx\pdfoutput\undefined
  \DeclareGraphicsExtensions{.pstex, .eps}
\else
  \ifx\pdfoutput\relax
    \DeclareGraphicsExtensions{.pstex, .eps}
  \else
    \ifnum\pdfoutput>0
      \DeclareGraphicsExtensions{.pdf}
    \else
      \DeclareGraphicsExtensions{.pstex, .eps}
    \fi
  \fi
\fi

\newcommand{\ZZ}{\mathbb{Z}}

\newcommand{\CC}{\mathbb{C}}

\newcommand{\V}{\mathcal{V}}

\newcommand{\noin}{\noindent}

\newcommand{\g}{\mathfrak{g}}

\newcommand{\m}{\mathfrak{m}}
\newcommand{\n}{\mathfrak{n}}
\newcommand{\p}{\mathfrak{p}}

\newcommand{\ad}{\text{ad}}

\begin{document}

\title{Dirac reductions and classical W-algebras}
\author{GAHNG SAHN LEE, ARIM SONG, UHI RINN SUH }
\address{}
\email{}
\thanks{This work was supported by NRF Grant \# 2019R1F1A10593631, \#2022R1C1C1008698 and Creative-Pioneering Researchers Program through Seoul National University}

\begin{abstract}
In the first part of this paper, we generalize Dirac reduction to the extent of non-local Poisson vertex superalgebra and non-local SUSY Poisson vertex algebra cases. Next, we modify this reduction so that we explain the structures of classical W-superalgebras and SUSY classical W-algebras in terms of the modified Dirac reduction.
\end{abstract}

\maketitle


\setcounter{tocdepth}{-1}

\pagestyle{plain}

\section{Introduction}\label{Sec:Introduction}

Dirac reductions of a Poisson algebra were first introduced in the paper \cite{Dirac50} of Dirac. For a Poisson algebra  $(P, \{  \ , \ \} )$ and a finite subset $\{\theta_i|\,i\in I\}\subset P$, if the $|I| \times |I|$ matrix $C$ with the $ij$-entry $C_{ij}= \{\theta_i, \theta_j\}$ is invertible, then 
\begin{equation} \label{eq:Dirac-original}
\{a,b\}^D:=\{a,b\} - \sum_{i\in I } \{a, \theta_i\} (C^{-1})_{ij} \{\theta_j, b\}, \quad a,b\in P
\end{equation}
defines a new Poisson bracket called Dirac reduced bracket. Since this new bracket satisfies $\{\theta_i, P\}^D=0$ for any $i\in I$, it induces a Poisson bracket on $P/I_P$, where $I_P$ is the associative algebra ideal generated by $\{\theta_i|i\in I\}$.
Moreover, in   \cite{DSKV14},  De Sole-Kac-Valeri  generalized the notion to Poisson vertex algebras (PVAs) cases. The Dirac reduced bracket of a Poisson vertex algebra $(\mathcal{P}, \{\ {}_\lambda \ \})$ associated with the finite subset $\{\theta_i | i\in I\}\subset \mathcal{P}$ is defined by
\begin{equation} \label{eq:Dirac-PVA original}
  \left\{ a \,_{\lambda}\, b \right\}^{D}=\left\{ a \,_{\lambda}\, b \right\}-\displaystyle\sum_{i,j \in I}\left\{ \theta_j \,_{\lambda+\partial}\, b \right\}_{\rightarrow}(C^{-1})_{ji}(\lambda+\partial)\left\{ a \,_{\lambda}\, \theta_i \right\}, \quad a,b\in \mathcal{P},
\end{equation}
provided that the $|I| \times |I|$-matrix $C(\lambda)=(\{\theta_j {}_\lambda \theta_i\})_{i,j\in I}$ is invertible with respect to the product $A(\lambda) \circ B(\lambda)= A(\lambda+\partial)B(\lambda)$ and $(C^{-1})(\lambda)$ is the inverse of $C(\lambda)$. Similarly to the Poisson algebra case, the bracket $\{ \ {}_\lambda \ \}^D$ induces a Poisson vertex algebra bracket on $\mathcal{P}/I_{\mathcal{P}}$, where $I_{\mathcal{P}}$ is the differential algebra ideal generated by $\{\theta_i | i\in I\}$.

Integrable Hamiltonian systems have been studied in the framework of PVAs by the authors in \cite{BDSK09, DSKV18, DSKV13}  and 
Dirac reductions are useful to obtain new systems  \cite{DSK13}. For example, the bi-Poisson structure of the Gelfand-Dickey hierarchy associated with the Lax operator $L= \partial^n + u_2 \partial^{n-2} +\cdots + u_n$ can be understood as a Dirac reduction of the Poisson structure induced from $\tilde{L}= \partial^n + u_1 \partial^{n-1}+u_2 \partial^{n-2} +\cdots + u_n$ (see \cite{Adl79, GelDic76,GelDic78}). Moreover, if we admit non-local Poisson structures for Hamiltonian equations, more systems can be understood by Dirac reductions. See \cite{DSK13} for detailed explanations and examples.

The third author described in \cite{Suh18} integrable Hamiltonian systems associated with a certain family of Poisson vertex superalgebras (PVSAs). Hence a natural question is whether we can get reduced systems via super-analogue of non-local PVAs and Dirac reductions. All the notions and properties for non-local PVAs in \cite{DSK13} can be generalized to the super-analogue by Koszul-Quillen sign rule. On the other hand, we introduce in Definition \ref{Def:Dirac modified super} {\it the Dirac reduction for PVSA $\mathcal{P}$.} The main idea is that we consider  $C(\lambda) =(\{\theta_j {}_\lambda \theta_i\})_{i,j\in I}$ as an element in $\mathfrak{gl}_{(m|n)} \otimes \mathcal{P}(\!(\lambda^{-1})\!)$ where $m$ and $n$ are the number of even and odd indices in $I$. If $C(\lambda)$ is invertible
via the multiplication in \eqref{eq:matrix product PVA}, then the  Dirac reduction of the Poisson $\lambda$-bracket $\left\{ \ \,_\lambda\, \ \right\}$ on $\mathcal{P}$ associated with $\theta_{I}=\{\theta_i| i\in I\}$ is given by
  \begin{equation} \label{super Dirac-intro}
  \left\{ a \,_{\lambda}\, b \right\}^{D}=\left\{ a \,_{\lambda}\, b \right\}-\displaystyle\sum_{i,j \in I} (-1)^{(p(a)+p(j))(p(b)+p(i))+p(i)+p(j)}\left\{ \theta_{j} \,_{\lambda+\partial}\, b \right\}_{\rightarrow}(C^{-1})_{ji}(\lambda+\partial)\left\{ a \,_{\lambda}\, \theta_{i} \right\},
\end{equation}
where $a,b$ are homogeneous elements in $\mathcal{P}$ and $p(i)=p(\theta_i)$.  In Theorem \ref{Theorem:super Dirac, part1} and Theorem \ref{Theorem:super Dirac, part2}, we show the following statements.

\begin{thm} \label{Theorem:super Dirac, intro}
 Let $\left< \theta_{I}\right>$ be the differential superalgebra ideal of $\mathcal{P}$ generated by $\theta_{I}$.
\begin{enumerate}
\item The bracket $  \left\{ \ \,_{\lambda}\, \  \right\}^{D}$ in  \eqref{super Dirac-intro} gives another PVSA structure on $\mathcal{P}$. In other words, it satisfies the sesquilinearity, skewsymmetry, Jacobi identity and Leibniz rule for PVSAs.
\item The Dirac reduced bracket induces a PVSA $\lambda$-bracket on $\mathcal{P}/ \left< \theta_{I}\right>$.
\end{enumerate}
\end{thm}

Furthermore, we also consider supersymmetric (SUSY) PVAs, which are PVSAs with an odd derivation $D$. SUSY vertex algebras and SUSY PVAs introduced in  \cite{K, HK07} to provide algebraic frameworks for SUSY conformal field theory.  Moreover, some superintegrable Hamiltonian systems can be understood in terms of SUSY PVAs. See \cite{CS}, for instance. The simplest example of superintegrable systems described by a SUSY PVA is the super-KdV equation: 
\[ \frac{du}{dt}= D^6 u + 3 D^{2}u Du + 3 u D^3u = \{ uDu\,  {}_\Lambda\,  u\} |_{\chi=\lambda=0}, \]
where $\{u \, {}_\Lambda \, u\} = (D^2+ \frac{3}{2}\lambda + \frac{1}{2} \chi D) u-\lambda^2 \chi.$ In spite of the fact that the second Poisson structure of the super-KdV has a non-local property, non-local super Hamiltonian operators have not been studied in the theory of SUSY PVAs. See \cite{MR85, MA92, OP91} for the SUSY bi-Poisson structure of the super KdV.

As the first step to understand non-local super Hamiltonians, we introduce the notion of {\it non-local SUSY PVAs} in Definition \ref{Def:non-local SUSY}. Simply speaking, a non-local SUSY PVA $\mathscr{P}$
is  a $\CC[D]$-module  endowed with a non-local $\Lambda$-bracket
\begin{equation} \label{eq:intro-non-local SUSY PVA}
\{ \ {}_\Lambda \ \}: \mathscr{P}\otimes \mathscr{P} \to \mathscr{P}[\chi](\!( \lambda^{-1})\!),
\end{equation} 
which induces an admissible non-local Poisson $\Lambda$-bracket and satisfies the skewsymmetry, Jacobi identity and Leibniz rule. The {\it SUSY Dirac reduced $\Lambda$-bracket} on $\mathscr{P}$ associated with the finite subset $\theta_{I}:= \{ \theta_i | i \in I\}\subset \mathscr{P}$ is defined as follow. Assuming the invertibility of $C(\Lambda):=(\{ \theta_{j} \,{}_{\Lambda}\, \theta_{i} \} )_{i,j \in I}$,
the \textit{Dirac reduced bracket} of  $\mathscr{P}$  associated with $\theta_{I}$ is given by 
  \begin{equation}\label{SUSY Dirac-intro}
    \{a{}_{\Lambda} b\}^{D}=\{a{}_{\Lambda}b\}-\sum_{i,j \in I}(-1)^{(p(a)+j)(p(b)+i)}\{\theta_{j}\,{}_{\Lambda+\nabla}\,b\}_{\rightarrow}(C^{-1})_{ji}(\Lambda+\nabla)\{a\,{}_{\Lambda}\,\theta_{i}\}
  \end{equation}
  for homogeneous $a,b\in \mathscr{P}$. Then the SUSY analogue of Theorem \ref{Theorem:super Dirac, intro} holds. See Theorem \ref{Theorem:SUSY Dirac, part1} and Theorem \ref{Theorem:SUSY Dirac, part2}.
  
  \begin{thm} Let $\left< \theta_{I}\right>$ be the differential superalgebra ideal of $\mathscr{P}$ generated by $\theta_{I}$. 
\begin{enumerate}
\item The bracket $  \left\{ \ \,_{\Lambda}\, \  \right\}^{D}$ in  \eqref{SUSY Dirac-intro} gives another SUSY PVA structure on $\mathscr{P}$. In other words, it satisfies the sesquilinearity, skewsymmetry, Jacobi identity and Leibniz rule for SUSY PVAs.
\item The Dirac reduced bracket on $\mathscr{P}$ induces a SUSY Poisson  $\Lambda$-bracket on $\mathscr{P}/\left< \theta_{I}\right>$.
\end{enumerate}
\end{thm}

\vskip 3mm

In the second part of this paper, we deal with classical W-algebras.  Let $\g$ be a simple Lie algebra with a Dynkin grading $\bigoplus_{i\in \frac{\ZZ}{2}} \g(i)$ associated with an $\mathfrak{sl}_2$-triple $\{E,H,F\}$ and a nondegenerate symmetric invariant bilinear form $(\, |\, )$. Let $\m= \mathcal{L} \oplus \bigoplus_{i \geq 1} \g{(i)}$, where $\mathcal{L}$ is a Lagrangian subspace with respect to the skew symmetric bilinear form $\left< \, \cdot\,  |\, \cdot \,  \right>:= (F|[\cdot, \cdot])$ on $\g_{\frac{1}{2}}$. The classical finite W-algebras $\mathcal{W}^{\text{fin}}(\g,F)$  is a Hamiltonian reduction of $\mathbb{C}[\g^*]\simeq S(\g)$ via the the moment map $\mu: \g^* \to \m^* $ at the regular value $\chi|_\m$ for   $\chi= (F|\,\cdot)\in \g^*$.  It is known that, for any  isotropic subspace $\ell$, if we take  $\n := \ell^{\perp} \oplus \bigoplus_{i \geq 1} \g{(i)}$  and replace $\m$ by $\m := \ell \oplus \bigoplus_{i \geq 1} \g{(i)}$ then 
\[\mathcal{W}^{\text{fin}}(\g,F)= \{ w \in S(\p)\,|\, \rho(\{n, w\})=0 \text{ for } n\in \n\}.\]
Here $\g=\m\oplus \p$ and $\rho(a)= a_{\p} + (F|a)$ for $a\in \g$, $a_{\p} \in \p$ and $a_{\m} \in \m$ such that $a= a_{\m}+a_{\p}$. As an affine analogue,  Drinfeld-Sokolov \cite{DS85} introduced the classical affine W-algebra $\mathcal{W}^k(\g, F)$ as a Hamiltonian reduction of the affine PVA $\mathcal{V}^k(\g)$.
Note that, in this paper, we fix $\ell=0$. Hence, $\n= \bigoplus_{i>  0} \g(i)$ and $\m= \bigoplus_{i \geq 1} \g(i)$.

Kac-Roan-Wakimoto \cite{KRW04} introduced W-algebras associated with Lie superalgebras and the corresponding classical finite W-algebras were established in \cite{DK06}, which we call W-superalgebras in this paper.
For a finite simple Lie superalgebra $\g$ with a nondegenerate supersymmetric invariant even bilinear form $(\ | \ )$, let us take an $\mathfrak{sl}_2$-triple $(E,H,F)$, which gives rise to the Dynkin grading $\g  = \bigoplus_{i\in \frac{\mathbb{Z}}{2}} \g(i)$.
 The classical W-superalgebra 
 \[\mathcal{W}^k(\g, F)= \{\, w  \in \mathcal{P}(\p)\,  |\, \rho\{ n \, {}_\lambda \, w\} =0 \text{ for } n\in \n \, \}\]
 is a PVSA endowed with the $\lambda$-bracket induced from $\mathcal{V}^k(\g)$. Here $\mathcal{P}(\p)$ is the differential superalgebra generated by $\p$. In order to describe free generators of $\mathcal{W}^k(\g,F)$, consider a homogeneous basis $\{\, v_i \, |\,  i\in J^F\}$ of $\ker (\ad F)=: \g^F$.  
 Then there is a unique element $\omega_i \in \mathcal{W}^k(\g,F)$ such that $\omega_i-v_i \in   \bigoplus_{n\geq 1} \mathcal{P}(\g^F) \otimes (\mathbb{C}[\partial]\otimes [E, \g_{<0}])^{\otimes n}$ for the differential superalgebra $\mathcal{P}(\g^F)$ generated by $\g^F$. Then the subset $\{\omega_i \,|\,i\in J_F\}$ freely generates $\mathcal{W}^k (\g,F)$ as a differential superalgebra.
 Moreover, in \cite{Suh20, Suh16}, the third author computed the summand $\gamma_i^1$ of $\omega_i-v_i$ in  $\mathcal{P}(\g^F) \otimes (\mathbb{C}[\partial]\otimes [E, \g_{<0}])$. Using the explicit formula of $\gamma_i^1$,  the $\lambda$-bracket between the generators can be expressed with elements in $\mathbb{C}[\lambda]\otimes \mathbb{C}[\partial^n \omega_i\,|\,n \in\mathbb{Z}_+, \, i\in I]$.

On the other hand, in various physical articles \cite{Dinar14, MarsRat86,RagSo99, Rag01,dBT93, BFOFW, BTvD, FRRT}, classical W-algebras were introduced as Dirac reductions of affine PVAs.  When $\g$ is a Lie algebra, De Sole-Kac-Valeri described a W-algebra $\mathcal{W}^k(\g, F)$ as a limit of Dirac reduction of $\mathcal{V}^k(\g)$ (see \cite{DSKV16}).
In this paper, we aimed to explain the $\lambda$-bracket of  a W-superalgebra $\mathcal{W}^k(\g,F)$, in terms of a Dirac reduced bracket of $\mathcal{V}^k(\g)$.

Unfortunately, it was not possible to find   $\theta_{I}\subset \mathcal{V}^k(\g, F)$ making $C(\lambda)=\left(\{\theta_j{}_{\lambda}\theta_i\}\right)$ invertible and the quotient algebra $\mathcal{V}^k(\g)/\langle \theta_{I}\rangle$ with the Dirac reduced bracket 
isomorphic to $\mathcal{W}^k(\g, F)$. In case $\g$ is a Lie algebra, De Sole-Kac-Valeri (\cite{DSKV16}) resolved the problem by including an extra parameter $t$ in the constraints $\theta_{I}(t)$. 
In this paper, instead of considering an additional parameter, we define {\it the modified Dirac reduced bracket on the quotient differential superalgebra $\mathcal{P}/\left< \theta_{I} \right>$} associated with a finite subset $\theta_{I}$ of PVSA $\mathcal{P}$. 
To be precise, let 
 $\pi:\mathcal{P}\rightarrow \mathcal{P}/\langle \theta_{I} \rangle$ be the canonical projection map and $\widetilde{C}(\lambda)$ be the matrix whose $ij$-entry is $\widetilde{C}_{ij}= \pi \left({C}_{ij}\right)\in\mathcal{P}/\left< \theta_{I} \right>[\lambda]$ for the $ij$-entry $C_{ij}$ of $C(\lambda)$ in \eqref{super Dirac-intro}. 
 Provided that $\widetilde{C}(\lambda)$ is invertible, we replace $(C^{-1})_{ji}(\lambda)$ in \eqref{super Dirac-intro} by $\widetilde{C}^{-1}_{ji}(\lambda)$ and obtain the modified Dirac reduced bracket.
Consequently, we prove  Theorem \ref{Thm: nonSUSY main} which implies the following statement.
\begin{thm} 
Let  $\theta_{I}= \{q_m^i - (F|q_m^i) | (i,m)\in I\}$, where $ \{q_m^i  | (i,m)\in  I\}$ is a basis of $[E,\g]$ and 
 $\pi: \mathcal{V}^k(\g) \to \mathcal{V}^k(\g)/\left< \theta_{I} \right>$ is the canonical projection map. Then the differential superalgebra $\mathcal{V}^k(\g)/\left< \theta_{I} \right>$ endowed with the {\it modified Dirac reduced bracket} is isomorphic to $\mathcal{W}^k(\g, F)$.
\end{thm}

\vskip 2mm
As a counterpart of W-algebras in the theory of SUSY vertex algebras, one can consider SUSY W-algebras. A SUSY W-algebra was introduced to understand the superfield formalism in super-Toda theory in \cite{DRS92,EH91,KMN91}.  The SUSY version of BRST formalism
 was introduced by Madsen-Ragoucy in \cite{MR94} and it was interpreted via SUSY vertex algebras in \cite{MRS21}.  On the other hand,  the third author showed in \cite{Suh20} that SUSY classical W-algebras can be obtained from the quasi-classical limit of SUSY BRST complexes or the SUSY analogue of Drinfeld-Sokolov reductions.

A SUSY classical W-algebra $\mathcal{W}^k(\bar{\g}, f)$ is governed by a Lie superalgebra $\g$ and an odd nilpotent element $f$ in  a subalgebra $\mathfrak{s} \simeq \mathfrak{osp}(1|2)$. Let  us consider a subspace $\p= \bigoplus_{i\leq 0} \g(i)$ and the differential superalgebra   $\mathscr{P}(\bar{\p})$ generated by $\bar{\p}$.  For $a\in \g$, we denote by $a_-\in\p$ and $a_+ \in \n$ the elements satisfying $a= a_+ + a_-$ and define $\rho_-(a):= a_- +(f|a)$.
Then the SUSY W-algebra 
\[\mathcal{W}^k(\bar{\g}, f)= \{ \bar{\omega}  \in \mathscr{P}(\bar{\p})\  |\  \rho_-(\{ \bar{n} \, {}_\Lambda \,   \bar{\omega}\}) =0 \text{ for any } n\in \mathfrak{n} \}\] 
is a SUSY PVA, whose $\Lambda$-bracket is induced from the affine SUSY PVA $\mathscr{V}^k(\bar{\g})$.
When a homogeneous basis $\{u_i \,|\, i\in J^f\}$ of $\ker(\ad f)=:\g^f$ is given, one can find a unique element  $\bar{\upsilon}_i \in \mathcal{W}^k(\bar{\g}, f)$ for $i\in J_f$ such that $\bar{\upsilon}_i -\bar{u}_i \in  \bigoplus_{n\geq 1} \mathscr{P}(\bar{\g}^f) \otimes (\mathbb{C}[D]\otimes \overline{[e, \g_{\leq -\frac{1}{2}}]})^{\otimes n}$. The $\Lambda$-brackets between the generators $\bar{\upsilon}_i$ can be found in \cite{Suh20}.

Similarly to the classical W-superalgebra case, we substitute $C(\Lambda)$  in \eqref{SUSY Dirac-intro} by $\widetilde{C}(\Lambda)$ whose $ij$-entry is $\pi\left(C_{ij}(\Lambda)\right)\in \mathscr{P}/\left< \theta_{I}\right>[\Lambda]$ to obtain a bracket called {\it modified SUSY Dirac reduction on the quotient space $\mathscr{P}/\left< \theta_{I}\right>$}. Then the following theorem holds. See Theorem \ref{Thm: SUSY main} for the precise statement.

\begin{thm} 
Let  $\theta_{I}= \{q_m^i - (f|q_m^i) | (i,m)\in I\}$, where $ \{q_m^i  | (i,m)\in  I\}$ is a basis of $[e,\g]$ and 
 $\pi: \mathscr{V}^k(\bar{\g}) \to \mathscr{V}^k(\bar{\g})/\left< \theta_{I} \right>$ be the canonical projection map. Then the differential superalgebra $\mathscr{V}^k(\g)/\left< \theta_{I} \right>$ endowed with the {\it modified Dirac reduced bracket} is isomorphic to the SUSY classical W-algebra $\mathcal{W}^k(\bar{\g}, f)$.
\end{thm}

This paper is organized as follows. In Section \ref{Sec:DR for PVA}, we review non-local PVAs and introduce Dirac reduction for Poisson vertex superalgebras and Poisson superalgebras.  Analogously, in Section  \ref{Sec: Dirac-SUSY}, we introduce non-local SUSY PVA and Dirac reduction for SUSY PVAs. 
In Section \ref{Sec:Structure of W-superalgebras}, we introduce modified Dirac reduced brackets for Poisson vertex superalgebras and describe the Poisson structures of W-superalgebras via modified Dirac reduced brackets of affine PVSAs.
In Section \ref{Sec:Structure of SUSY W-algebra}, we explain SUSY analogue of Section \ref{Sec:Structure of W-superalgebras} and SUSY W-algebras via modified Dirac reduced brackets of SUSY affine PVAs.

Throughout this paper, the base field is $\mathbb{C}$. The set of integers (resp. nonnegative integers) is denoted by $\ZZ$ (resp. $\ZZ_{+}$).

\section{Dirac reduction of Poisson (vertex) superalgebras} \label{Sec:DR for PVA}

 A $\ZZ/2\ZZ$-graded vector space $V=V_{\bar 0}\oplus V_{\bar 1}$ is called a \textit{vector superspace} and its homogeneous element $a \in V_{\bar \imath}$ is called \textit{even} (resp. \textit{odd}) if $i =0$ (resp. $i =1$). Denote the {\it parity} of $a\in V_{\bar \imath}$ by $\tilde{a}:=i$. A linear operator $\phi : V \rightarrow V$ is called \textit{even} if $\phi(V_{\bar \imath})\subset V_{\bar{\imath}}$ and \textit{odd} if $\phi(V_{\bar \imath})\subset V_{\bar{\imath}+\bar{1}}$. A $\CC$-algebra $A$ is called a \textit{superalgebra} if $A$ is a vector superspace satisfying $A_{\bar \imath}A_{\bar \jmath} \subset A_{\bar{\imath}+\bar{\jmath}}$ for $i,j \in \{0, 1\}$.  If there is a linear operator $d$ on a superalgebra $A$ with satisfying 
 \[d(ab)=d(a) b+(-1)^{\tilde{d}\tilde{a}}a d(b)\]
 for $a,b \in A$, then $A$ is called a \textit{differential superalgebra} with a \textit{derivation} $d$. Here and further, for simplicity, we just call the tuple $(A, d)$ a \textit{differential algebra}. In the rest of the paper, whenever the parity $\tilde{a}$  of an element $a$ in a vector superspace is considered, we assume that $a$ is homogeneous even though it is not mentioned.

 
 \subsection{Non-local Poisson vertex superalgebras}\label{Subsec:Non-local PVSA}\hfill

  A vector superspace $\mathfrak{g}$ endowed with a bilinear bracket $[\cdot , \cdot ]: \g \times \g \to \g$ is a \textit{Lie superalgebra} if it satisfies 
  \begin{itemize}
    \item (skewsymmetry) $[a,b]=-(-1)^{\tilde{a}\tilde{b}}[b,a]$,
    \item (Jacobi identity) $(-1)^{\tilde{c}\tilde{a}}[a,[b,c]]+(-1)^{\tilde{a}\tilde{b}}[b,[c,a]]+(-1)^{\tilde{b}\tilde{c}}[c,[a,b]]=0$
  \end{itemize}
  for $a, b,c \in \mathfrak{g}$. If a Lie superalgebra $(P, \{ \cdot , \cdot \})$ is a unital supercommutative associative algebra with the Leibniz rule:
\begin{equation}
  \left\{ a , b c \right\}=\left\{ a , b \right\} c+(-1)^{\tilde{a}\tilde{b}}b \left\{ a , c \right\}
\end{equation}
for $a,b,c \in P$, then $P$ is called a {\it Poisson superalgebra.}

To define a non-local Poisson vertex superalgebra in a similar manner, let us introduce an admissible non-local $\lambda$-bracket. See \cite{DSK13} for detailed properties of non-local Poisson vertex algebras.
On a $\CC[\partial]$-module $\mathcal{R}$, a \textit{non-local $\lambda$-bracket} is a parity preserving bilinear map $[ \cdot\,  {}_\lambda \,\cdot ]: \mathcal{R} \times \mathcal{R} \to \mathcal{R}(\!(\lambda^{-1})\!)$ satisfying the \textit{sesquilinearity:}
\begin{equation} \label{eq:sesqui}
 [\partial a _\lambda b]=-\lambda[ a _\lambda b], \ [a _\lambda \partial b]=(\lambda+\partial)[ a _\lambda b]
\end{equation}
for $a, b \in \mathcal{R}$. We denote by $[a_{\lambda} b] =\sum_{n\in \ZZ }\lambda^{n} a_{(n)}b  $ for $a_{(n)}b\in \mathcal{R}$.

Consider a superspace 
\begin{equation}\label{eq:admissibility, space}
    \mathcal{R}_{\lambda,\mu}:=\mathcal{R}[\![\lambda^{-1}, \mu^{-1}, (\lambda+\mu)^{-1}]\!][\lambda,\mu]
\end{equation}
and let $\iota_{\mu,\lambda}: \mathcal{R}_{\lambda, \mu} \hookrightarrow \mathcal{R}(\!(\lambda^{-1})\!)(\!(\mu^{-1})\!)$ be defined by the geometric expansion of $(\lambda+\mu)^m$ for an integer $m$ in the domain  $|\mu|>|\lambda|$. 
Then we can identify $\mathcal{R}_{\lambda,\mu}$ with the image $\iota_{\mu,\lambda}(\mathcal{R}_{\lambda,\mu}) \subset \mathcal{R}(\!(\lambda^{-1})\!)(\!(\mu^{-1})\!)$.
 If a non-local $\lambda$-bracket on $\mathcal{R}$ satisfies 
 \begin{equation*}
   [ a {}_{\lambda} [ b {}_{\mu} c ] ] \in \mathcal{R}_{\lambda,\mu}
 \end{equation*}
 for all $a, b, c \in \mathcal{R}$, then it is said to be \textit{admissible}.
If a $\CC[\partial]$-module $\mathcal{R}$ and an admissible $\lambda$-bracket $[ \cdot\, {}_{\lambda} \,\cdot ]$ are given to satisfy the following axioms:
  \begin{itemize}
    \item (skewsymmetry) $[b _\lambda a]=-(-1)^{\tilde{a}\tilde{b}}{}_{\leftarrow}[a \,_{-\lambda-\partial}\, b]$ in $\mathcal{R}(\!(\lambda^{-1})\!),$  
    \item (Jacobi identity) $[ a _\lambda [ b _\mu c]]=[[ a _\lambda b] \,_{\lambda+\mu}\, c]+(-1)^{\tilde{a}\tilde{b}}[b _\mu [ a _\lambda c]]$ in $\mathcal{R}_{\lambda,\mu}$
  \end{itemize}
for  $a, b, c \in \mathcal{R}$. We call $\mathcal{R}$ a \textit{non-local Lie conformal superalgebra} (\textit{LCA}). The skewsymmetry can be rewritten as
\[ \sum_{n\in \ZZ} \lambda^{n}b_{(n)}a =-(-1)^{\tilde{a}\tilde{b}}\sum_{n\in \ZZ}(-\lambda-\partial)^{n}a_{(n)}b =-(-1)^{\tilde{a}\tilde{b}}\sum_{n\in \ZZ}\sum_{r\in \ZZ_+}\binom{n}{r}(-1)^{n}(\partial^{r}a_{(n)}b)\lambda^{n-r}\]
and the three terms in the Jacobi identity are in $\mathcal{R}(\!(\lambda^{-1})\!)(\!(\mu^{-1})\!)$,   $\mathcal{R}(\!((\lambda+\mu)^{-1})\!)(\!(\lambda^{-1})\!)$ and $\mathcal{R}(\!(\mu^{-1})\!)(\!(\lambda^{-1})\!)$, respectively. By the admissibility, the LHS and RHS of Jacobi identity can be compared in $\mathcal{R}_{\lambda,\mu}$ via the identification $\iota_{\mu,\lambda}$, $\iota_{\lambda,\lambda+\mu}$ and  $\iota_{\mu,\lambda}$.

\begin{defn}\label{Defn:non-local PVA}
  A \textit{non-local Poisson vertex superalgebra} (\textit{non-local PVSA} or just \textit{PVSA}) is a quadruple $(\mathcal{P},\partial,\left\{ \cdot \,_\lambda\, \cdot \right\},\cdot)$ which satisfies the following axioms:
  \begin{itemize}
    \item $(\mathcal{P}, \partial, \left\{ \cdot \,_\lambda\, \cdot \right\})$ is a non-local LCA,
    \item $(\mathcal{P}, \partial, \cdot)$ is a unital supercommutative associative differential algebra,
    \item (left Leibniz rule) $\left\{ a _\lambda bc \right\}=\left\{ a _\lambda b \right\}c+(-1)^{\tilde{a}\tilde{b}}b\left\{ a _\lambda c \right\}$ for $a, b, c \in \mathcal{P}.$
  \end{itemize}
 \end{defn}

The supersymmetric algebra generated by a superspace $V$ is 
$$S(V):=S(V_{\bar 0})\otimes \bigwedge(V_{\bar 1}),$$
where $S(V_{\bar 0})$ is the symmetric algebra generated by $V_{\bar 0}$ and $\bigwedge(V_{\bar 1})$ is the exterior algebra generated by $V_{\bar{1}}$. 
Suppose $V$ is a superspace with a homogeneous basis $\mathcal{B}=\mathcal{B}_{0} \sqcup\mathcal{B}_{1}$, where $\mathcal{B}_{0}=\{u_{i} | i \in I_0\}$ and $\mathcal{B}_{1}=\{u_{i} | i \in I_1\}$ consist of even and odd elements,  respectively. Then the supersymmetric algebra
\begin{equation}\label{2.3}
  \mathcal{P}=S(\CC[\partial]\otimes V)=\CC[u_{i}^{(m)}|i \in I_{0} \sqcup I_{1}, \; m \in \ZZ_{+}]
\end{equation}
is a differential algebra with even derivation $\partial$ on $\mathcal{P}$ defined by $\partial(u_i^{(m)})= u_i^{(m+1)}$.



\begin{prop}(Master formula of non-local PVSA).\label{Master formula}
   For the differential algebra $\mathcal{P}$ in \eqref{2.3}, let $ \frac{\partial}{\partial u_{i}^{(m)}}$ be the derivation of parity $\tilde{u}_i$ on  $\mathcal{P}$  such that  $\frac{\partial}{\partial u_{i}^{(m)}}u_{j}^{(n)}=\delta_{m,n}\delta_{i,j}.$
If  $\mathcal{P}$ is a PVSA with non-local $\lambda$-bracket $\left\{ \cdot \,_\lambda\, \cdot \right\}$, then for $f,g\in\mathcal{P}$,
  \begin{equation}\label{2.4}
    \left\{ f {}_{\lambda} g \right\}=\displaystyle\sum_{\substack{i,j \in I_{0}\sqcup I_{1} \\ m,n \in \mathbb{Z}_{+}}}X_{i,j}^{f,g}\frac{\partial g}{\partial u_{j}^{(n)}}(\lambda+\partial)^{n}\left\{ u_{i} {}_{\lambda+\partial} u_{j} \right\}_{\rightarrow}(-\lambda-\partial)^{m}\frac{\partial f}{\partial u_{i}^{(m)}},
  \end{equation}
where $X_{i,j}^{f,g}=(-1)^{\tilde{f}\tilde{g}}(-1)^{\tilde{u}_{i}\tilde{u}_{j}}(-1)^{\tilde{g}\tilde{u}_{j}}(-1)^{\tilde{u}_{j}}$ and $\left\{ u_{i} {}_{\lambda+\partial} u_{j} \right\}_{\to}=\sum_{l\in \ZZ} (u_{i\, (l)} u_j) (\lambda+\partial)^l$. 
\end{prop}
\begin{proof}

  By the skewsymmetry and the left Leibniz rules of PVSAs, one can deduce the \textit{right Leibniz rule}:
 \begin{equation*}
   \left\{ ab {}_{\lambda} c \right\}=(-1)^{\tilde{b}\tilde{c}}\left\{ a {}_{\lambda+\partial} c \right\}_{\rightarrow}b+(-1)^{\tilde{a}(\tilde{b}+\tilde{c})}\left\{ b {}_{\lambda+\partial} c \right\}_{\rightarrow}a,
 \end{equation*}
where 
 $ \left\{ a {}_{\lambda+\partial} c \right\}_{\rightarrow}b=\sum_{n\in \ZZ}a_{(n)}c\,(\lambda+\partial)^{n}b.$ The formula \eqref{2.4} follows from the sesquilinearity and left and right Leibniz rule of $\lambda$-brackets. We refer to Theorem 4.8 in \cite{DSK13} and Proposition 4.4 in \cite{Suh18} for detailed proof.
\end{proof}

\begin{prop} \label{Prop:PVA from generators}
Let $V$ be a vector superspace and $\mathcal{P}:=S(\CC[\partial]\otimes V)$. If a bracket $[ \cdot {}_{\lambda} \cdot ] : V \otimes V \rightarrow\mathcal{P} (\!(\lambda^{-1})\!)$ satisfies the admissibility, skewsymmetry and Jacobi identity, then it can be uniquely extended to a non-local PVSA $\lambda$-bracket on $\mathcal{P}$ using the sesquilinearity and Leibniz rule.
\end{prop}
\begin{proof}
 It can be proved by combining the results of Theorem 4.8 of \cite{DSK13} and Theorem 1.15 of \cite{BDSK09}.
 \end{proof}

\begin{ex} \label{ex: affine PVSA}
  Let $\mathfrak{g}$ be a Lie superalgebra with a nondegenerate supersymmetric invariant even bilinear form $(\cdot|\cdot)$. For $k\in \CC$, the differential algebra $\mathcal{V}^k(\g):=S(\CC[\partial]\otimes\mathfrak{g})$ with the $\lambda$-bracket
  \begin{equation}
    \left\{ a {}_{\lambda} b \right\}=[a,b]+k\lambda(a|b) \;\text{for}\; a,b \in \mathfrak{g}
  \end{equation}
is called an affine PVSA. This $\lambda$-bracket satisfies all the conditions of Proposition \ref{Prop:PVA from generators}. Hence the PVSA structure of $\mathcal{V}^k(\g)$ is fully determined by the master formula \eqref{2.4}.
\end{ex}

 \begin{ex}\label{ex:non-local PVA}
  The differential algebra $\mathcal{P}=\CC[u^{(n)}|n \in \ZZ_{+}]$ endowed with the $\lambda$-bracket
  \begin{equation*}
    \left\{ u {}_{\lambda} u \right\}=\lambda^{-1}
    \end{equation*}  
   is a non-local PVSA. By Proposition \ref{Prop:PVA from generators}, the $\lambda$-bracket can be extended to the whole space $\mathcal{P}$. See \cite{DSK13} for more examples.
\end{ex}

\subsection{Dirac reduction of Poisson vertex superalgebras} \label{Subsec:Dirac PVA}
 \hfill

Let $\mathcal{P}$ be a PVSA and denote by $\text{Mat}_{(r|s)}=\text{End}(\CC_{(r|s)})$ for nonnegative integers $r$ and $s$.
Consider the superspace  $\mathcal{M}_{(r|s)}(\lambda):= \text{Mat}_{(r|s)} \otimes \mathcal{P}(\!(\lambda^{-1})\!)$, where the parity of $a\otimes F(\lambda)\in \mathcal{M}_{(r|s)}(\lambda)$ is given by $\tilde{a}+\tilde{F}$. It is an associative algebra with respect to the product 
\begin{equation}\label{eq:matrix product PVA}
  \begin{aligned}
   \left( a \otimes F(\lambda)\right)  \circ \left( b \otimes G(\lambda)\right)&=\left(a\otimes F(\lambda+\partial)\right)\left(b\otimes G(\lambda)\right)\\
   &=(-1)^{\tilde{b}\tilde{F}} ab \otimes F(\lambda+\partial) G(\lambda).
  \end{aligned}
\end{equation}
Then $\text{Id}_{(r|s)} \otimes 1$ is the unity in $\mathcal{M}_{(r|s)}(\lambda)$ for the identity matrix $\text{Id}_{(r|s)}$ in $ \text{Mat}_{(r|s)}$. 
 In addition, we say $A(\lambda) \in \mathcal{M}_{(r|s)}(\lambda)$ is  \textit{invertible} if there is an element $A^{-1}(\lambda)\in \mathcal{M}_{(r|s)}(\lambda)$ such that 
\[A(\lambda) \circ A^{-1}(\lambda)= A^{-1}(\lambda) \circ A(\lambda)= \text{Id}_{(r|s)}\otimes 1.\]
For the simplicity of notation, we sometimes denote $\mathcal{M}_{(r|s)}(\lambda)$ by $\mathcal{M}(\lambda)$. 

Let $I=I_{0} \sqcup I_{1} \subset \ZZ$ be a finite index set, where $I_{i}=I \cap (2\ZZ+ i)$ and $i \in \{0,1\}$. Let $\theta_{I}:=\{\, \theta_i \, |\, i\in I\}$ be a subset of $\mathcal{P}$ consisting of homogeneous elements such that $(-1)^{\tilde{\theta}_i} =(-1)^i$. For $r= |I_0|$ and  $s= |I_1|$, consider the element  
\begin{equation} \label{eq: DR matrix for PVA} 
C(\lambda):= \sum_{i,j\in I} e_{ij} \otimes  \{ \theta_{j}\, {}_{\lambda}\, \theta_{i} \}\in \mathcal{M}_{(r|s)}(\lambda).
\end{equation}

\begin{defn} \label{Def:Dirac modified super}
Assume that the element $C(\lambda)$ in \eqref{eq: DR matrix for PVA} is invertible and write its inverse as $(C^{-1})(\lambda)=\sum_{i,j}e_{ij}\otimes (C^{-1})_{ij}(\lambda)$. The {\it Dirac reduced bracket} of the Poisson $\lambda$-bracket $\left\{ \cdot \,_\lambda\, \cdot \right\}$ on $\mathcal{P}$ associated with $\theta_{I}$ is the bilinear map
\[ \left\{ \cdot \,_\lambda\, \cdot \right\}^{D} : \mathcal{P} \times \mathcal{P} \rightarrow \mathcal{P}(\!(\lambda^{-1})\!) \] given by
  \begin{equation}\label{2.8}
  \left\{ a \,_{\lambda}\, b \right\}^{D}=\left\{ a \,_{\lambda}\, b \right\}-\displaystyle\sum_{i,j \in I} (-1)^{(\tilde{a}+j)(\tilde{b}+i)+i+j}\left\{ \theta_{j} \,_{\lambda+\partial}\, b \right\}_{\rightarrow}(C^{-1})_{ji}(\lambda+\partial)\left\{ a \,_{\lambda}\, \theta_{i} \right\},
\end{equation}
where $a,b$ are  in $ \mathcal{P}$. 
\end{defn}

For the next proposition, we generalize the multiplication \eqref{eq:matrix product PVA}  as follows:
\begin{align*}
\left(\text{Mat}_{(p_1|q_1)(r_1|s_1)} \otimes  \mathcal{P}(\!(\lambda^{-1})\!)\right) \times \left(\text{Mat}_{(p_2|q_2)(r_2|s_2)}\otimes   \mathcal{P}(\!(\lambda^{-1})\!)\right) &\rightarrow \left(\text{Mat}_{(p_1|q_1)(r_2|s_2)} \otimes   \mathcal{P}(\!(\lambda^{-1})\!)\right)\\
\Big(\left( a \otimes F(\lambda)\right) , \left( b \otimes G(\lambda)\right)\Big)&\mapsto (-1)^{\tilde{b}\tilde{F}} ab \otimes F(\lambda+\partial) G(\lambda).
\end{align*}
Here $p_1,q_1,r_1=p_2,s_1=q_2, r_2, s_2$ are nonnegative integers and $\text{Mat}_{(p_1|q_1)(r_1|s_1)}$ is the set of linear maps from $\CC_{(p_1|q_1)}$ to $\CC_{(r_1|s_1)}$.
\begin{prop} \label{prop:Trivial example, PVA}
Let $\mathcal{P}= \CC[\theta^{(n)}_k|k\in I, n \in \ZZ_{+}]$ be a PVSA with the derivation $\partial : \theta_{k}^{(n)} \mapsto \theta_{k}^{(n+1)}$. If $C(\lambda)$ in \eqref{eq: DR matrix for PVA} is invertible, the corresponding Dirac reduced bracket is trivial.
\end{prop}

\begin{proof}

Let us use the identification 
$ \mathcal{P}(\!(\lambda^{-1})\!)\simeq \text{Mat}_{(1|0)} \otimes  \mathcal{P}(\!(\lambda^{-1})\!)$. 
For $a,b\in \mathcal{P}$, recall the master formula \eqref{2.4}:
\begin{equation}\label{eq:master}
\begin{aligned}
&     \left\{ a {}_{\lambda} b \right\}
 =\displaystyle\sum_{i,j \in I, \, n\in \ZZ_+}(-1)^{\tilde{a}\tilde{b}+\tilde{b}j+ij+j}\frac{\partial b}{\partial\theta_{j}^{(n)}}(\lambda+\partial)^n\left\{ \theta_i {}_{\lambda+\partial} \theta_j \right\}_{\to}(-\lambda-\partial)^m\frac{\partial a}{\partial \theta_{i}^{(m)}}.
\end{aligned}
\end{equation}
For given $i,j\in I$, the RHS of \eqref{eq:master} can be written as 
\begin{equation} \label{eq:master-ij}
\begin{aligned}
  \displaystyle\sum_{n\in \ZZ_+}(-1)^{\tilde{a}\tilde{b}+i}\left(e_{j}^{T} \otimes\frac{\partial b}{\partial\theta_{j}^{(n)}}(\lambda+\partial)^{n}\right) \big(e_{ji} \otimes \left(C_{ji}(\lambda+\partial)\right)_{\to} \big)\left(e_{i} \otimes (-\lambda-\partial)^{m}\frac{\partial a}{\partial \theta_i^{(m)}}\right),
    \end{aligned}
    \end{equation}
where $e_i\in \text{Mat}_{(r|s)(1|0)}$ is the $i$-th standard column vector, $e_i^T\in \text{Mat}_{(1|0)(r|s)}$ is its transpose and $e_{ij}$ is the $ij$-th elementary matrix in $\text{Mat}_{(r|s)}.$ Note that 
\begin{equation}\label{eq:trivialDR-1}
 e_{ji}\otimes C_{ji}(\lambda)= \sum_{k,l\in I} (e_{jk}\otimes C_{jk}(\lambda+\partial)) (e_{kl}\otimes (C^{-1})_{kl}(\lambda+\partial)) (e_{li} \otimes C_{li}(\lambda))
\end{equation}
and 
\begin{equation}\label{eq:trivialDR-2}
\begin{aligned}
& \sum_{i\in I, \, n\in \ZZ_+} (-1)^i( e_{li}\otimes C_{li}(\lambda+\partial)) \left( e_i\otimes  (-\lambda-\partial)^{m}\frac{\partial a}{\partial \theta_i^{(m)}} \right)=(-1)^{\tilde{a}l}e_l \otimes \{a{}_\lambda \theta_l\},\\
& \sum_{j\in I\, n\in \ZZ_+}  \left( e_j^T \otimes \frac{\partial b}{\partial \theta_j^{(n)}}(\lambda+\partial)^{n}\right) (e_{jk}\otimes C_{jk}(\lambda+\partial))=e_k^T  \otimes \{\theta_k{}_\lambda b\}.
\end{aligned}
\end{equation}
If we substitute $\big(e_{ji} \otimes C_{ji}(\lambda+\partial) \big)$ in \eqref{eq:master-ij} with the RHS of \eqref{eq:trivialDR-1}, then by \eqref{eq:trivialDR-2} we get 
\begin{equation}\label{eq:trivialDR-3}
\begin{aligned}
\{a{}_\lambda b\}& = \sum_{k,l\in I}(-1)^{\tilde{a}(\tilde{b}+l)}(e_k^T \otimes \{\theta_k{}_{\lambda+\partial} b\}) (e_{kl}\otimes (C^{-1})_{kl}(\lambda+\partial))(e_l \otimes \{a{}_\lambda \theta_l\})
\\
&  =\sum_{k,l\in I} (-1)^{(\tilde{a}+k)(\tilde{b}+l)+k+l}\{\theta_k{}_{\lambda+\partial} b\} (C^{-1})_{kl}(\lambda+\partial) \{a{}_\lambda \theta_l\}.
\end{aligned}
\end{equation}
Hence, $\{a{}_\lambda b\}^D= \{a{}_\lambda b\}- [\text{RHS of }\eqref{eq:trivialDR-3}]=0$.
 \end{proof}

Let us introduce the \textit{adjoint} of  $F(\lambda)= \sum_{n\in \ZZ} F_n \lambda^n \in \mathcal{P}(\!(\lambda^{-1})\!)$ for $F_n\in \mathcal{P}$ which is defined by $F^*(\lambda)=  \sum_{n\in \ZZ}(-\lambda-\partial)^n F_n$. In addition, the {\it adjoint} of an even element $A(\lambda)=\sum_{i,j}e_{ij} \otimes A_{ij}(\lambda)=\sum_{i,j; n \in \ZZ}e_{ij} \otimes A_{ij;n}\lambda^{n}\in\mathcal{M}_{(r|s)}(\lambda)$  is given by
\begin{equation}\label{adjoint}
A^{*}(\lambda):=\displaystyle\sum_{i,j \in I ; \, n \in \ZZ}(-1)^{ij+j} e_{ji}\otimes A_{ij}^{*}(\lambda).
\end{equation}
Then $(A(\lambda) \circ B(\lambda))^*=  B^*(\lambda)\circ A^*(\lambda)$ for any even elements $A(\lambda),  B(\lambda) \in  \mathcal{M}(\lambda)$.

Now we show that Dirac reduced bracket \eqref{2.8} is indeed a PVSA bracket. 
The following lemma is needed to prove Theorem \ref{Theorem:super Dirac, part1}.

\begin{lem}\label{Lemma:Inverse operator}
     Let $\left\{ \cdot \,_\lambda\, \cdot \right\} :  \mathcal{P} \times  \mathcal{P} \rightarrow  \mathcal{P}(\!(\lambda^{-1})\!)$ be a non-local PVSA $\lambda$-bracket on $\mathcal{P}$. Suppose $C(\lambda)=(C_{ij}(\lambda))_{i,j\in I} \in \mathcal{M}(\lambda)$ in \eqref{eq: DR matrix for PVA} is invertible  and let $(C^{-1})(\lambda)=\left((C^{-1})_{ij}(\lambda)\right)_{i,j\in I}\in  \mathcal{M}(\lambda)$ be its inverse. For $a \in  \mathcal{P}$,  assume that 
  \begin{center}
    $\left\{ a \,{}_{\lambda}\, C_{ij}(\mu) \right\} \in  \mathcal{P}_{\lambda,\mu}$ \,for all $i, j\in I$.
  \end{center}
  Then we have $\left\{ a \,{}_{\lambda}\, (C^{-1})_{ij}(\mu) \right\}$, $\left\{ (C^{-1})_{ij}(\lambda) \,{}_{\lambda+\mu}\, a \right\} \in  \mathcal{P}_{\lambda,\mu}$. If $a\in \mathcal{P}$ is homogenous, then
   \begin{equation} \label{eq:super_lemma_first equality} 
  \begin{aligned}
  &\{a \,{}_{\lambda}\, (C^{-1})_{ij}(\mu)\}\\
  &=-\displaystyle\sum_{r,t\in I}\displaystyle^{\tilde{a}(i+r)+(i+t)(j+r)+r+t}(C^{-1})_{ir}(\lambda+\mu+\partial)\left\{ a \,{}_{\lambda}\, C_{rt;n} \right\}(\mu+\partial)^n(C^{-1})_{tj}(\mu)
  \end{aligned}
  \end{equation}
  and
  \begin{equation}\label{eq:super_lemma_second equality} 
  \begin{aligned}
    &\{(C_{ij}^{-1})(\lambda) \,{}_{\lambda+\mu}\, a\}\\
    &=-\displaystyle\sum_{r,t\in I}\displaystyle(-1)^{\tilde{a}(i+j+r+t)+(r+t)(j+t)}\left\{ {C_{rt}} \,{}_{\lambda+\mu+\partial}\, a \right\}_{\rightarrow}(\lambda+\partial)^n(C^{-1})_{tj}(\lambda)
   (C^{-1})^{*}_{ir}(\mu),
    \end{aligned}
    \end{equation}
    where $(C^{-1})_{ir}^{*}(\mu):=\left((C^{-1})_{ir}\right)^{*}(\mu).$ \\
   \end{lem}
\begin{proof}
For the first assertion, see Lemma 2.1, Lemma 2.3 and Lemma 2.4 in \cite{DSK13}. 

In order to show \eqref{eq:super_lemma_first equality} and \eqref{eq:super_lemma_second equality}, observe that
\begin{align*}
  e_{ij} \otimes (C^{-1})_{ij}(\lambda)&=\displaystyle\sum_{r,t\in I}(e_{ir} \otimes (C^{-1})_{ir}(\lambda))\circ(e_{rt} \otimes C_{rt}(\lambda))\circ(e_{tj} \otimes (C^{-1})_{tj}(\lambda))\\
  &=\displaystyle\sum_{r,t\in I}(-1)^{(i+r)(r+t)+(t+j)(i+t)} \left(e_{ij} \otimes (C^{-1})_{ir}(\lambda+\partial)C_{rt}(\lambda+\partial)(C^{-1})_{tj}(\lambda)\right).
\end{align*} 
Hence we have
\[(C^{-1})_{ij}(\lambda)=\displaystyle\sum_{r,t\in I}(-1)^{(i+t)(j+r)+r+t}(C^{-1})_{ir}(\lambda+\partial)C_{rt}(\lambda+\partial)(C^{-1})_{tj}(\lambda).\]
Using the left Leibniz rule, we obtain
\begin{align}
&\nonumber \left\{ a {}_{\lambda} (C^{-1})_{ij}(\mu) \right\}=\displaystyle\sum_{r,t\in I}(-1)^{tj+ir+rt+ij+r+t}\left\{ a {}_{\lambda}(C^{\!-1})_{ir}(\mu+\partial)C_{rt}(\mu+\partial)(C^{\!-1})_{tj}(\mu) \right\}\\
\label{eq:super_lemma_proof_1}
&\quad=\displaystyle\sum_{t\in I}(-1)^{(i+t)(j+t)}\{ a {}_{\lambda}\displaystyle\sum_{r\in I}(-1)^{(i+r)(r+t)}(C^{\!-1})_{ir}(\mu+\partial)C_{rt}(\mu+\partial)\}(C^{\!-1})_{tj}(\mu)\\
\label{eq:super_lemma_proof_2}
&\quad +\displaystyle\sum_{r\in I}(-1)^{(i+r)(j+r)+\tilde{a}(i+r)}(C^{\!-1})_{ir}(\lambda+\mu+\partial)\{ a {}_{\lambda}\displaystyle\sum_{t\in I}(-1)^{(r+t)(t+j)}C_{rt}(\mu+\partial)(C^{\!-1})_{tj}(\mu) \}\\
&\quad-\displaystyle\sum_{r,t\in I}(-1)^{t(j+1)+r(i+1)+rt+ij+(i+r)\tilde{a}}(C^{\!-1})_{ir}(\lambda+\mu+\partial)\left\{ a {}_{\lambda} C_{rt}(\mu+\partial) \right\}(C^{\!-1})_{tj}(\mu).
\end{align}
Since 
$\displaystyle \sum_{r\in I}(-1)^{(i+r)(r+t)}(C^{\!-1})_{ir}(\mu+\partial)C_{rt}(\mu+\partial)=\delta_{it}$, we have $\eqref{eq:super_lemma_proof_1}=0$. Similarly, $\eqref{eq:super_lemma_proof_2}=0$ and hence we showed \eqref{eq:super_lemma_first equality}.

For \eqref{eq:super_lemma_second equality}, using the right Leibniz rule, one can observe that 
\begin{align}
  &\nonumber \left\{ (C^{-1})_{ij}(\lambda) {}_{\lambda+\mu} a \right\}=\displaystyle\sum_{r,t\in I}(-1)^{tj+ir+rt+ij+r+t}\left\{ (C^{\!-1})_{ir}(\lambda+\partial)C_{rt}(\lambda+\partial)(C^{\!-1})_{tj}(\lambda) {}_{\lambda+\mu} a \right\}\\
\label{eq:super_lemma_proof_3}  &=\displaystyle\sum_{t\in I}(-1)^{(t+j)(t+i+a)}\displaystyle\{ \sum_{r\in I}(-1)^{(i+r)(r+t)}(C^{\!-1})_{ir}(\lambda+\partial)C_{rt}(\lambda+\partial) {}_{\lambda+\mu} a \}_{\rightarrow}(C^{\!-1})_{tj}(\lambda)\\
\label{eq:super_lemma_proof_4}  &+\displaystyle\sum_{r\in I}(-1)^{(i+r)\tilde{a}}\displaystyle \{ \sum_{t\in I} (-1)^{(r+t)(t+j)}C_{rt}(\lambda+\partial)(C^{\!-1})_{tj}(\partial) {}_{\lambda+\mu} a\}_{\rightarrow}(C^{\!-1})^{*}_{ir}(\mu)\\
  &-\displaystyle\sum_{r,t\in I}(-1)^{(r+t)(j+t)+\tilde{a}(t+j+i+r)}\left\{ C_{rt}(\lambda+\partial) {}_{\lambda+\mu+\partial} a \right\}_{\rightarrow}(C^{\!-1})_{tj}(\lambda)(C^{\!-1})^{*}_{ir}(\mu).
\end{align}
Since $\eqref{eq:super_lemma_proof_3}=\eqref{eq:super_lemma_proof_4}=0$,  \eqref{eq:super_lemma_second equality} holds.
\end{proof}

\begin{thm} \label{Theorem:super Dirac, part1}
  Let $ \mathcal{P}$ be a PVSA and let $\theta_I \subset  \mathcal{P}$ be a finite set of homogeneous elements. If $C(\lambda)\in \mathcal{M}(\lambda)$ in \eqref{eq: DR matrix for PVA} is invertible, then the Dirac reduced bracket $\left\{ \cdot \,_{\lambda}\, \cdot \right\}^{D}$ given by \eqref{2.8} is a Poisson $\lambda$-bracket on $\mathcal{P}$.
\end{thm}

\begin{proof}
The sesquilinearity, admissibility and 
Leibniz rule can be proved by direct computations  (see Theorem 2.2(a) of \cite{DSKV14}).

  To show the skewsymmetry, we substitute $\lambda$ by $-\lambda-\partial$ in \eqref{2.8}
 and let $\partial$ act on the left. Then we have 
        \begin{align*}
      &\left\{ a \,_{-\lambda-\partial}\, b \right\}^{D}\\
      &=\left\{ a \,_{-\lambda-\partial}\, b \right\}-\sum_{\alpha,\beta\in I} ({}_{\leftarrow}\left\{ a\,_{-\lambda-\partial}\, \theta_\beta \right\})\big(|_{x=-\lambda-\partial}(C^{-1})_{\alpha\beta}(x)\big)({}_{\leftarrow}\left\{ \theta_\alpha \,_{-\lambda-\partial}\, b \right\})\\
      &=\left\{ a \,_{-\lambda-\partial}\, b \right\}-\displaystyle\sum_{\alpha,\beta\in I}(-1)^{\tilde{a}\beta+\tilde{b}\alpha+\alpha+\beta+\alpha\beta+1}\left\{ \theta_\beta\,_{\lambda+\partial}\, a \right\}_{\rightarrow}(C^{-1})_{\beta\alpha}(\lambda+\partial)\left\{ b\,_{\lambda}\, \theta_\alpha \right\}.
     \end{align*}
Here, for the last equality, we used the fact that $(C^{-1})_{\beta\alpha}(\lambda)=(-1)^{\alpha\beta+\alpha+\beta+1}\big(|_{x=-\lambda-\partial}(C^{-1})_{\alpha\beta}(x)\big)$, which is induced from the identity $C^{-1}(\lambda) \circ C(\lambda)=\left(C(\lambda)\circ (C^{-1})(\lambda)\right)^*$. 
By multiplying $-(-1)^{\tilde{a}\tilde{b}}$ on the both sides, we get the skewsymmetry.

To prove the Jacobi identity, one can expand the following equation by Definition \ref{Def:Dirac modified super} of Dirac reduced bracket.  
\allowdisplaybreaks
 \begin{align}
    \label{2.15}
    &\{ a \,_{\lambda}\, \left\{ b \,_\mu\, c \right\}^{D}\}^{D}=\left\{ a \,_\lambda\, \left\{ b \,_\mu\, c \right\} \right\}\\
    \label{2.16}
    &-  \sum_{\gamma,\delta\in I}(-1)^{(\tilde{b}+\delta)(\tilde{c}+\gamma)+\gamma+\delta}\{ a \,_\lambda\, \left\{ \theta_{\delta} \,_{y}\, c \right\}\}  \left(\vert_{y=\mu+\partial} (C^{-1})_{\delta\gamma}(\mu+\partial)\left\{ b \,_\mu\, \theta_{\gamma} \right\}\right)\\
    \label{2.17}
    &+  \nonumber\sum_{\gamma,\delta,\zeta,\eta \in I}(-1)^{(\tilde{b}+\delta)(\tilde{c}+\gamma)+\tilde{a}(\tilde{c}+\zeta)+(\gamma+\zeta)(\delta+\eta)+\gamma+\delta+\zeta+\eta}\\
        & \hspace{0.7cm} \{ \theta_{\delta} \,_{\lambda+\mu+\partial}\;\, c\}_{\rightarrow} (C^{-1})_{\delta\zeta}(\lambda+\mu+\partial)\left\{ a _{\lambda} \left\{ \theta_{\eta} \;_{y}\, \theta_{\zeta} \right\} \right\}\left( \big\vert_{y=\mu+\partial} (C^{-1})_{\eta\gamma}(\mu+\partial)\left\{ b \,_\mu\, \theta_{\gamma} \right\}\right)\\
    \label{2.18}
    &-  \sum_{\gamma,\delta \in I} (-1)^{  (\tilde{b}+\delta)(\tilde{c}+\gamma)+\tilde{a}(\tilde{c}+\gamma)+\gamma+\delta}\{ \theta_{\delta} \,_{\lambda+\mu+\partial}\;\, c\}_{\rightarrow}(C^{-1})_{\delta\gamma}(\lambda+\mu+\partial)\left\{ a _\lambda \left\{ b \;_{\mu}\, \theta_{\gamma} \right\} \right\}\\
    \label{2.19}
    &-  \sum_{\alpha,\beta \in I} (-1)^{(\tilde{a}+\beta)(\tilde{b}+\tilde{c}+\alpha)+\alpha+\beta}\{ \theta_{\beta} \,_{\lambda+\partial}\;\, \left\{ b \,_\mu\, c \right\}\}_{\rightarrow}(C^{-1})_{\beta\alpha}(\lambda+\partial)\left\{ a _\lambda \theta_{\alpha} \right\}\\
    \label{2.20}
    &\nonumber+  \sum_{\alpha,\beta,\gamma,\delta \in I} (-1)^{(\tilde{b}+\delta)(\tilde{c}+\gamma)+(\tilde{a}+\beta)(\tilde{c}+\alpha+\delta)+\alpha+\beta+\gamma+\delta}\\
        &\hspace{0.7cm} \{ \theta_{\beta} \,_{x} \left\{ \theta_{\delta} \,_{y}\, c \right\}\} \left( \big\vert_{x=\lambda+\partial} (C^{-1})_{\beta\alpha}(\lambda+\partial)\left\{ a \,_\lambda\, \theta_{\alpha} \right\}\right)\left( \big\vert_{y=\mu+\partial} (C^{-1})_{\delta\gamma}(\mu+\partial)\left\{ b \,_\mu\, \theta_{\gamma} \right\}\right)\\
    \label{2.21}
    &-  \nonumber\sum_{\alpha,\beta,\gamma,\delta,\zeta,\eta \in I} (-1)^{(\tilde{b}+\delta)(\tilde{c}+\gamma)+\tilde{a}\tilde{c}+(\tilde{a}+\beta)(\alpha+\gamma)+(\gamma+\zeta)(\delta+\eta)+\beta\zeta+(\gamma+\eta)(\tilde{a}+\beta)+\alpha+\beta+\gamma+\delta+\zeta+\eta}\\
        &\hspace{0.7cm} \nonumber \{ \theta_{\delta} \,_{\lambda+\mu+\partial}\;\, c\}_{\rightarrow}(C^{-1})_{\delta\zeta}(\lambda+\mu+\partial)\left\{ \theta_{\beta} \,_{x} \left\{ \theta_{\eta} \;_{y}\, \theta_{\zeta} \right\} \right\}\\
            &\hspace{0.7cm} \left( \big\vert_{x=\lambda+\partial} (C^{-1})_{\beta\alpha}(\lambda+\partial)\left\{ a \,_{\lambda}\, \theta_{\alpha} \right\}\right)\left( \big\vert_{y=\mu+\partial} (C^{-1})_{\eta\gamma}(\mu+\partial)\left\{ b \,_{\mu}\, \theta_{\gamma} \right\}\right)\\
    \label{2.22}
    &+  \nonumber\sum_{\alpha,\beta,\gamma,\delta \in I} (-1)^{(\tilde{b}+\delta)(\tilde{c}+\gamma)+\tilde{a}(\tilde{c}+\gamma)+(\tilde{a}+\beta)(\tilde{b}+\alpha+\gamma)+\alpha+\beta+\gamma+\delta}\\
        & \hspace{0.7cm}  \{ \theta_{\delta} \,_{\lambda+\mu+\partial} c\}_{\rightarrow}(C^{-1})_{\delta\gamma}(\lambda+\mu+\partial) \left\{ \theta_{\beta} \,_{\lambda+\partial}\, \left\{ b \,_\mu\, \theta_{\gamma} \right\}\right\}_{\rightarrow}(C^{-1})_{\beta\alpha}(\lambda+\partial)\left\{ a \,_{\lambda}\,\theta_{\alpha} \right\}.
    \end{align}
Note that we used Lemma \ref{Lemma:Inverse operator} to obtain \eqref{2.17} and \eqref{2.21}. 

We also get the expansion of $\{b{}_\mu\{a{}_\lambda c\}^D\}^D$ by exchanging the roles of $a$ and $b$ and roles of $\lambda$ and $\mu$ in the above equation. We shall denote the corresponding terms to \eqref{2.15}$\sim$\eqref{2.22} by \eqref{2.15}$'$$\sim$\eqref{2.22}$'$.

\vskip 2mm
    
Write the remaining term in the Jacobi identity with Dirac reduced bracket $\left\{ \ {}_{\lambda} \ \right\}^{D}$ as follows:
\allowdisplaybreaks
    \begin{align}
    \label{2.23}
    &\{ \{a _{\lambda} b\}^{D} \;_{\lambda+\mu} c\}^{D}=\{ \{a _{\lambda} b\} \;_{\lambda+\mu} c\}\\
    \label{2.24}
    &-   \sum_{\alpha,\beta\in I}   (-1)^{(\tilde{a}+\beta)(\alpha+\tilde{b}+\tilde{c})+\alpha+\beta}\left\{ \left\{ \theta_{\beta} \,_{x}\, b \right\} \,_{\lambda+\mu+\partial}\, c \right\}_{\rightarrow}\left( \big\vert_{x=\lambda+\partial} (C^{-1})_{\beta\alpha}(\lambda+\partial)\left\{ a \,_{\lambda}\, \theta_{\alpha} \right\}\right)\\
    \label{2.25}
    &\nonumber-\sum_{\alpha,\beta,\zeta,\eta\in I}     (-1)^{(\tilde{a}+\beta+1)(\alpha+\beta)+\tilde{c}(\tilde{a}+\tilde{b}+\zeta+\eta)+(\zeta+\eta)(\alpha+\eta)+(\tilde{a}+\alpha)(\beta+\zeta)+\beta\zeta+\beta+\zeta}\\
        &\hspace{0.5cm}\left\{ \left\{ \theta_{\eta} \,_{x}\, \theta_{\zeta} \right\} \,_{\lambda+\mu+\partial}\, c \right\}_{\rightarrow}\left( \big\vert_{x=\lambda+\partial}(C^{-1})_{\eta\alpha}(\lambda+\partial)\left\{ a \,_{\lambda}\, \theta_{\alpha} \right\}\right)(C^{-1})_{\zeta\beta}(\mu+\partial)\left\{ \theta_{\beta} \,_{-\mu-\partial}\, b \right\}\\    
    &\nonumber-\sum_{\alpha,\beta\in I}     (-1)^{(\tilde{a}+\beta)(\tilde{b}+\alpha)+(\alpha+\beta)(\tilde{a}+\tilde{c}+\alpha+1)+(\tilde{b}+\beta)(\tilde{a}+\tilde{c}+\beta)}\\
    &\label{2.26}\hspace{0.7cm}\left\{ \left\{ a \,_{\lambda}\, \theta_{\alpha} \right\} \,_{\lambda+\mu+\partial}\, c \right\}_{\rightarrow}(C^{-1})^{*}_{\alpha\beta}(\mu+\partial)\left\{ \theta_{\beta} \,_{-\mu-\partial}\, b \right\}\\
    \label{2.27}
    &-   \sum_{\gamma,\delta\in I}     (-1)^{(\tilde{a}+\tilde{b}+\delta)(\tilde{c}+\gamma)+\gamma+\delta}\left\{ \theta_{\delta} \,_{\lambda+\mu+\partial}\, c \right\}_{\rightarrow}(C^{-1})_{\delta\gamma}(\lambda+\mu+\partial)\{ \left\{ a \,_\lambda\, b \right\} _{\lambda+\mu}\, \theta_{\gamma} \}\\
    &\nonumber+   \sum_{\alpha,\beta,\gamma,\delta\in I}     (-1)^{(\tilde{a}+\tilde{b}+\delta)(\tilde{c}+\gamma)+(\tilde{a}+\beta)(\tilde{b}+\alpha+\gamma)+\alpha+\beta+\gamma+\delta}\\
    &\label{2.28}\left\{ \theta_{\delta} \,_{\lambda+\mu+\partial}\, c \right\}(C^{-1})_{\delta\gamma}(\lambda+\mu+\partial)\left\{ \left\{ \theta_{\beta} \,_{x} b \right\} \,_{\lambda+\mu+\partial}\, \theta_{\gamma} \right\}_{\rightarrow}\left( \big\vert_{x=\lambda+\partial}(C^{-1})_{\beta\alpha}(\lambda+\partial)\left\{ a \,_{\lambda}\, \theta_{\alpha} \right\}\right)\\    
    &+\nonumber\sum_{\alpha,\beta,\gamma,\delta,\zeta,\eta\in I} (-1)^{(\tilde{a}+\tilde{b}+\delta)(\tilde{c}+\gamma)+(\tilde{a}+\beta)(\tilde{b}+\alpha)+\gamma(\tilde{a}+\alpha)+(\tilde{b}+\beta)(\tilde{a}+\beta+\gamma)}\\
    &\hspace{0.7cm}\nonumber
    (-1)^{\gamma(\alpha+\beta+\zeta+\eta)+(\alpha+\eta)(\zeta+\eta)+(\tilde{a}+\alpha)(\beta+\zeta)+\beta\zeta+\alpha+\gamma+\delta+\zeta} \left\{ \theta_{\delta} \,{}_{\lambda+\mu+\partial}\, c \right\} (C^{-1})_{\delta\gamma}(\lambda+\mu+\partial)\\
        &\label{2.29}\hspace{0.7cm}\left\{ \left\{ \theta_{\eta} \,_{x} \theta_{\zeta} \right\} \,_{\lambda+\mu+\partial}\, \theta_{\gamma} \right\}_{\rightarrow}\left( \big\vert_{x=\lambda+\partial}(C^{-1})_{\eta\alpha}(\lambda+\partial)\left\{ a \,_{\lambda}\, \theta_{\alpha} \right\}\right)(C^{-1})_{\zeta\beta}(\mu+\partial)\left\{ \theta_{\beta} \,_{-\mu-\partial}\, b \right\}\\  
    &+\nonumber\sum_{\alpha,\beta,\gamma,\delta \in I}     (-1)^{(\tilde{a}+\tilde{b}+\delta)(\tilde{c}+\gamma)+(\tilde{a}+\beta)(\tilde{b}+\alpha)+(\alpha+\beta)(\tilde{a}+\alpha+\gamma)+(\tilde{b}+\beta)(\tilde{a}+\beta+\gamma)+\alpha+\beta+\gamma+\delta}\\
        &\label{2.30}\hspace{0.7cm}\left\{ \theta_{\delta} \,_{\lambda+\mu+\partial}\, c \right\}\!(C^{-1})_{\delta\gamma}(\lambda+\mu+\partial)\left\{ \left\{ a \,_{\lambda} \theta_{\alpha} \right\} \,_{\lambda+\mu+\partial}\, \theta_{\gamma} \right\}_{\rightarrow}(C^{-1})^{*}_{\beta\alpha}(\mu+\partial)\left\{ \theta_{\beta} \,_{-\mu-\partial}\, b \right\}.
    \end{align}
 We claim that the sum of three terms in the following triples are all trivial.
    \begin{equation*}
    \begin{aligned}
	&(\eqref{2.16},\eqref{2.19}',\eqref{2.26}),  (\eqref{2.17},\eqref{2.22}',\eqref{2.30}), (\eqref{2.18},\eqref{2.18}',\eqref{2.27}),(\eqref{2.19},\eqref{2.16}',\eqref{2.24}), \\
        &(\eqref{2.20},\eqref{2.20}',\eqref{2.25}), (\eqref{2.21},\eqref{2.21}',\eqref{2.29}), (\eqref{2.22},\eqref{2.17}',\eqref{2.28}).
    \end{aligned}
    \end{equation*}
    
Consider the first triple (\eqref{2.16},\eqref{2.19}$'$,\eqref{2.26}) . We have 
    \begin{align*}
        \eqref{2.26}&=-\displaystyle\sum_{\alpha,\beta  \in I}  (-1)^{\tilde{c}\alpha+\tilde{b}\tilde{c}}(-1)^{\alpha\beta+\alpha+\beta+1}(-1)^{\tilde{b}\beta+1}\left\{ \left\{ a \,_{\lambda}\, \theta_{\alpha} \right\} \,_{\lambda+\mu+\partial}\, c \right\}_{\rightarrow}(C^{-1})_{\alpha\beta}(\mu+\partial)\left\{ b \,_{\mu}\, \theta_{\beta} \right\}.
    \end{align*}
 Here and further, we use the identity $(C^{-1})^{*}_{\beta\alpha}(\mu)=-(-1)^{\alpha\beta+\alpha+\beta}(C^{-1})_{\alpha\beta}(\mu)$. We get following summation by substituting $\alpha$ and $\beta$ in  \eqref{2.19}$'$ with $\gamma$ and $\delta$ respectively and also by substituting $\alpha$ and $\beta$ in \eqref{2.26} with $\delta$ and $\gamma$ respectively.
        \begin{align*}
            &\eqref{2.16}+ \eqref{2.19}'+\eqref{2.26}=-\sum_{\gamma,\delta  \in I}  (-1)^{\tilde{b}\tilde{c}+\tilde{b}\gamma+\tilde{c}\delta+\gamma+\delta+\gamma\delta}\\
            & (\{\{ a {}_{\lambda} \{ \theta_{\delta} {}_{\mu+\partial} c \}\}-(-1)^{\tilde{a}\delta}\{ \theta_{\delta} {}_{\mu+\partial}\{ a{}_{\lambda} c\} \}-\{\{ a {}_{\lambda} \theta_{\delta}\} {}_{\lambda+\mu+\partial} c \})(C^{-1})_{\delta\gamma}(\mu+\partial)\{ b \,_{\mu} {}\theta_{\gamma} \}=0.
        \end{align*}
        
To show the similar result for the second triple  (\eqref{2.17}, \eqref{2.22}$'$, \eqref{2.30}), consider the definition of the adjoint and the skewsymmetry. Then we have
        \begin{align*}
            \eqref{2.30}&=\displaystyle\sum_{\alpha,\beta,\gamma,\delta  \in I}  (-1)^{\tilde{a}\tilde{c}+\tilde{b}\tilde{c}+\tilde{a}\gamma+\tilde{c}\delta+\gamma+\delta+\gamma\delta+\alpha\gamma}(-1)^{\alpha\beta+\alpha+\beta+\tilde{b}\beta}\left\{ \theta_{\delta} \,_{\lambda+\mu+\partial}\, c \right\}_{\rightarrow}(C^{-1})_{\delta\gamma}(\lambda+\mu+\partial)\\
            &\hspace{0.5cm}\nonumber\left\{ \left\{ a \,_{\lambda} \theta_{\alpha} \right\} \,_{\lambda+\mu+\partial}\, \theta_{\gamma} \right\}_{\rightarrow}(C^{-1})_{\alpha\beta}(\mu+\partial)\left\{ b \,_{\mu}\, \theta_{\beta} \right\}.
        \end{align*}
By changing the indices, we can show that the sum $\eqref{2.17}+ \eqref{2.22}'+\eqref{2.30}$ is 
            \begin{align*}
               & \sum_{\gamma,\delta,\zeta,\eta  \in I}  (-1)^{\tilde{a}\tilde{c}+\tilde{b}\tilde{c}+\tilde{a}\zeta+\tilde{c}\delta+\zeta+\delta+\delta\zeta+\eta\zeta+\gamma\eta+\eta+\gamma+\tilde{b}\gamma}\left\{ \theta_{\delta} \,_{\lambda+\mu+\partial}\, c \right\}(C^{-1})_{\delta\zeta}(\mu+\partial)\\
              &  \left(\left\{ a {}_{\lambda} \left\{ \theta_{\eta} {}_{\mu+\partial}\, \theta_{\zeta} \right\} \right\}-(-1)^{\tilde{a}\eta}\left\{ \theta_{\eta} {}_{\mu+\partial} \left\{ a{}_{\lambda} \theta_{\zeta}\right\} \right\}-\left\{ \left\{ a {}_{\lambda} \theta_{\eta} \right\} {}_{\lambda+\mu+\partial}\, \theta_{\zeta} \right\} \right)(C^{-1})_{\eta\gamma}(\mu+\partial)\left\{ b \,_{\mu}\, \theta_{\gamma} \right\}=0.
            \end{align*}

\vskip 2mm            
           
Next the sum $\eqref{2.18}+\eqref{2.18}'+ \eqref{2.27}$ of the third triple is
       \begin{align*}
            &-\displaystyle\sum_{\gamma,\delta  \in I}  (-1)^{\tilde{a}\tilde{c}+\tilde{b}\tilde{c}+\tilde{a}\gamma+\tilde{b}\gamma+\tilde{c}\delta+\gamma+\delta+\gamma\delta}\left\{ \theta_{\delta} \,_{\lambda+\mu+\partial}\, c \right\}_{\rightarrow}(C^{-1})_{\delta\gamma}(\lambda+\mu+\partial)\\
            &\hspace{0.5cm}(\left\{ a \,_{\lambda} \left\{ b \;_{\mu}\, \theta_\gamma \right\} \right\}-(-1)^{\tilde{a}\tilde{b}}\left\{ b \,_{\mu} \left\{ a\,_{\lambda}\, \theta_\gamma \right\} \right\}-\left\{ \left\{ a \,_{\lambda}\, b \right\} \,_{\lambda+\mu}\, \theta_\gamma \right\})=0.
        \end{align*}
\vskip 2mm    
The sum of fourth triple (\eqref{2.19}, \eqref{2.16}$'$, \eqref{2.24}) equals to $0$ by the similar computations to the case of the first triple.
\vskip 2mm    
Consider the fifth triple (\eqref{2.20}, \eqref{2.20}$'$, \eqref{2.25}). Then
\begin{align*}
    \eqref{2.20}'&=\displaystyle\sum_{\alpha,\beta,\gamma,\delta  \in I}  (-1)^{\tilde{a}\tilde{c}+\tilde{a}\gamma+\tilde{c}\delta+\gamma+\delta+\gamma\delta+\tilde{b}\tilde{c}+\tilde{b}\delta+\tilde{b}\alpha+\tilde{c}\beta+\beta\delta+\alpha+\beta+\alpha\beta}(-1)^{\tilde{b}\delta+\beta\delta+\tilde{a}\tilde{b}+\tilde{a}\beta}\\
    &\hspace{0.5cm}\{ \theta_{\beta} \,_{y} \left\{ \theta_{\delta} \,_{x}\, c \right\}\}\left( \vert_{x=\lambda+\partial} (C^{-1})_{\delta\gamma}(\lambda+\partial)\left\{ a \,_\lambda\, \theta_{\gamma} \right\}\right)\left( \vert_{y=\mu+\partial} (C^{-1})_{\beta\alpha}(\mu+\partial)\left\{ b \,_\mu\, \theta_{\alpha} \right\}\right)\\
\end{align*}
and
\begin{align*}
    \eqref{2.25}&=\displaystyle\sum_{\alpha,\beta,\zeta,\eta  \in I}  (-1)^{\tilde{a}\alpha+\alpha+\tilde{a}\tilde{c}+\tilde{b}\tilde{c}+\alpha\eta+\eta+\zeta\eta+\tilde{c}\eta+\tilde{c}\zeta+\tilde{a}\zeta+\beta\zeta+\beta+\zeta}(-1)^{\tilde{b}\beta}\\
    &\hspace{0.5cm}\left\{ \left\{ \theta_{\eta} \,_{x}\, \theta_{\zeta} \right\} \,_{\lambda+\mu+\partial}\, c \right\}_{\rightarrow}\left( \vert_{x=\lambda+\partial}(C^{-1})_{\eta\alpha}(\lambda+\partial)\left\{ a \,_{\lambda}\, \theta_{\alpha} \right\}\right)(C^{-1})_{\zeta\beta}(\mu+\partial)\left\{ b \,_{\mu}\, \theta_{\beta} \right\}.
\end{align*}
Hence the sum of the three terms of equals to
\begin{align*}
& \displaystyle\sum_{\alpha,\beta,\gamma,\delta  \in I}  (-1)^{\tilde{a}\alpha+\alpha+\tilde{a}\tilde{c}+\tilde{b}\tilde{c}+\alpha\beta+\beta+\beta\delta+\tilde{c}\beta+\tilde{c}\delta+\tilde{a}\delta+\gamma\delta+\gamma+\delta+\tilde{b}\gamma}\\
&\hspace{1cm}( \left\{ \theta_{\beta} \,_{\lambda+\partial} \left\{ \theta_{\delta} \;_{\mu+\partial}\, c \right\} \right\}-(-1)^{\beta\delta}\left\{ \theta_{\delta} \,_{\mu+\partial} \left\{ \theta_{\beta} \;_{\lambda+\partial}\, c \right\} \right\}-
\left\{ \left\{ \theta_{\beta} \,_{\lambda+\partial}\, \theta_{\delta} \right\} \,_{\lambda+\mu+\partial}\, c \right\})\\            
&\hspace{1cm}\nonumber (C^{-1})_{\beta\alpha}(\lambda+\partial)\left\{ a \,_{\lambda}\, \theta_{\alpha} \right\} (C^{-1})_{\delta\gamma}(\mu+\partial)\left\{ b \,_{\mu}\, \theta_{\gamma} \right\}=0.
\end{align*}

\vskip 2mm

Consider the sixth triple (\eqref{2.21}, \eqref{2.21}$'$,\eqref{2.29}). We have
    \begin{align*}
        \eqref{2.21}'&=-\displaystyle\sum_{\alpha,\beta,\gamma,\delta,\zeta,\eta  \in I}  (-1)^{\tilde{a}\tilde{b}+\tilde{a}\beta+\tilde{a}\tilde{c}+\tilde{a}\gamma+\tilde{c}\delta+\gamma+\delta+\tilde{b}\tilde{c}+\tilde{b}\alpha+\alpha+\beta+\alpha\beta+\eta\gamma+\eta+\delta\zeta+\zeta+\eta\zeta+\beta\zeta}\\
          &\hspace{3cm}\nonumber\{ \theta_{\delta} \,_{\lambda+\mu+\partial}\;\, c\}_{\rightarrow}(C^{-1})_{\delta\zeta}(\lambda+\mu+\partial)\left\{ \theta_{\beta} \,_{y} \left\{ \theta_{\eta} \;_{x}\, \theta_{\zeta} \right\} \right\}\\
                &\hspace{3cm}\nonumber\left( \vert_{x=\lambda+\partial} (C^{-1})_{\eta\gamma}(\lambda+\partial)\left\{ a \,_{\lambda}\, \theta_{\gamma} \right\}\right)\left( \vert_{y=\mu+\partial} (C^{-1})_{\beta\alpha}(\mu+\partial)\left\{ b \,_{\mu}\, \theta_{\alpha} \right\}\right)
    \end{align*}
    and
    \begin{align*}
        \eqref{2.29}&=-\displaystyle\sum_{\alpha,\beta,\gamma,\delta,\zeta,\eta  \in I}  (-1)^{\tilde{a}\tilde{c}+\tilde{b}\tilde{c}+\tilde{c}\delta+\gamma+\delta+\gamma\delta+\tilde{a}\alpha+\alpha+\alpha\eta+\eta+\zeta\eta+\gamma\eta+\gamma\zeta+\tilde{a}\zeta+\beta\zeta+\beta+\zeta+\tilde{b}+\beta}\\
        &\hspace{0.5cm}\nonumber\left\{ \theta_{\delta} \,_{\lambda+\mu+\partial}\, c \right\}_{\rightarrow}(C^{-1})_{\delta\gamma}(\lambda+\mu+\partial)\\
        &\hspace{0.5cm}\nonumber\left\{ \left\{ \theta_{\eta} \,_{x} \theta_{\zeta} \right\} \,_{\lambda+\mu+\partial}\, \theta_{\gamma} \right\}_{\rightarrow}\left( \vert_{x=\lambda+\partial}(C^{-1})_{\eta\alpha}(\lambda+\partial)\left\{ a \,_{\lambda}\, \theta_{\alpha} \right\}\right)(C^{-1})_{\zeta\beta}(\mu+\partial)\left\{ b \,_{\mu}\, \theta_{\beta} \right\}.
    \end{align*}
By changing the indices properly, we can show that the sum of \eqref{2.21}, \eqref{2.21}$'$ and \eqref{2.29} is 
    \begin{align*}
&-\sum_{\alpha,\beta,\gamma,\delta,\zeta,\eta  \in I}    (-1)^{\tilde{a}\tilde{c}+\tilde{b}\tilde{c}+\tilde{b}\gamma+\tilde{c}\delta+\zeta+\delta+\delta\zeta+\tilde{a}\alpha+\alpha+\alpha\beta+\beta\eta+\beta\zeta+\zeta\eta+\tilde{a}\eta+\gamma\eta+\gamma+\eta+\tilde{b}\gamma} \\
        &\hspace{2cm}\left\{ \theta_{\delta} \,_{\lambda+\mu+\partial}\, c \right\}_{\rightarrow}(C^{-1})_{\delta\gamma}(\lambda+\mu+\partial)\\
        &\hspace{2cm} ( \left\{ \theta_{\beta} \,_{\lambda+\partial} \left\{ \theta_{\eta} \;_{\mu+\partial}\, \theta_{\zeta} \right\} \right\}-(-1)^{\beta\eta}\left\{ \theta_{\eta} \,_{\mu+\partial} \left\{ \theta_{\beta} \;_{\lambda+\partial}\, \theta_{\zeta} \right\} \right\}-
        \left\{ \left\{ \theta_{\beta} \,_{\lambda+\partial}\, \theta_{\eta} \right\} \,_{\lambda+\mu+\partial}\, \theta_{\zeta} \right\})\\            
        &\hspace{2cm}\nonumber (C^{-1})_{\eta\alpha}(\lambda+\partial)\left\{ a \,_{\lambda}\, \theta_{\alpha} \right\} (C^{-1})_{\zeta\beta}(\mu+\partial)\left\{ b \,_{\mu}\, \theta_{\gamma} \right\}=0.
    \end{align*}

\vskip 2mm
Finally, the sum of the last triple (\eqref{2.22}, \eqref{2.17}$'$, \eqref{2.28}) equals to $0$ by the similar arguments as in the case of the second triple (\eqref{2.17}, \eqref{2.22}$'$, \eqref{2.30}), we proved the Jacobi identity.
\end{proof}
  
  \begin{thm}\label{Theorem:super Dirac, part2}
Consider the PVSA $ \mathcal{P}$ and the invertible element $C(\lambda)\in \mathcal{M}(\lambda)$  in Theorem \ref{Theorem:super Dirac, part1}. Then the following statements hold.
  \begin{itemize}
      \item[\textrm{(a)}] For any $i\in I$ and $a \in  \mathcal{P}$, we have  $\left\{ a \,_{\lambda}\, \theta_i \right\}^{D}=\left\{ \theta_i \,_{\lambda}\, a \right\}^{D}=0$.
      \item[\textrm{(b)}]  Consider the differential algebra ideal $\left< \theta_I \right>$ of $\mathcal{P}$ generated by $\theta_{I}$. Then the corresponding Dirac reduction induces a PVSA structure on  $\mathcal{P}/\left< \theta_I \right>$. In other words, the bracket $\left\{ \cdot {}_{\lambda} \cdot \right\}^{D} : \mathcal{P}/\left< \theta_I \right>\times \mathcal{P}/\left< \theta_I \right> \rightarrow \mathcal{P}/\left< \theta_I \right>(\!(\lambda^{-1})\!)$ is a well-defined Poisson $\lambda$-bracket.
  \end{itemize}
\end{thm}

\begin{proof}
  For $i\in I$ and $a\in \mathcal{P}$, we have
  \allowdisplaybreaks
  \begin{align*}
    \left\{ a \,{}_{\lambda}\, \theta_i \right\}^{D}
    &=\left\{ a \,{}_{\lambda}\, \theta_i \right\}-\displaystyle\sum_{\alpha,\beta\in I}  (-1)^{\tilde{a}i+\tilde{a}\alpha+i\beta+\alpha\beta+\alpha+\beta}\left\{ \theta_\beta \,{}_{\lambda+\partial}\, \theta_i \right\}_{\rightarrow}(C^{-1})_{\beta\alpha}(\lambda+\partial)\left\{ a \,{}_{\lambda}\, \theta_\alpha \right\}=0
  \end{align*}
since
  \begin{equation*}
    \displaystyle\sum_{\beta\in I}  (-1)^{i\beta+\alpha\beta+\beta+i\alpha}\left\{ \theta_\beta \,{}_{\lambda+\partial}\, \theta_i \right\}_{\rightarrow}(C^{-1})_{\beta\alpha}(\lambda)=\delta_{i\alpha}.
  \end{equation*}
 In addition, the equality $\{\theta_i{}_\lambda a\}^{D}=0$ holds by the skewsymmetry. The statement (b) directly follows from (a).
\end{proof}

\subsection{Dirac reduction of Poisson superalgebras} \label{Subsec:Dirac PA} \hfill

As a finite analogue of Section \ref{Subsec:Dirac PVA}, we briefly introduce the Dirac reduction for Poisson superalgebras.  On the other hand, the result in this section can be understood as a super-analogue of \cite{Dirac50}.

Let $\left(P, \{\cdot,\cdot\}\right)$ be a Poisson superalgebra and  $I=I_{0} \sqcup I_{1}$ be the same index set as in Section \ref{Subsec:Dirac PVA}. Consider a subset $\theta_{I} = \{\theta_{i} |i\in I\}\subset P$ such that $(-1)^{p(\theta_i)}= (-1)^i$ and the matrix
\begin{equation} \label{eq: DR matrix for PA} 
C:= \sum_{i,j\in I} e_{ij} \otimes  \{ \theta_{i}, \theta_{j} \}\in \text{Mat}_{(r|s)}\otimes P
\end{equation}
where $r= |I_{0}|$ and $s=|I_1|$.

\begin{defn}  Assuming the element $C$ in \eqref{eq: DR matrix for PA} is invertible, the \textit{Dirac reduced bracket} on $P$  associated with $\theta_{I}$ is the bilinear map $\{ \cdot , \cdot\}^{D} : P\times P \rightarrow P$ given by
  \begin{equation}\label{eq: DR bracket for PA}
    \left\{ a \,, b \right\}^{D}=\left\{ a \,, b \right\}-\displaystyle\sum_{i,j \in I} (-1)^{ij+j}\left\{ a \,, \theta_i \right\}(C^{-1})_{ij}\left\{ \theta_j \,, b \right\}.
  \end{equation}
\end{defn}
By the similar proof to Theorem \ref{Theorem:super Dirac, part1} and \ref{Theorem:super Dirac, part2}, we get the following theorem.

\begin{thm}
Suppose the matrix $C$ in \eqref{eq: DR matrix for PA} is invertible.
  \begin{itemize}
   \item[\textrm{(a)}] The Dirac reduced bracket $\{ \cdot , \cdot\}^{D}$ given in \eqref{eq: DR bracket for PA}  is a Poisson bracket on $P$.
   \item[\textrm{(b)}]  All the elements in $\theta_{I}$ are central with respect to $\{\cdot, \cdot\}^{D}$.
   \item[\textrm{(c)}]  Let $\left< \theta_I \right>$ be the associative algebra ideal of $P$ generated by $\theta_{I}$. Then the bracket $\left\{ \cdot , \cdot \right\}^{D}$ induces a well-defined Poisson bracket on $P/\left< \theta_I \right>$.
   \end{itemize}\qed
\end{thm}

\section{Dirac reduction of supersymmetric Poisson vertex algebras} \label{Sec: Dirac-SUSY}

In this section, we introduce non-local $N_k=1$ supersymmetric (SUSY) Poisson vertex algebras (PVAs) and their Dirac reduction. Since we only deal with the $N_k=1$ case, we use the terms SUSY LCAs and SUSY PVAs instead of  $N_k=1$ SUSY LCAs and  $N_k=1$ SUSY PVAs, respectively. For detailed properties of local SUSY PVAs, we refer to \cite{HK07}. 

\subsection{Supersymmetric Poisson vertex algebras}\label{Subcec: SUSY PVA}\hfill 

For an odd indeterminate $\chi$ and an even indeterminate $\lambda$, let us consider the noncommutative associative algebra $\CC[\Lambda]$ generated by $\Lambda=(\lambda, \chi)$   subject to the relations: 
\begin{equation} \label{eq:susy relation}
 \chi^2=-\lambda, \quad \chi \lambda= \lambda \chi.
\end{equation}
 Then $\CC[\Lambda]\simeq \CC[\lambda]\oplus \chi \CC[\lambda]$ as vector superspaces. Similarly, we define another associative algebra $\CC(\!(\Lambda^{-1})\!):= \CC(\!(\lambda^{-1})\!)\oplus \chi \CC(\!(\lambda^{-1})\!)$ generated by $\chi, \lambda$ and $\lambda^{-1}$ satisfying $\lambda \lambda^{-1}=\lambda^{-1} \lambda=1$, $\chi \lambda^{-1}=\lambda^{-1} \chi$ along with the relation \eqref{eq:susy relation}.

Let $\mathscr{R}$ be a superspace endowed with an odd derivation $D:\mathscr{R}\to \mathscr{R}$. In other words, we consider a $\CC[D]$-module $\mathscr{R}$. We extend the $\CC[D]$-module structure to $\mathscr{R} (\!( \Lambda ^{-1})\!):= \CC(\!(\Lambda^{-1})\!)\otimes \mathscr{R}$ via the relations
\begin{equation}\label{eq: chi and D}
D\chi= -\chi D+ 2\lambda, \quad D\lambda^{\pm}= \lambda^{\pm} D.
\end{equation}
 A {\it ($N_k=1$) SUSY non-local $\Lambda$-bracket} is a {\it parity reversing} bilinear map 
\[[ \cdot {}_\Lambda \cdot ]: \mathscr{R} \times \mathscr{R} \to \mathscr{R} (\!( \Lambda ^{-1})\!)\]
satisfying the {\it sesquilinearity} :  
\begin{equation} \label{sesqui}
 [Da _\Lambda b]=\chi[ a _\Lambda b], \quad [a _\Lambda Db]=-(-1)^{\tilde{a}}(D+\chi)[ a _\Lambda b]
 \end{equation}
 for $a,b \in \mathscr{R}$. We denote the coefficient of $\lambda^n \chi^i$ of $[a{}_\Lambda b]$ for $n\in \ZZ$, $i\in\{0,1\}$ by $a_{(n|i)}b\in \mathscr{R}$ so that $[a{}_\Lambda b] = \sum_{n\in \ZZ, \, i=\{0,1\}} \lambda^n \chi^{i} a_{(n|i)}b.$ In addition, we let $\partial:= D^2$.
 
Consider a pair $\Gamma=(\mu,\gamma)$ of an even indeterminate $\mu$ and an odd indeterminate $\gamma$ and the associative algebra $\CC(\!(\Gamma^{-1})\!)$ which is  isomorphic to $\CC(\!(\Lambda^{-1})\!)$ via the map defined by $\mu \mapsto \lambda$ and $\gamma \mapsto \chi$. Then 
\begin{equation} \label{eq:Jacobi_first}
[a{}_\Lambda [b{}_\Gamma c]] =\sum_{m,n\in \ZZ, \, i,j\in \{0,1\}} (-1)^{j(\tilde{a}+1)} \mu^n\gamma^j \lambda^m\chi^i a_{(m|i)}(b_{(n|j)}c )\in \mathscr{R}(\!(\Lambda^{-1})\!) (\!(\Gamma^{-1})\!).
\end{equation}
Assuming $\Lambda$ and $\Gamma$ supercommute, one can check that  $\mathscr{R}(\!(\Lambda^{-1})\!) (\!(\Gamma^{-1})\!)\simeq\CC_{\chi\gamma} \otimes \mathscr{R}(\!(\lambda^{-1})\!) (\!(\mu^{-1})\!)$ where $\CC_{\chi\gamma}:= \CC \oplus \CC \chi \oplus  \CC\gamma  \oplus \CC\chi\gamma$. Recall the notation \eqref{eq:admissibility, space} and denote $\mathscr{R}_{\Lambda, \Gamma}:= \CC_{\chi\gamma} \otimes \mathscr{R}_{\lambda, \mu}$. The superspace  $\mathscr{R}_{\Lambda, \Gamma}$ can be embedded in $\mathscr{R}(\!(\Lambda^{-1})\!) (\!(\Gamma^{-1})\!)$
 and if
  \begin{equation}
[a{}_\Lambda [b{}_\Gamma c]] \in \mathscr{R}_{\Lambda, \Gamma} 
\end{equation}
 for any $a,b,c\in \mathscr{R}$, then the $\Lambda$-bracket is called {\it admissible}.

\begin{defn} \label{Def:non-local SUSY}
  A \textit{non-local   SUSY LCA} is a $\CC[D]$-module $\mathscr{R}$ endowed with an admissible non-local $\Lambda$-bracket which satisfies the following axioms for  $a,b,c\in \mathscr{R}$:
  \begin{itemize}
    \item (skewsymmetry) $[a _\Lambda b]=(-1)^{\tilde{a}\tilde{b}}{}_{\leftarrow}[b \,_{-\Lambda-\nabla}\, a]$,
     for  $-\Lambda-\nabla=(-\lambda-\partial, -\chi-D)$,
       \item (Jacobi identity) $[ a _\Lambda [ b \,_\Gamma\, c]]+(-1)^{\tilde{a}}[[ a _\Lambda b] \,_{\Lambda+\Gamma}\, c]=(-1)^{(\tilde{a}+1)(\tilde{b}+1)}[b \,_\Gamma [ a _\Lambda c]]$ in $\mathscr{R}_{\Lambda, \Gamma}$.
      \end{itemize}
   
\end{defn}

The RHS of the skewsymmetry can be written as 
    \begin{equation*}
      {}_{\leftarrow}\left[ b {}_{-\Lambda-\nabla} a \right]=\displaystyle\sum_{n\in\ZZ, \, i=\{0,1\}}(-\lambda-\partial)^{n}(-\chi-D)^i b_{(n|i)}a
    \end{equation*}
and the Jacobi identity can be understood via \eqref{eq:Jacobi_first} and 
    \begin{align*}
      \left[ \chi^{i}a_{(n|i)}b \;{}_{\Lambda+\Gamma}\; c \right]=(-\chi)^{i}[ a_{(n|i)}b \,{}_{\Lambda+\Gamma}\, c ].
    \end{align*}

\begin{defn} \label{Def:non-local SUSYPVA}
  A \textit{non-local  SUSY PVA} is a tuple ($\mathscr{P}, D, \{\cdot _\Lambda \cdot\}, \cdot)$ which satisfies the following axioms:
  \begin{itemize}
    \item ($\mathscr{P}, D, \{\cdot _\Lambda \cdot\})$ is a non-local SUSY LCA.
    \item ($\mathscr{P}, D, \cdot)$ is an unital supercommutative associative differential algebra.
    \item$\{a _\Lambda bc\}=\{a _\Lambda b\}c+(-1)^{\tilde{a}\tilde{b}+\tilde{b}}b\{a _\Lambda c\}$   for  $a, b, c \in \mathscr{P}.$
  \end{itemize}
Note that the last axiom is called the \textit{left Leibniz rule}.
\end{defn}

\begin{prop} \label{prop:non-SUSY and SUSY}
  A non-local SUSY PVA is a non-local PVSA. More precisely, if $\mathscr{P}$  is a non-local SUSY PVA, then the $\lambda$-bracket defined by 
 \begin{equation}\label{5.2}
   \{a{}_\lambda b\}:= \sum_{n\in \ZZ} \lambda^n a_{(n|1)}b, \quad a,b\in \mathscr{P}
 \end{equation}
 induces a non-local PVSA structure on $\mathscr{P}$.
\end{prop}

\begin{proof}
The sesquilinearity, skewsymmetry, Jacobi identity and Leibniz rule of PVAs can be checked similarly to Proposition 4.3 in \cite{Suh20}. The only thing to check is the admissibility.  By \eqref{eq:Jacobi_first}, we have 
\begin{equation} \label{eq:admissibility_SUSY-nonSUSY}
\{a{}_\lambda \{ b{}_\mu c\}\} = \sum_{m,n\in \ZZ}  \mu^n\lambda^m a_{(m|1)} (b_{(n|1)}c)  
\end{equation}
and the admissibility of the $\Lambda$-bracket implies $\eqref{eq:admissibility_SUSY-nonSUSY} \in \mathscr{R}_{\lambda, \mu}$. 
\end{proof}

We note that non-local SUSY $\Lambda$-bracket on a $\CC[D]$-module $\mathscr{R}$ is admissible if and only if the induced $\lambda$-bracket on $\mathscr{R}$ is admissible. The only if part can be shown as in the proof of Proposition \ref{prop:non-SUSY and SUSY}. In order to see the converse, observe that if the $\lambda$-bracket is admissible then $[a{}_\lambda [b{}_\mu c]],$ $ [Da{}_\lambda [b{}_\mu c]],$  $ [a{}_\lambda [Db{}_\mu c]],$ $[Da{}_\lambda [Db{}_\mu c]]\in \mathscr{R}_{\lambda, \mu}$ for $a,b,c\in \mathscr{R}$. In other words,    $\sum_{n,m\in \ZZ}\lambda^m \mu^n a_{(m|i)}(b_{(n|j)} c) \in \mathscr{R}_{\lambda, \mu}$ for $i,j=0,1$. Hence, by \eqref{eq:Jacobi_first}, we have $[a{}_\Lambda [b {}_\Gamma c]] \in \mathscr{R}_{\Lambda, \Gamma}$.

\begin{prop}\hfill
  \begin{itemize}
    \item[\textrm{(1)}] Let $V$ be a vector superspace and $\mathscr{R}(V) := \CC[D]\otimes V$. Suppose
    \begin{equation} \label{eq:SUSY bracket on space}
      \left[ \cdot {}_{\Lambda} \cdot \right] : V \times V \rightarrow  \mathscr{R}(V)(\!(\Lambda^{-1})\!)
    \end{equation}
    is an odd bilinear map satisfying the admissibility, skewsymmetry and Jacobi identity of SUSY LCAs. Then \eqref{eq:SUSY bracket on space} can be extended to $\mathscr{R}(V)$ by the sesquilinearity, which gives a non-local SUSY LCA structure on the space $\mathscr{R}(V)$.
      \item[\textrm{(2)}] Let $\mathscr{R}$ be a SUSY LCA. Then the supersymmetric algebra $\mathscr{P}:=S(\mathscr{R})$ endowed with the $\Lambda$-bracket induced by that of $\mathscr{R}$ by Leibniz rule is a SUSY PVA.
  \end{itemize}
\end{prop}

\begin{proof}
 Since $\chi \mathscr{R}_{\Lambda, \Gamma}$,  $\gamma  \mathscr{R}_{\Lambda, \Gamma}$, and $(D+\chi+\gamma)  \mathscr{R}_{\Lambda, \Gamma}$ are subsets of  $\mathscr{R}_{\Lambda, \Gamma}$, we have 
\[ [ DV{}_\Lambda [V{}_\Gamma V]], \  [ V{}_\Lambda [DV{}_\Gamma V]], \ [ DV{}_\Lambda [V{}_\Gamma DV]] \subset \mathscr{R}_{\Lambda, \Gamma}.\]
One can inductively show that the $\Lambda$-bracket on $\mathscr{R}(V)$ is admissible. For the skewsymmetry and Jacobi identity, we refer to \cite{CS}. Hence we showed (1). Similarly, the skewsymmetry and Jacobi identity for (2) can be checked as Theorem 2.15 in \cite{CS}. The admissibility of the $\Lambda$-bracket on $S(\mathscr{R})$ is equivalent to the admissibility of the corresponding $\lambda$-bracket \eqref{5.2} and it already has been shown in Proposition \ref{Prop:PVA from generators}.
\end{proof}

Let us consider the differential algebra 
\begin{equation} \label{eq:SUSY diff poly}
\mathscr{P}=\CC[u_i^{[n]}|i\in I, \, n\in \ZZ_{+}]
\end{equation}
 of polynomials generated by homogeneous variables $u_i$ endowed with the odd derivation $D$ such that  $D: u_i^{[n]} \mapsto u_i^{[n+1]}$. If we denote $\tilde{i}:= \tilde{u}_i$, then $\widetilde{u}_i^{[n]}\equiv \tilde{i}+n$ (mod 2). Denote by $\frac{\partial}{\partial u_{i}^{[m]}}:\mathscr{P}\to \mathscr{P}$ the derivation of parity $\tilde{u}_i^{[m]}$, which satisfies the property 
 $ \frac{\partial}{\partial u_{i}^{[m]}}(u_{j}^{[n]})=\delta_{i,j}\delta_{m,n}.$ If $\mathscr{P}$ is a non-local SUSY PVA, then the $\Lambda$-brackets between two elements $a$, $b \in \mathscr{P}$ can be obtained by the following theorem.

\begin{thm}[\cite{CS},  Master formula] \label{SUSY Master} Suppose $\mathscr{P}$ in \eqref{eq:SUSY diff poly} is a non-local SUSY PVA. For $a, b \in \mathscr{P}$, we have
  \begin{align*}
    \left\{ a {}_{\Lambda} b \right\}&=\displaystyle\sum_{i,j \in I,\; m,n\in \ZZ_{\geq 0}}S(a_{(i,m)},b_{(j,n)})b_{(j,n)}\left\{ a {}_{\Lambda+\nabla} b \right\}_{\rightarrow}a_{(i,m)}\\
    &=\displaystyle\sum_{i,j \in I,\; m,n\in \ZZ_{\geq 0}}S(a_{(i,m)},b_{(j,n)})(-1)^{n(i+m+1)+m(i+j+1)+\frac{m(m-1)}{2} }\\
    &\hspace{3cm}b_{(j,m)}(\chi+D)^{n}\left\{ u_{i} {}_{\Lambda+\nabla} u_{j} \right\}_{\rightarrow}(\chi+D)^{m}a_{(i,m)},
  \end{align*}
  where 
  \begin{itemize}
  \item $\Lambda+\nabla= (\lambda+ \partial, \chi+D),$
    \item $a_{(i,m)}:=\frac{\partial}{\partial u_{i}^{[m]}}a$ and $b_{(j,n)}:=\frac{\partial}{\partial u_{j}^{[n]}}b$,
    \item $S(a_{(i,m)},b_{(j,n)}):=(-1)^{\tilde{b}_{(j,n)}}(-1)^{\tilde{b}_{(j,n)}( \tilde{u}_{j}^{[n]}+\tilde{a})}(-1)^{\tilde{a}_{(i,m)}\tilde{u}_{j}^{[n]}}$.
  \end{itemize}
  \end{thm}

  \begin{ex} [affine SUSY PVA] \label{Ex:affine SUSY PVA}
  Let $\g$ be a simple finite Lie superalgebra with a nondegenerate supersymmetric invariant even bilinear form $(\cdot|\cdot)$. For the parity reversed space $\bar{\g}$ of $\g$, the affine SUSY PVA is $\mathscr{V}^k(\bar{\g})= S(\CC[D]\otimes \bar{\g})$ endowed with the $\Lambda$-bracket 
  \[ \{\bar{a}{}_\Lambda \bar{b}\} = (-1)^{\tilde{a}} (\overline{[a,b]} +k\chi (a|b))\]
  for $a,b\in \g$. On the other hand, since
  
  \begin{align*}
    \{D\bar{a} {}_{\Lambda}D\bar{b}\}=\chi \left\{ \bar{a} \,{}_{\Lambda}\, D\bar{b} \right\}=-\lambda \overline{[a,b]}+\chi\big(D\overline{[a,b]}+k\lambda(a|b)\big),
  \end{align*}
  the affine PVSA of level $k$ can be embedded in  $\mathscr{V}^k(\bar{\g})$ via the injective PVSA  homomorphism  
 $a \mapsto D\bar{a}$ for $a\in \g$.
  \end{ex}

\begin{ex}
Let $\mathscr{P}=\CC[u^{[n]}| n\in \ZZ_+]$ be the differential algebra generated by an even variable $u$. The bracket 
\[ \{u{}_\Lambda u\}= \lambda^{-1} \chi\]
defines a non-local SUSY PVA structure on $\mathscr{P}$. The map from the non-local PVSA in Example \ref{ex:non-local PVA} to $\mathscr{P}$ defined by $u\mapsto u$ is an injective PVSA homomorphism.
\end{ex}

\subsection{Dirac reduction of supersymmetric PVAs}\label{Subsec: Dirac-SUSY}
~\\
Let $\mathscr{P}$ be a SUSY PVA and let $\text{Mat}_{(r|s)}=\text{End}(\CC_{(r|s)})$ for nonnegative integers $r$ and $s$. 
Consider  $F(\Lambda)= F_0(\lambda)+ F_1(\lambda)\chi$, $G(\Lambda)= G_0(\lambda)+ G_1(\lambda)\chi \in \mathscr{P}(\!(\Lambda^{-1})\!)$ for $F_i(\lambda), G_i(\lambda) \in \mathscr{P}(\!(\lambda^{-1})\!)$ where $i \in \{0,1\}$. Then the product $\circ$ on $\mathscr{P}$ defined by
\begin{equation*}
 F(\Lambda)\circ G(\Lambda):= F(\Lambda+\nabla)G(\Lambda)=\left(F_0(\lambda+\partial)+ F_1(\lambda+\partial) (\chi+D) \right)\left( G_0(\lambda)+ G_1(\lambda)\chi\right)
\end{equation*}
gives an associative algebra structure on $\mathscr{P}(\!(\Lambda^{-1})\!)$. The {\it adjoint}  $F^*(\Lambda)\in \mathscr{P}(\!(\Lambda^{-1})\!)$ of $F(\Lambda) \in  \mathscr{P}(\!(\Lambda^{-1})\!)$ is defined by 
\begin{equation*}
F^*(\Lambda)= \sum_{ n \in \ZZ_+} \big((-\lambda-\partial)^n F_{(n|0)} +(-1)^{\tilde{F}+1} (-\chi-D) (-\lambda -\partial)^n F_{(n|1)}\big),
\end{equation*}
where $F(\Lambda)=  \sum_{ n \in \ZZ} \big( F_{(n|0)} \lambda^n +  F_{(n|1)} \chi \lambda^n\big).$

Similarly, the superspace $\mathcal{M}_{(r|s)}(\Lambda):=\text{Mat}_{(r|s)} \otimes \mathscr{P}(\!(\Lambda^{-1})\!)$ is an associative algebra with respect to the product
\begin{equation}\label{3.1}
  (a \otimes F(\Lambda)) \circ (b \otimes G(\Lambda))=(-1)^{\tilde{b}\tilde{F}}  ab \otimes F(\Lambda)\circ G(\Lambda).
\end{equation}
The multiplication \eqref{3.1} can be generalized to the map:
\begin{align*}
 \circ\ : \left(\text{Mat}_{(p_1|q_1)(r_1|s_1)} \otimes \mathscr{P}(\!(\Lambda^{-1})\!)\right) &\times \left(\text{Mat}_{(p_2|q_2)(r_2|s_2)}\otimes  \mathscr{P}(\!(\Lambda^{-1})\!)\right) \rightarrow \text{Mat}_{(p_1|q_1)(r_2|s_2)} \otimes  \mathscr{P}(\!(\Lambda^{-1})\!) \\
  \Big(\left( a \otimes F(\Lambda)\right) &, \left( b \otimes G(\Lambda)\right)\Big)\longmapsto (-1)^{\tilde{b}\tilde{F}}  ab \otimes F(\Lambda+\nabla) G(\Lambda)
  \end{align*}
  for  $p_1,q_1,r_1=p_2,s_1=q_2, r_2, s_2\in \ZZ_+$. Likewise Section \ref{Sec:DR for PVA}, we also denote $\mathcal{M}_{(r|s)}(\Lambda)$ by $\mathcal{M}(\Lambda)$. Since the identity element in $\mathcal{M}(\Lambda)$ is $\text{Id}_{(r|s)}\otimes 1$, an element $A(\Lambda) \in \mathcal{M}(\Lambda)$ is called \textit{invertible} if there is an element $A^{-1}(\Lambda) \in \mathcal{M}(\Lambda)$ such that
\begin{equation}\label{eq: SUSY invertible element}
  A(\Lambda)\circ A^{-1}(\Lambda)=A^{-1}(\Lambda)\circ A(\Lambda)=\text{Id}_{(r|s)} \otimes 1.
\end{equation}

Let $I=I_{0} \sqcup I_{1} \subset \ZZ$ be a finite index set, where $I_{i}=I \cap (2\ZZ+ i)$ and $i \in \{0,1\}$. Let $\theta_{I}:=\{\, \theta_i \, |\, i\in I\}$ be a subset of $\mathscr{P}$ consisting of homogeneous elements such that $(-1)^{\tilde{\theta}_i} =(-1)^i$. For $r= |I_0|$ and  $s= |I_1|$ consider the element
\begin{equation}\label{3.3}
  C(\Lambda):=\displaystyle\sum_{i\in I}e_{ij} \otimes \{ \theta_{j} \,{}_{\Lambda}\, \theta_{i} \} \in \mathcal{M}(\Lambda).
\end{equation}
\begin{defn}\label{Defn:SUSY Dirac}
  Suppose $C(\Lambda)$ in \eqref{3.3} is invertible. The \textit{Dirac reduced bracket} of the SUSY $\Lambda$-bracket $\{ \cdot \,{}_{\Lambda}\, \cdot \}$ on $\mathscr{P}$  associated with $\theta_{I}$ is the bilinear map
  \begin{equation}
   \left\{\cdot {}_{\Lambda} \cdot \right\}^{D} : \mathscr{P} \times \mathscr{P} \rightarrow \mathscr{P}(\!(\Lambda^{-1})\!)
  \end{equation}
  given by
  \begin{equation}\label{3.5}
    \{a{}_{\Lambda} b\}^{D}=\{a{}_{\Lambda}b\}-\sum_{i,j \in I}(-1)^{(\tilde{a}+j)(\tilde{b}+i)}\{\theta_{j}\,{}_{\Lambda+\nabla}\,b\}_{\rightarrow}\left(C^{-1}\right)_{ji}(\Lambda+\nabla)\{a\,{}_{\Lambda}\,\theta_{i}\}
  \end{equation}
for $a,b\in \mathscr{P}$.
\end{defn}

  \begin{prop} \label{prop:Trivial example, SUSY PVA}
  Let $\mathscr{P}= \CC[\theta^{(n)}_k|\,k\in  I_{0} \sqcup  I_{1}, n \in \ZZ_{+}]$ be a SUSY PVA. If $C(\Lambda)$ in \eqref{3.3} is invertible, the corresponding Dirac reduced bracket is the trivial bracket.
  \end{prop}
  
  \begin{proof}
  
We use the SUSY analogue of the argument proving Proposition 2.7.

  
  By the master formula in Theorem \ref{SUSY Master}, we have
  \begin{equation}\label{eq:SUSY master}
  \begin{aligned}
  \left\{ a {}_{\Lambda} b \right\}
  =\displaystyle\sum_{i,j \in  I, \, m,n\in \ZZ_+}(-1)^{\tilde{a}\tilde{b}+\tilde{b}+(\tilde{b}+i+m)(j+n)+n(i+m+1)+m(i+j+1)+\lfloor\frac{m}{2} \rfloor }\\
  \frac{\partial b}{\partial\theta_{j}^{[n]}}(\chi+D)^n\left\{ \theta_{i} {}_{\Lambda+\nabla} \theta_{j} \right\}_{\to}(-\chi-D)^m\frac{\partial a}{\partial \theta_{i}^{[m]}}
  \end{aligned}
  \end{equation}
for   $a,b\in \mathscr{P}$.
 If we identify
  $\mathscr{P}(\!(\Lambda^{-1})\!)= \text{Mat}_{(1|0)} \otimes \mathscr{P}(\!(\Lambda^{-1})\!)$, then
 the RHS of \eqref{eq:SUSY master} can be written as 
  \begin{equation} \label{eq:SUSY master-ij}
  \begin{aligned}
    \displaystyle\sum_{ m,n\in \ZZ_+}&(-1)^{\tilde{a}\tilde{b}+\tilde{b}+bn+im+j+n+m+\lfloor\frac{m}{2} \rfloor}\\
    &\left(e_{j}^{T} \otimes\frac{\partial b}{\partial\theta_{j}^{[n]}}(\chi+D)^{n}\right) \big(e_{ji} \otimes C_{ji}(\Lambda+\nabla) \big)_{\to}\left(e_{i} \otimes (-\chi-D)^{m}\frac{\partial a}{\partial \theta_i^{[m]}}\right)
      \end{aligned}
      \end{equation}
      for given $i,j\in  I$. Note that 
  \begin{equation}\label{eq:SUSY trivialDR-1}
   e_{ji}\otimes C_{ji}(\Lambda)= \sum_{k,l\in  I}(e_{jk}\otimes C_{jk}(\Lambda+\nabla)) (e_{kl}\otimes (C^{-1})_{kl}(\Lambda+\nabla)) (e_{li} \otimes C_{li}(\Lambda))
  \end{equation}
  and 

  \begin{align}
    \label{3.10}&\{ a \,{}_{\Lambda}\, \theta_j \}=\displaystyle\sum_{i \in  I, m \in \ZZ_{+}}(-1)^{j(\tilde{a}+i+m)+m(i+j+1)+\frac{m(m-1)}{2}}\{ \theta_{i} \,{}_{\Lambda+\nabla}\, \theta_j \}(-\chi-D)^{m}\frac{\partial a}{\partial \theta_{i}^{[m]}},\\
    \label{3.11}&\{ \theta_{i} \,{}_{\Lambda}\, b \}=\displaystyle\sum_{j \in  I, n \in \ZZ_{+}}(-1)^{(\tilde{b}+j+n)(j+n+i+1)+n(i+1)}\frac{\partial b}{\partial \theta_{j}^{[n]}}(\chi+D)^{n}\{ \theta_{i} \,{}_{\Lambda}\, \theta_{j} \}
  \end{align}
  if we use the fact that $\left\lfloor \frac{m}{2}\right\rfloor=\frac{m(m-1)}{2} \;(\text{mod}\; 2)$. Now substitute $\big(e_{ji} \otimes C_{ji}(\Lambda+\nabla) \big)$ in \eqref{eq:SUSY master-ij} with the RHS of \eqref{eq:SUSY trivialDR-1}. Then by \eqref{3.10} and \eqref{3.11}, we get 
  \begin{equation}\label{eq:SUSY trivialDR-3}
  \begin{aligned}
  \{a{}_\Lambda b\}& = \sum_{k,l\in  I}(-1)^{\tilde{a}\tilde{b}+\tilde{a}l}(e_k^T \otimes \{\theta_k{}_{\Lambda+\nabla} b\}) (e_{kl}\otimes (C^{-1})_{kl}(\Lambda+\nabla))(e_l \otimes \{a{}_\Lambda \theta_l\})
  \\
  &  =\sum_{k,l\in  I} (-1)^{(\tilde{a}+k)(\tilde{b}+l)}\{\theta_k{}_{\Lambda+\nabla} b\}\, (C^{-1})_{kl}(\Lambda+\nabla)\, \{a{}_\Lambda \theta_l\}.
  \end{aligned}
  \end{equation}
  Hence, $\{a{}_\Lambda b\}^D= \{a{}_\Lambda b\}- [\text{RHS of }\eqref{eq:SUSY trivialDR-3}]=0$.
   \end{proof}

Let us introduce the following lemma which is needed to prove Theorem \ref{Theorem:SUSY Dirac, part1}.
\begin{lem}\label{Lemma:Inverse operator-SUSY}
  Let $\left\{ \cdot \,_\Lambda\, \cdot \right\} : \mathscr{P} \times \mathscr{P} \rightarrow \mathscr{P}(\!(\Lambda^{-1})\!)$ be a non-local SUSY $\Lambda$-bracket on $\mathscr{P}$. Suppose $C(\Lambda)=(C_{ij}(\Lambda))_{i,j\in  I} \in \mathcal{M}(\Lambda)$ in \eqref{3.3} is invertible  and let $(C^{-1})(\Lambda)=((C^{-1})_{ij}(\Lambda))_{i,j\in  I}\in  \mathcal{M}(\Lambda)$ be its inverse. For $a \in \mathscr{P}$,  assume that 
  \begin{center}
    $\left\{ a \,{}_{\Lambda}\, C_{ij}(\Gamma) \right\} \in \mathscr{P}_{\Lambda,\Gamma}$ \,for all $i, j\in  I$.
  \end{center}
  Then we have $\left\{ a \,{}_{\Lambda}\, (C^{-1})_{ij} (\Gamma) \right\}$, $\left\{ (C^{-1})_{ij}(\Gamma) \,{}_{\Lambda+\Gamma}\, a \right\} \in \mathscr{P}_{\Lambda,\Gamma}$.  If $a\in \mathscr{P}$ is homogenous, then
\begin{equation}\label{Lemma:Inverse operator-SUSY1}
\begin{aligned}
      &\{a \,{}_{\Lambda}\, (C^{-1})_{ij}(\Gamma)\}\\
      &=-\displaystyle\sum_{r,t\in  I, n \in \ZZ}(-1)^{(i+t)(j+r)+(\tilde{a}+1)(i+r+1)}(C^{-1})_{ir}(\Lambda+\Gamma+\nabla)\left\{ a \,{}_{\Lambda}\, C_{rt} (\Gamma+\nabla)\right\} (C^{-1})_{tj}(\gamma)
\end{aligned}
\end{equation}
      and
\begin{equation}\label{Lemma:Inverse operator-SUSY2}
\begin{aligned}
        &\{(C^{-1})_{ij}(\Lambda) \,{}_{\Lambda+\Gamma}\, a\}\\
        &=-\displaystyle\sum_{r,t\in I, n \in \ZZ}(-1)^{\tilde{a}(i+j+r+t)+(r+j)(t+j)}\left\{ C_{rt}(\Lambda+\nabla) \,{}_{\Lambda+\Gamma+\nabla}\, a \right\}_{\rightarrow}(C^{-1})_{tj}(\Lambda)
        (C^{-1})^*_{ir}(\Gamma).
\end{aligned}
\end{equation}
  \end{lem}
\begin{proof}
It is same argument as in the proof of Lemma \ref{Lemma:Inverse operator}. For the admissibility, we can apply the similar proof to Lemma 2.4 in \cite{DSK13}. 

Let us show \eqref{Lemma:Inverse operator-SUSY1} and \eqref{Lemma:Inverse operator-SUSY2}.
Since $(C^{-1})(\Lambda) \circ C(\Lambda)= C(\Lambda)\circ (C^{-1})(\Lambda)= \text{Id}_{(r|s)}\otimes 1$, we have
    \begin{equation} \label{eq:Inverse operator-SUSY proof1}
    \begin{aligned}
      & (C^{-1})_{ij}(\Lambda)=\displaystyle\sum_{r,t\in  I}(-1)^{(i+t)(j+r)}(C^{-1})_{ir}(\Lambda)\circ C_{rt}(\Lambda)\circ (C^{-1})_{tj}(\Lambda)
    \end{aligned}
        \end{equation}
and 
\begin{equation}\label{eq:Inverse operator-SUSY proof2}
\begin{aligned}
&  \displaystyle\sum_{r\in  I}(-1)^{(i+r+1)(r+t)}C^{\!-1}_{ir}(\Gamma)\circ C_{rt}(\Gamma)  = \displaystyle\sum_{r\in  I}(-1)^{(r+t+1)(r+i)}C_{tr}(\Gamma)\circ C^{\!-1}_{ri}(\Gamma) =\delta_{it}.
\end{aligned}
\end{equation}    
\vskip2mm
\noindent  Now, using the left and right Leibniz rules, we obtain
    \begin{align}
    \nonumber &\left\{ a {}_{\Lambda} (C^{-1})_{ij}(\Gamma) \right\}
    =\displaystyle\sum_{r,t\in  I}(-1)^{(i+t)(j+r)}\left\{ a {}_{\Lambda} C^{-1}_{ir}(\Gamma)\circ C_{rt}(\Gamma)\circ C^{-1}_{tj}(\Gamma) \right\}\\
   \label{eq:Inverse operator-SUSY proof3} 
   &=\displaystyle\sum_{t\in  I}(-1)^{(i+t)(j+t)}\{ a {}_{\Lambda} \displaystyle\sum_{r\in  I}(-1)^{(i+r+1)(r+t)}C^{-1}_{ir}(\Gamma)\circ C_{rt}(\Gamma) \}C^{-1}_{tj}(\Gamma)\\ 
     \label{eq:Inverse operator-SUSY proof4} 
   &+\displaystyle\sum_{r\in  I}(-1)^{(i+j)(r+j)+(\tilde{a}+1)(i+r+1)}C^{-1}_{ir}( \Lambda+\Gamma+\nabla)\{ a {}_{\Lambda}  \displaystyle\sum_{t\in  I}(-1)^{(r+t+1)(t+j)}C_{rt}(\Gamma)\circ C^{-1}_{tj}(\Gamma) \} \\
  \nonumber &-\displaystyle\sum_{r,t\in  I}(-1)^{(i+t)(j+r)+(\tilde{a}+1)(i+r+1)}C^{-1}_{ir}(\Lambda+\Gamma+\nabla)\left\{ a {}_{\Lambda} C_{rt}(\Gamma+\nabla) \right\} C^{-1}_{tj}(\Gamma).    
\end{align}
By  \eqref{eq:Inverse operator-SUSY proof2}, we have $\eqref{eq:Inverse operator-SUSY proof3}=\eqref{eq:Inverse operator-SUSY proof4}=0$ and hence \eqref{Lemma:Inverse operator-SUSY1} holds.
 The last assertion \eqref{Lemma:Inverse operator-SUSY2} can be proved similarly. 
  \end{proof}

\begin{thm} \label{Theorem:SUSY Dirac, part1}
  Let $\mathscr{P}$ be a non-local SUSY Poisson vertex algebra with the SUSY $\Lambda$-bracket $\left\{ \cdot \,_{\Lambda}\, \cdot \right\}$. Let $\theta_i \in \mathscr{P}$ for each $i\in  I$ be homogeneous element. If the $C(\Lambda)\in \mathcal{M}(\Lambda)$ defined in \eqref{3.3} is invertible, then the Dirac reduced bracket $\left\{ \cdot \,_{\Lambda}\, \cdot \right\}^{D}$ given by \eqref{3.5} is a SUSY $\Lambda$-bracket on $\mathscr{P}$.
\end{thm}

\begin{proof}
The sesquilinearity, admissibility and Leibniz rule can be proved by direct computations.

 To show the skewsymmetry, we substitute $\Lambda$ by $-\Lambda-\nabla$ in \eqref{3.5} and let $\nabla$ act on the left. Then we have
    \begin{align*}
            &\left\{ a \,_{-\Lambda-\nabla}\, b \right\}^{D}\\
            &=\left\{ a \,_{-\Lambda-\nabla}\, b \right\}-\displaystyle\sum_{i,j \in  I}(-1)^{\tilde{a}\tilde{b}+\tilde{a}i+\tilde{b}j+ij}(-1)^{(\tilde{a}+i+1)(\tilde{b}+i)}(-1)^{(\tilde{b}+j+1)(i+j+1)}\\
            &\hspace{2.2cm}  \Big({}_{\leftarrow}\left\{ a\,_{-\Lambda-\nabla}\, \theta_i \right\}\Big)\left(\Big|_{X=-\Lambda-\nabla}(C^{-1})_{ji}(X)\right)\Big({}_{\leftarrow}\left\{ \theta_j \,_{-\Lambda-\nabla}\, b \right\}\Big)\\
            &=\left\{ a \,_{-\Lambda-\nabla}\, b \right\}-\displaystyle\sum_{i,j \in  I}(-1)^{\tilde{a}i+\tilde{b}j+ij}\left\{ \theta_i\,_{\Lambda+\nabla}\, a \right\}_{\rightarrow}(C^{-1})_{ij}(\Lambda+\nabla)\left\{ b\,_{\Lambda}\, \theta_j \right\}.\\   
        \end{align*}
Here, for the last equality, we used the fact that $(C^{-1})_{ij}(\Lambda)=(-1)^{ij+i+j+1}(C^{-1})^{*}_{ji}(\Lambda)$.
Then, by multiplying $(-1)^{\tilde{a}\tilde{b}}$ on the both sides, we get the skewsymmetry.\\
  
Now, we need to show the following equality to see the Jacobi identity:
 \begin{equation} \label{SUSY Dirac Jacobi}
  \{ a \,{}_{\Lambda} \left\{ b \,{}_{\Gamma}\, c \right\}^{D} \}^{D}+(-1)^{\tilde{a}\tilde{b}+\tilde{a}+\tilde{b}}\{ b \,{}_{\Gamma} \left\{ a \,{}_{\Lambda}\, c \right\}^{D} \}^{D}+(-1)^{\tilde{a}}\{ \left\{ a \,{}_{\Lambda}\, b \right\}^{D} \,{}_{\Lambda+\Gamma}\, c \}^{D}=0.
 \end{equation}
Each of term in the LHS of \eqref{SUSY Dirac Jacobi} can be expanded as follows. The first term in \eqref{SUSY Dirac Jacobi} is 
  \allowdisplaybreaks
  \begin{align}
    \nonumber&\{ a \,{}_{\Lambda} \left\{ b \,{}_{\Gamma}\, c \right\}^{D}\}^{D}\\
    \nonumber&=\left\{ a \,_\Lambda\, \left\{ b \,_{\Gamma}\, c \right\} \right\}^{D}-\displaystyle\sum_{\alpha,\beta \in I}\{ a \,{}_{\Lambda}\, (-1)^{\tilde{b}\tilde{c}+\tilde{b}\alpha+\tilde{c}\beta+\alpha\beta} \left\{ \theta_{\beta} \,{}_{\Gamma+\nabla}\, c \right\}\left(C^{-1}\right)_{\beta\alpha}(\Gamma+\nabla)\left\{ b \,{}_{\Gamma}\, \theta_{\alpha} \right\}\}^{D}\\
    &=\{ a {}_{\Lambda} \left\{ b \,{}_{\Gamma}\, c \right\}\}^{D}\label{Jacobi 1-1} \\
    &\hspace{0.5cm}-\displaystyle\sum_{\alpha,\beta\in  I}(-1)^{\tilde{b}\tilde{c}+\tilde{b}\alpha+\tilde{c}\beta+\alpha\beta}\left\{ a \,{}_{\Lambda}\, \left\{ \theta_{\beta} \,{}_{\Gamma+\nabla}\, c \right\} \right\}^{D}\left(C^{-1}\right)_{\beta\alpha}(\Gamma+\nabla)\left\{ b \,{}_{\Gamma}\, \theta_{\alpha} \right\}\label{Jacobi 1-2}\\
    &\hspace{0.5cm}-\displaystyle\sum_{\alpha,\beta\in  I}(-1)^{\tilde{b}\tilde{c}+\tilde{b}\alpha+\tilde{c}\beta+\alpha\beta+(\tilde{a}+1)(\tilde{c}+\beta+1)}\left\{ \theta_{\beta} \,{}_{\Lambda+\Gamma+\nabla}\, c \right\}\{ a \,{}_{\Lambda}\,\left(C^{-1}\right)_{\beta\alpha}(\Gamma+\nabla)\}^{D}\left\{ b \,{}_{\Gamma}\, \theta_{\alpha} \right\}\label{Jacobi 1-3}\\
    &\hspace{0.5cm}-\displaystyle\sum_{\alpha,\beta\in  I}(-1)^{\tilde{b}\tilde{c}+\tilde{b}\alpha+\tilde{c}\beta+\alpha\beta+(\tilde{a}+1)(\tilde{c}+\alpha)}\left\{ \theta_{\beta} \,{}_{\Lambda+\Gamma+\nabla}\, c \right\}\left(C^{-1}\right)_{\beta\alpha}( \Lambda+\Gamma+\nabla)\left\{ a \,{}_{\Lambda}\,\left\{ b \,{}_{\Gamma}\, \theta_{\alpha} \right\} \right\}^{D}\label{Jacobi 1-4}.
  \end{align}
The second term in \eqref{SUSY Dirac Jacobi} is 
  \allowdisplaybreaks
  \begin{align}
    \nonumber&\{ b \,{}_{\Gamma} \left\{ a \,{}_{\Lambda}\, c \right\}^{D}\}^{D}\\
    \nonumber&=\left\{ b \,_\Gamma\, \left\{ a \,_\Lambda\, c \right\} \right\}^{D}-\displaystyle\sum_{\alpha,\beta\in  I}\{ b \,{}_{\Gamma}\, (-1)^{\tilde{a}\tilde{c}+\tilde{a}\alpha+\tilde{c}\beta+\alpha\beta} \left\{ \theta_{\beta} \,{}_{\Lambda+\nabla}\, c \right\}\left(C^{-1}\right)_{\beta\alpha}(\Lambda+\nabla)\left\{ a \,{}_{\Lambda}\, \theta_{\alpha} \right\}\}^{D}\\
    &=\{ b {}_{\Gamma} \left\{ a \,{}_{\Lambda}\, c \right\}\}^{D}\label{Jacobi 2-1} \\
    &\hspace{0.5cm}-\displaystyle\sum_{\alpha,\beta\in  I}(-1)^{\tilde{a}\tilde{c}+\tilde{a}\alpha+\tilde{c}\beta+\alpha\beta}\left\{ b \,{}_{\Gamma}\, \left\{ \theta_{\beta} \,{}_{\Lambda+\nabla}\, c \right\} \right\}^{D}\left(C^{-1}\right)_{\beta\alpha}(\Lambda+\nabla)\left\{ a \,{}_{\Lambda}\, \theta_{\alpha} \right\}\label{Jacobi 2-2}\\
    &\hspace{0.5cm}-\displaystyle\sum_{\alpha,\beta\in  I}(-1)^{\tilde{a}\tilde{c}+\tilde{a}\alpha+\tilde{c}\beta+\alpha\beta+(\tilde{b}+1)(\tilde{c}+\beta+1)}\left\{ \theta_{\beta} \,{}_{\Lambda+\Gamma+\nabla}\, c \right\}\{ b \,{}_{\Gamma}\,\left(C^{-1}\right)_{\beta\alpha}(\Lambda+\nabla)\}^{D}\left\{ a \,{}_{\Lambda}\, \theta_{\alpha} \right\}\label{Jacobi 2-3}\\
    &\hspace{0.5cm}-\displaystyle\sum_{\alpha,\beta\in  I}(-1)^{\tilde{a}\tilde{c}+\tilde{a}\alpha+\tilde{c}\beta+\alpha\beta+(\tilde{b}+1)(\tilde{c}+\alpha)}\left\{ \theta_{\beta} \,{}_{\Lambda+\Gamma+\nabla}\, c \right\}\left(C^{-1}\right)_{\beta\alpha}( \Lambda+\Gamma+\nabla)\left\{ b \,{}_{\Gamma}\,\left\{ a \,{}_{\Lambda}\, \theta_{\alpha} \right\} \right\}^{D}\label{Jacobi 2-4}
  \end{align}
and the third term in \eqref{SUSY Dirac Jacobi} is 
  \allowdisplaybreaks
  \begin{align}
    \nonumber&\{ \{ a {}_{\Lambda} b \}^{D} {}_{\Lambda+\Gamma} \,c \}^{D}\\
    \nonumber&=\{ \{ a {}_{\Lambda} b \} {}_{\Lambda+\Gamma} c \}^{D}-\displaystyle\sum_{\alpha,\beta\in  I}\{\; (-1)^{\tilde{a}\tilde{b}+\tilde{a}\alpha+\tilde{b}\beta+\alpha\beta}\{ \theta_{\beta} {}_{\Lambda+\nabla} b \}(C^{-1})_{\beta\alpha}(\Lambda+\nabla)\{ a {}_{\Lambda} \theta_{\alpha} \}\;\, {}_{\Lambda+\Gamma}\, c \}^{D}\\
    &=\{ \{ a {}_{\Lambda} b \} {}_{\Lambda+\Gamma}\, c \}^{D}\label{Jacobi 3-1}\\
    &-\displaystyle\sum_{\alpha,\beta\in  I}(-1)^{\tilde{a}\tilde{b}+\tilde{a}\alpha+\tilde{b}\beta+\alpha\beta+\tilde{a}\tilde{c}+\tilde{c}\beta}\{ \{ \theta_{\beta} {}_{\Lambda+\nabla} b \} {}_{\Lambda+\Gamma+\nabla}\; c \}_{\rightarrow}^{D}\left(C^{-1}\right)_{\beta\alpha}(\Lambda+\nabla)\{ a {}_{\Lambda} \theta_{\alpha} \}\label{Jacobi 3-2}\\
    &-\displaystyle\sum_{\alpha,\beta\in  I} A (-1)^{\tilde{a}\tilde{b}+\tilde{a}\alpha+\tilde{b}\beta+\alpha\beta}\{ \left(C^{-1}\right)_{\beta\alpha}(\Lambda+\nabla) {}_{\Lambda+\Gamma+\nabla}\; c \}_{\rightarrow}^{D}\{ a {}_{\Lambda} \theta_{\alpha} \}\{ \theta_{\beta} {}_{-\Gamma-\nabla} b \}\label{Jacobi 3-3}\\
  &-\displaystyle\sum_{\alpha,\beta\in  I} B (-1)^{ \tilde{a}\tilde{b}+\tilde{a}\alpha+\tilde{b}\beta+\alpha\beta} \{ \{ a {}_{\Lambda} \theta_{\alpha} \} {}_{\Lambda+\Gamma+\nabla}\; c \}_{\rightarrow}^{D} \left(C^{-1}\right)^{*}_{\beta\alpha}(\Gamma)
  \{ \theta_{\beta} {}_{-\Gamma-\nabla} b \},\label{Jacobi 3-4}
  \end{align}
where $A=(-1)^{\tilde{c}(\tilde{a}+\alpha+1)+(\tilde{b}+\beta+1)(\beta+\tilde{a}+\tilde{c})}$ and $B=(-1)^{(\tilde{b}+\beta+1)(\alpha+\beta+1)+(\alpha+\tilde{b})(\tilde{a}+\alpha+1+\tilde{c})}$.
  \noin Moreover, each Dirac reduced bracket in \eqref{Jacobi 1-1}$, \cdots,$ \eqref{Jacobi 3-4} can be expanded into two terms via \eqref{3.5}.
  When $(\largestar)$ indicates one of $\eqref{Jacobi 1-1}, \cdots, \eqref{Jacobi 3-4}$, let us denote  $(\largestar)=(\largestar.1)+(\largestar.2)$, where $(\largestar.1)$ is induced from the first term in the RHS of \eqref{3.5} and $(\largestar.2)$ is induced from the second term in the RHS of \eqref{3.5}.
  For example, \eqref{Jacobi 1-1} is 
 \begin{align*}
 &\{ a _\Lambda \{ b \,_\Gamma\, c\}\}^{D}\\
  \tag{\text{\ref{Jacobi 1-1}.1}}&=\{ a \,{}_{\Lambda}\, \{ b \,{}_{\Gamma}\, c \} \}\\
  \tag{\text{\ref{Jacobi 1-1}.2}}&-\displaystyle\sum_{\alpha, \beta\in  I}(-1)^{\tilde{a}(\tilde{b}+\tilde{c}+1)+\tilde{a}\alpha+(\tilde{b}+\tilde{c}+1)\beta+\alpha\beta} \left\{ \theta_{\beta} \,{}_{\Lambda+\nabla}\, \{ b \,{}_{\Gamma}\, c \} \right\}\left(C^{-1}\right)_{\beta\alpha}(\Lambda+\nabla)\left\{ a \,{}_{\Lambda}\, \theta_{\alpha} \right\}.
 \end{align*}   
Now, we aim to show the sum of three terms in each triple below is trivial.
  \begin{align*}
    &\Big((\ref{Jacobi 1-1}.1), (\ref{Jacobi 2-1}.1), (\ref{Jacobi 3-1}.1)\Big), 
    \Big((\ref{Jacobi 1-1}.2), (\ref{Jacobi 2-2}.1), (\ref{Jacobi 3-2}.1)\Big), 
    \Big((\ref{Jacobi 1-2}.1), (\ref{Jacobi 2-1}.2), (\ref{Jacobi 3-4}.1)\Big), \\
    &\Big((\ref{Jacobi 1-4}.1), (\ref{Jacobi 2-4}.1), (\ref{Jacobi 3-1}.2)\Big), 
    \Big((\ref{Jacobi 1-2}.2), (\ref{Jacobi 2-2}.2), (\ref{Jacobi 3-3}.1)\Big), 
    \Big((\ref{Jacobi 1-3}.1), (\ref{Jacobi 2-4}.2), (\ref{Jacobi 3-4}.2)\Big), \\
    &\Big((\ref{Jacobi 1-3}.2), (\ref{Jacobi 2-3}.2), (\ref{Jacobi 3-3}.2)\Big), 
    \Big((\ref{Jacobi 1-4}.2), (\ref{Jacobi 2-3}.1), (\ref{Jacobi 3-2}.2)\Big).
  \end{align*}
  
  First, since
  \begin{align*}
    (\ref{Jacobi 1-1}.1)&=\{ a \,{}_{\Lambda} \left\{ b \,{}_{\Gamma}\, c \right\}\},\\
    (\ref{Jacobi 2-1}.1)&=\{ b \,{}_{\Gamma} \left\{ a \,{}_{\Lambda}\, c \right\}\}, \\
    (\ref{Jacobi 3-1}.1)&=\{ \left\{ a \,{}_{\Lambda}\, b \right\} \,{}_{\Lambda+\Gamma}\, c \},
  \end{align*}
we have
\begin{equation*}
  (\ref{Jacobi 1-1}.1)+(-1)^{\tilde{a}\tilde{b}+\tilde{a}+\tilde{b}}(\ref{Jacobi 2-1}.1)+(-1)^{\tilde{a}}(\ref{Jacobi 3-1}.1)=0
\end{equation*}
which directly follows from the Jacobi identity of $\left\{ \cdot {}_{\Lambda} \cdot \right\}$. 

\vskip 3mm

The terms in the second triple are 
\begin{align*}
  (\ref{Jacobi 1-1}.2)&=-\displaystyle\sum_{\alpha, \beta\in  I}(-1)^{\tilde{a}(\tilde{b}+\tilde{c}+1)+\tilde{a}\alpha+(\tilde{b}+\tilde{c}+1)\beta+\alpha\beta}\left\{ \theta_{\beta} {}_{\Lambda+\nabla} \left\{ b {}_{\Gamma} c \right\} \right\}(C^{-1})_{\beta\alpha}(\Lambda+\nabla)\left\{ a {}_{\Lambda} \theta_{\alpha} \right\},\\
  (\ref{Jacobi 2-2}.1)&=-\displaystyle\sum_{\alpha, \beta\in  I}(-1)^{\tilde{a}\tilde{c}+\tilde{a}\alpha+ 
  \tilde{c}\beta+\alpha\beta}\left\{ b {}_{\Gamma} \left\{ \theta_{\beta} {}_{\Lambda+\nabla} c \right\} \right\}(C^{-1})_{\beta\alpha}(\Lambda+\nabla)\left\{ a {}_{\Lambda} \theta_{\alpha} \right\},\\
  (\ref{Jacobi 3-2}.1)&=-\displaystyle\sum_{\alpha,\beta\in  I}(-1)^{\tilde{a}\tilde{b}+\tilde{a}\alpha+\tilde{b}\beta+\alpha\beta+\tilde{c}(\tilde{a}+\beta)}\{ \{ \theta_{\beta} {}_{\Lambda+\nabla} b \} {}_{\Lambda+\Gamma+\nabla} c \}(C^{-1})_{\beta\alpha}(\Lambda+\nabla)\{ a {}_{\Lambda} \theta_{\alpha} \}.  
\end{align*}
Hence, 
  \begin{align*}
    &(\ref{Jacobi 1-1}.2)+(-1)^{\tilde{a}\tilde{b}+\tilde{a}+\tilde{b}}(\ref{Jacobi 2-2}.1)+(-1)^{\tilde{a}}(\ref{Jacobi 3-2}.1)\\
    &=\displaystyle\sum_{\alpha, \beta\in  I}(-1)^{\tilde{a}(\tilde{b}+\tilde{c}+1)+\tilde{a}\alpha+(\tilde{b}+\tilde{c}+1)\beta+\alpha\beta+1}\\
    &\Big(\{ \theta_{\beta} {}_{\Lambda+\nabla} \{ b \,{}_\Gamma\, c\}\}\!+\!(-1)^{\tilde{b}\beta+\tilde{b}+\beta}\{b \,_\Gamma \{ \theta_{\beta} \,{}_{\Lambda+\nabla}\, c\}\}\!+\!(-1)^{\beta}\{\{ \theta_{\beta} \,{}_{\Lambda+\nabla}\, b\} \,_{\Lambda+\Gamma+\nabla}\, c\}\Big) (C^{-1})_{\beta\alpha}(\Lambda+\nabla)\left\{ a \,{}_{\Lambda}\, \theta_\alpha \right\}=0.
  \end{align*} 
  
  \vskip 3mm
  
Let us observe the third triple. The first two terms are 
  \begin{align*}
   (\ref{Jacobi 1-2}.1)&=\displaystyle\sum_{\alpha, \beta\in  I}(-1)^{\tilde{b}\tilde{c}+\tilde{b}\alpha+\tilde{c}\beta+\alpha\beta+1}\{ a {}_{\Lambda} \{ \theta_{\beta} \,{}_{\Gamma+\nabla}\, c \} \}(C^{-1})_{\beta\alpha}(\Gamma+\nabla)\{ b {}_{\Gamma} \theta_\alpha \},\\
   (\ref{Jacobi 2-1}.2)&=-\displaystyle\sum_{\alpha, \beta\in  I}(-1)^{\tilde{b}(\tilde{a}+\tilde{c}+1)+\tilde{b}\alpha+(\tilde{a}+\tilde{c}+1)\beta+\alpha\beta}\{ \theta_{\beta} \,{}_{\Gamma+\nabla}\, \{ a {}_{\Lambda} c \} \}(C^{-1})_{\beta\alpha}(\Gamma+\nabla)\{ b {}_{\Gamma} \theta_\alpha \}.
  \end{align*}
In (\ref{Jacobi 3-4}.1), which is the third term in the triple, we switch $\alpha$ and $\beta$ and get 
  \begin{align*}
    &(\ref{Jacobi 3-4}.1)\\
    &=-\displaystyle\sum_{\alpha,\beta\in  I}(-1)^{\tilde{b}\tilde{c}+\tilde{c}\beta+\alpha+\beta+1}(-1)^{\tilde{b}\alpha}\{ \{ a \,{}_{\Lambda}\, \theta_{\beta} \} \,{}_{\Lambda+\Gamma+\nabla}\, c \}(C^{-1})^{*}_{\alpha\beta}(\Gamma)\{ b {}_{\Gamma} \theta_\alpha \}\\    
    &=-\displaystyle\sum_{\alpha,\beta\in  I}(-1)^{\tilde{b}\tilde{c}+\tilde{c}\beta+\alpha+\beta+1}(-1)^{\tilde{b}\alpha}(-1)^{\alpha\beta+\alpha+\beta+1}\{ \{ a \,{}_{\Lambda}\, \theta_{\beta} \} \,{}_{\Lambda+\Gamma+\nabla}\, c \}(C^{-1})_{\beta\alpha}(\Gamma+\nabla)\{ b {}_{\Gamma} \theta_\alpha \}.
  \end{align*}
  Then
   \begin{align*}
    &(\ref{Jacobi 1-2}.1)+(-1)^{\tilde{a}\tilde{b}+\tilde{a}+\tilde{b}}(\ref{Jacobi 2-1}.2)+(-1)^{\tilde{a}}(\ref{Jacobi 3-4}.1)\\
    &=\displaystyle\sum_{\alpha, \beta\in  I}(-1)^{\tilde{b}\tilde{c}+\tilde{b}\alpha+\tilde{c}\beta+\alpha\beta+1}\left(\{ a {}_{\Lambda} \{ \theta_\beta \,{}_{\Gamma+\nabla}\, c\}\}+(-1)^{\tilde{a}\beta+\tilde{a}+\beta}\{\theta_\beta \,{}_{\Gamma+\nabla} \{ a \,{}_{\Lambda}\, c\}\}+(-1)^{\tilde{a}}\{\{ a \,{}_\Lambda\, \theta_\beta\} \,_{\Lambda+\Gamma+\nabla}\, c\}\right)\\
    &\hspace{0.5cm} (C^{-1})_{\beta\alpha}(\Gamma+\nabla)\left\{ b \,{}_{\Gamma}\,\theta_\alpha \right\}=0.
  \end{align*} 

\vskip 3mm

The three terms in the fourth triple are 
  \begin{align*}
    &(\ref{Jacobi 1-4}.1)=-\displaystyle\sum_{\alpha, \beta\in  I}(-1)^{\tilde{b}\tilde{c}+\tilde{b}\alpha+\tilde{c}\beta+\alpha\beta+\tilde{a}\alpha+\tilde{a}\tilde{c}+\alpha+\tilde{c}}\left\{ \theta_{\beta} \,{}_{\Lambda+\Gamma+\nabla}\, c \right\}\left(C^{-1}\right)_{\beta\alpha}( \Lambda+\Gamma+\nabla)\left\{ a \,{}_{\Lambda}\,\left\{ b \,{}_{\Gamma}\, \theta_{\alpha} \right\} \right\},\\
    &(\ref{Jacobi 2-4}.1)=-\displaystyle\sum_{\alpha, \beta\in  I}(-1)^{\tilde{a}\tilde{c}+\tilde{a}\alpha+\tilde{c}\beta+\alpha\beta+\tilde{b}\alpha+\tilde{b}\tilde{c}+\alpha+c}\left\{ \theta_{\beta} \,{}_{\Lambda+\Gamma+\nabla}\, c \right\}\left(C^{-1}\right)_{\beta\alpha}( \Lambda+\Gamma+\nabla)\left\{ b \,{}_{\Gamma}\,\left\{ a \,{}_{\Lambda}\, \theta_{\alpha} \right\} \right\},\\
    &(\ref{Jacobi 3-1}.2)=-\displaystyle\sum_{\alpha, \beta\in  I}(-1)^{(\tilde{a}+\tilde{b}+1)\tilde{c}+(\tilde{a}+\tilde{b}+1)\alpha+\tilde{c}\beta+\alpha\beta}\left\{ \theta_\beta {}_{\Lambda+\Gamma+\nabla} c \right\}\left(C^{-1}\right)_{\beta\alpha}( \Lambda+\Gamma+\nabla)\left\{ \left\{ a {}_{\Lambda} b \right\} {}_{\Lambda+\Gamma} \theta_{\alpha} \right\}.
  \end{align*}
Hence, 
  \begin{align*}
    &(\ref{Jacobi 1-4}.1)+(-1)^{\tilde{a}\tilde{b}+\tilde{a}+\tilde{b}}(\ref{Jacobi 2-4}.1)+(-1)^{\tilde{a}}(\ref{Jacobi 3-1}.2)\\
    &=\displaystyle\sum_{\alpha, \beta\in  I}(-1)^{\tilde{b}\tilde{c}+\tilde{b}\alpha+\tilde{c}\beta+\alpha\beta+\tilde{a}\alpha+\tilde{a}\tilde{c}+\alpha+\tilde{c}+1}\left\{ \theta_\beta \,{}_{\Gamma+\nabla}\, c \right\}(C^{-1})_{\beta\alpha}(\Gamma+\nabla)\\
    &\hspace{0.5cm}  \left(\{ a {}_{\Lambda} \{ b \,{}_\Gamma\, \theta_\alpha\}\}+(-1)^{\tilde{a}\tilde{b}+\tilde{a}+\tilde{b}}\{b \,_\Gamma \{ a \,{}_{\Lambda}\, \theta_\alpha\}\}+(-1)^{\tilde{a}}\{\{ a \,{}_\Lambda\, b\} \,_{\Lambda+\Gamma}\, \theta_\alpha\}\right)=0.
  \end{align*}
  
\vskip 3mm

In the fifth triple, the first term is
  \begin{align*}
    &(\ref{Jacobi 1-2}.2)=\displaystyle\sum_{\alpha, \beta, \delta, \epsilon\in  I}(-1)^{\tilde{b}\tilde{c}+\tilde{b}\alpha+\tilde{c}\beta+\alpha\beta}(-1)^{\tilde{a}(\tilde{c}+\beta+1)+\tilde{a}\delta+(\tilde{c}+\beta+1)\epsilon+\delta\epsilon }\left\{ \theta_\epsilon \,{}_{X}\, \left\{ \theta_{\beta} \,{}_{Y}\, c \right\} \right\}\\
    &\hspace{1cm}\left(\Big\vert_{X=\Lambda+\nabla}\left(C^{-1}\right)_{\epsilon\delta}(\Lambda+\nabla)\left\{ a \,{}_{\Lambda}\, \theta_{\delta} \right\}\right)\left(\Big\vert_{Y=\Gamma+\nabla}\left(C^{-1}\right)_{\beta\alpha}(\Gamma+\nabla)\left\{ b \,{}_{\Gamma}\, \theta_{\alpha} \right\}\right)\\
    &=\displaystyle\sum_{\alpha, \beta, \delta, \epsilon\in  I}(-1)^{\tilde{a}\tilde{b}+\tilde{a}\tilde{c}+\tilde{a}\delta+\tilde{a}+\tilde{b}\tilde{c}+\tilde{b}\alpha+\tilde{b}\epsilon+\tilde{c}\beta+\tilde{c}\epsilon+\alpha\beta+\delta\epsilon+\epsilon}\left\{ \theta_\epsilon \,{}_{X}\, \left\{ \theta_{\beta} \,{}_{Y}\, c \right\} \right\}\\
    &\hspace{1cm}\left(\Big\vert_{Y=\Gamma+\nabla}\left(C^{-1}\right)_{\beta\alpha}(\Gamma+\nabla)\left\{ b \,{}_{\Gamma}\, \theta_{\alpha} \right\}\right)\left(\Big\vert_{X=\Lambda+\nabla}\left(C^{-1}\right)_{\epsilon\delta}(\Lambda+\nabla)\left\{ a \,{}_{\Lambda}\, \theta_{\delta} \right\}\right).
  \end{align*}
The second term in the fifth triple is
  \allowdisplaybreaks
  \begin{align*}
    &(\ref{Jacobi 2-2}.2)=\displaystyle\sum_{\alpha,\beta, \delta, \epsilon\in  I}(-1)^{\tilde{a}\tilde{c}+\tilde{a}\alpha+\tilde{c}\beta+\alpha\beta}(-1)^{\tilde{b}(\beta+\tilde{c}+1)+\tilde{b}\delta+(\beta+\tilde{c}+1)\epsilon+\delta\epsilon}\left\{ \theta_{\epsilon} {}_{Y} \left\{ \theta_\beta {}_{X} c \right\} \right\}\\
    &\hspace{1cm}\left(\Big\vert_{Y=\Gamma+\nabla}\left(C^{-1}\right)_{\epsilon\delta}(\Gamma+\nabla)\left\{ b \,{}_{\Lambda}\, \theta_{\delta} \right\}\right)\left(\Big\vert_{X=\Lambda+\nabla}\left(C^{-1}\right)_{\beta\alpha}(\Lambda+\nabla)\left\{ a \,{}_{\Lambda}\, \theta_{\alpha} \right\}\right)\\
    &=\displaystyle\sum_{\alpha,\beta, \delta, \epsilon \in  I}(-1)^{\tilde{a}\tilde{c}+\tilde{a}\delta+\tilde{c}\epsilon+\delta\epsilon+\tilde{b}\epsilon+\tilde{b}\tilde{c}+\tilde{b}+\tilde{b}\alpha+\beta\epsilon+c\beta+\beta+\alpha\beta}\left\{ \theta_{\beta} {}_{Y} \left\{ \theta_{\epsilon} \,{}_{X}\, c \right\} \right\}\\
    &\hspace{1cm}\left(\Big\vert_{Y=\Gamma+\nabla}\left(C^{-1}\right)_{\beta\alpha}(\Gamma+\nabla)\left\{ b \,{}_{\Gamma}\, \theta_{\alpha} \right\}\right)\left(\Big\vert_{X=\Lambda+\nabla}\left(C^{-1}\right)_{\epsilon\delta}(\Lambda+\nabla)\left\{ a \,{}_{\Lambda}\, \theta_{\delta} \right\}\right).
  \end{align*}
  The last equality of the above is obtained by switching $\alpha$ and $\delta$ (resp. $\beta$ and $\epsilon$).
 The third term in the fifth triple is 
  \begin{align*}
    &(\ref{Jacobi 3-3}.1)\\
    &=-\displaystyle\sum_{\alpha,\beta\in  I}(-1)^{\tilde{a}\alpha+\alpha\beta+\tilde{a}\tilde{c}+\tilde{c}\alpha+\tilde{b}\tilde{c}+\tilde{c}\beta+\tilde{a}\beta+\tilde{a}}\{ \left(C^{-1}\right)_{\beta\alpha}(\Lambda+\nabla) {}_{\Lambda+\Gamma+\nabla}\; c \}_{\rightarrow}\{ a {}_{\Lambda} \theta_{\alpha} \}\{ \theta_{\beta} {}_{-\Gamma-\nabla} b \}\\
    &=\displaystyle\sum_{\alpha,\beta,\delta,\epsilon\in  I}(-1)^{\tilde{a}\alpha+\alpha\beta+\tilde{a}\tilde{c}+\tilde{c}\alpha+\tilde{b}\tilde{c}+\tilde{c}\beta+\tilde{a}\beta+\tilde{a}}(-1)^{
      \tilde{c}(\beta+\alpha+\epsilon+\delta)+\alpha(\epsilon+\delta+1)+\epsilon\delta}(-1)^{(\delta+\alpha+1)(\beta+\epsilon+1)}\\
      &\hspace{1cm} \{ \left\{ \theta_\delta {}_{\Lambda+\nabla} \theta_{\epsilon} \right\} {}_{\Lambda+\Gamma+\nabla}\; c \}_{\rightarrow}\left(C^{-1}\right)_{\beta\epsilon}(-\Gamma-\nabla)\left(C^{-1}\right)_{\delta\alpha}(\Lambda+\nabla)\{ a {}_{\Lambda} \theta_{\alpha} \}\{ \theta_{\beta} {}_{-\Gamma-\nabla} b \}\\
    &=\displaystyle\sum_{\alpha,\beta,\delta,\epsilon\in  I}(-1)^{\tilde{a}\alpha+\tilde{a}\tilde{c}+\tilde{b}\tilde{c}+\beta\epsilon+\alpha\delta+\tilde{c}\delta+\tilde{c}\epsilon+\tilde{b}\delta+\tilde{a}\tilde{b}+\tilde{b}\beta}\left\{ \left\{ \theta_{\delta} {}_{X} \theta_\epsilon \right\} {}_{\Lambda+\Gamma+\nabla} c \right\}\\
    &\hspace{1cm}\left(C^{-1}\right)_{\epsilon\beta}\!(\Gamma+\nabla)\left\{ b \,{}_{\Gamma}\, \theta_\beta \right\}\left(\Big\vert_{X=\Lambda+\nabla}\left(C^{-1}\right)_{\delta\alpha}(\Lambda+\nabla)\left\{ a \,{}_{\Lambda}\, \theta_{\alpha} \right\}\right)\\
    &=\displaystyle\sum_{\alpha,\beta,\delta,\epsilon\in  I}(-1)^{\tilde{a}\delta+\tilde{a}\tilde{c}+\tilde{b}\tilde{c}+\alpha\beta+\delta\epsilon+\tilde{c}\epsilon+\tilde{c}\beta+\tilde{b}\epsilon+\tilde{a}\tilde{b}+\tilde{b}\alpha}\left\{ \left\{ \theta_{\epsilon} {}_{X} \theta_\beta \right\} {}_{\Lambda+\Gamma+\nabla} c \right\}\\
    &\hspace{1cm}\left(C^{-1}\right)_{\beta\alpha}\!(\Gamma+\nabla)\left\{ b \,{}_{\Gamma}\, \theta_{\alpha} \right\}\left(\Big\vert_{X=\Lambda+\nabla}\left(C^{-1}\right)_{\epsilon\delta}(\Lambda+\nabla)\left\{ a \,{}_{\Lambda}\, \theta_{\delta} \right\}\right).
  \end{align*}
  Here, the second equality follows from Lemma \ref{Lemma:Inverse operator-SUSY} and the last equality is obtained by substituting $\alpha$ with $\delta$, $\beta$ with $\alpha$, $\delta$ with $\epsilon$ and $\epsilon$ with $\beta$.
Hence,
  \begin{align*}
    &(\ref{Jacobi 1-2}.2)+(-1)^{\tilde{a}\tilde{b}+\tilde{a}+\tilde{b}}(\ref{Jacobi 2-2}.2)+(-1)^{\tilde{a}}(\ref{Jacobi 3-3}.1) \; \\
    &=\displaystyle\sum_{\alpha, \beta, \delta, \epsilon\in  I}(-1)^{\tilde{a}\tilde{b}+\tilde{a}\tilde{c}+\tilde{a}\delta+\tilde{a}+\tilde{b}\tilde{c}+\tilde{b}\alpha+\tilde{b}\epsilon+\tilde{c}\beta+\tilde{c}\epsilon+\alpha\beta+\delta\epsilon+\epsilon}\\
    &\hspace{0.5cm}  \left(\{ \theta_{\epsilon} {}_{\Lambda+\nabla} \{ \theta_\beta \,{}_{\Gamma+\nabla}\, c\}\}+(-1)^{\beta\epsilon+\beta+\epsilon}\{\theta_{\beta} \,{}_{\Gamma+\nabla} \{ \theta_{\epsilon} \,{}_{\Lambda+\nabla}\, c\}\}+(-1)^{\epsilon}\{\{ \theta_{\epsilon} \,{}_{\Lambda+\nabla}\, \theta_{\beta}\} \,_{\Lambda+\Gamma+\nabla}\, c\}\right)\\
    &\hspace{0.5cm}  \left(\left(C^{-1}\right)_{\beta\alpha}(\Gamma+\nabla)\left\{ b \,{}_{\Gamma}\, \theta_{\alpha} \right\}\right)\Big(\left(C^{-1}\right)_{\epsilon\delta}(\Lambda+\nabla)\left\{ a \,{}_{\Lambda}\, \theta_{\delta} \right\}\Big)=0.
  \end{align*} 
 
 \vskip 3mm

 The first term in the sixth triple can be expanded by Lemma \ref{Lemma:Inverse operator-SUSY} as follows:
  \begin{align*}
    &(\ref{Jacobi 1-3}.1)=\displaystyle\sum_{\alpha,\beta,\delta,\epsilon \in  I}(-1)^{\tilde{a}\tilde{c}+\tilde{a}\epsilon+\tilde{b}\tilde{c}+\tilde{b}\alpha+\tilde{c}\beta+\tilde{c}+\beta\epsilon+\epsilon\delta+\alpha\delta+\epsilon}\\
    &\left\{ \theta_{\beta} \,{}_{\Lambda+\Gamma+\nabla}\, c \right\}\left(C^{-1}\right)_{\beta\epsilon}( \Lambda+\Gamma+\nabla)\left\{ a {}_{\Lambda} \left\{ \theta_{\delta} {}_{Y} \theta_{\epsilon} \right\} \right\}\left(\Big\vert_{Y=\Gamma+\nabla}\left(C^{-1}\right)_{\delta\alpha}(\Gamma+\nabla)\left\{ b \,{}_{\Gamma}\, \theta_{\alpha} \right\}\right).
    \end{align*}
For the other two terms, by changing indices properly, we get 
  \begin{align*}  
    &(\ref{Jacobi 2-4}.2)=\displaystyle\sum_{\alpha, \beta, \delta, \epsilon\in  I}(-1)^{\tilde{a}\tilde{b}+\tilde{a}\tilde{c}+\tilde{a}\epsilon+\tilde{a}\delta+\tilde{b}\tilde{c}+\tilde{b}+\tilde{b}\alpha+\tilde{c}\beta+\tilde{c}+\beta\epsilon+\epsilon+\delta\epsilon+\delta+\alpha\delta}\\
    &\left\{ \theta_{\beta} \,{}_{\Lambda+\Gamma+\nabla}\, c \right\}\left(C^{-1}\right)_{\beta\epsilon}( \Lambda+\Gamma+\nabla)\left\{ \theta_{\delta} {}_{\Gamma+\nabla} \left\{ a {}_{\Lambda} \theta_{\epsilon} \right\} \right\}\left(C^{-1}\right)_{\delta\alpha}(\Gamma+\nabla)\left\{ b \,{}_{\Gamma}\, \theta_{\alpha} \right\}\\
    &(\ref{Jacobi 3-4}.2)=\displaystyle\sum_{\alpha, \beta, \delta, \epsilon\in  I}(-1)^{\tilde{a}\tilde{c}+\tilde{a}\epsilon+\tilde{b}\tilde{c}+\tilde{b}\alpha+\tilde{c}\beta+\tilde{c}+\beta\epsilon+\epsilon\delta+\alpha\delta+\epsilon}\\
    &\left\{ \theta_{\beta} \,{}_{\Lambda+\Gamma+\nabla}\, c \right\}\left(C^{-1}\right)_{\beta\epsilon}( \Lambda+\Gamma+\nabla)\left\{ \left\{ a {}_{\Lambda} \theta_{\delta} \right\} {}_{\Lambda+\Gamma+\nabla} \theta_\epsilon \right\}\left(C^{-1}\right)_{\delta\alpha}(\Gamma+\nabla)\left\{ b \,{}_{\Gamma}\, \theta_{\alpha} \right\}.
  \end{align*}
  Then
  \begin{align*}
    &(\ref{Jacobi 1-3}.1)+(-1)^{\tilde{a}\tilde{b}+\tilde{a}+\tilde{b}}(\ref{Jacobi 2-4}.2)+(-1)^{\tilde{a}}(\ref{Jacobi 3-4}.2) \; \\
    &=\displaystyle\sum_{\alpha,\beta,\delta, \epsilon \in  I}(-1)^{\tilde{a}\tilde{c}+\tilde{a}\epsilon+\tilde{b}\tilde{c}+\tilde{b}\alpha+\tilde{c}\beta+\tilde{c}+\beta\epsilon+\epsilon\delta+\alpha\delta+\epsilon}\left\{ \theta_{\beta} \,{}_{\Lambda+\Gamma+\nabla}\, c \right\}\left(C^{-1}\right)_{\beta\epsilon}( \Lambda+\Gamma+\nabla)\\
    &\hspace{0.5cm}\left(\{ a {}_{\Lambda} \{ \theta_\delta \,{}_{\Gamma+\nabla}\, \theta_{\epsilon}\}\}+(-1)^{\tilde{a}\delta+\tilde{a}+\delta}\{\theta_{\delta} \,{}_{\Gamma+\nabla} \{ a \,{}_{\Lambda}\, \theta_{\epsilon}\}\}+(-1)^{\tilde{a}}\{\{ a \,{}_{\Lambda}\, \theta_{\delta}\} \,_{\Lambda+\Gamma+\nabla}\, \theta_\epsilon\}\right)\\
    &\hspace{0.5cm}\left(C^{-1}\right)_{\delta\alpha}(\Gamma+\nabla)\left\{ b \,{}_{\Gamma}\, \theta_{\alpha} \right\}=0.
    \end{align*} 
 
 \vskip 3mm

Let us observe the seventh triple. We have
\allowdisplaybreaks
  \begin{align*}
    &(\ref{Jacobi 1-3}.2)\\
    &=-\displaystyle\sum_{\alpha, \beta, \delta, \epsilon\in  I}(-1)^{\tilde{b}\tilde{c}+\tilde{b}\alpha+\tilde{c}\beta+\alpha\beta+(\tilde{a}+1)(\tilde{c}+\beta+1)}(-1)^{\tilde{a}(\alpha+\beta+1)+\tilde{a}\delta+(\alpha+\beta+1)\epsilon+\delta\epsilon}\\
    &\hspace{2cm}\{ \theta_\beta \,{}_{\Lambda+\Gamma+\nabla}\, c \}\{ \theta_{\epsilon} \,{}_{\Lambda+\nabla}\, (C^{-1})_{\beta\alpha}(\Gamma+\nabla) \}(C^{-1})_{\epsilon\delta}(\Lambda+\nabla)\left\{ a \,{}_{\Lambda}\, \theta_{\delta} \right\}\left\{ b \,{}_{\Gamma}\, \theta_{\alpha} \right\}\\
    &=-\displaystyle\sum_{\alpha, \beta, \delta, \epsilon, \zeta, \eta\in  I}(-1)^{\tilde{b}\tilde{c}+\tilde{b}\alpha+\tilde{c}\beta+\alpha\beta+(\tilde{a}+1)(\tilde{c}+\beta+1)}(-1)^{\tilde{a}(\alpha+\beta+1)+\tilde{a}\delta+(\alpha+\beta+1)\epsilon+\delta\epsilon}\\
    &\hspace{2cm}(-1)^{(\beta+\eta)(\alpha+\zeta)+(\epsilon+1)(\beta+\zeta+1)+1}(-1)^{(\alpha+\eta+1)(\tilde{a}+\epsilon)}\{ \theta_\beta \,{}_{\Lambda+\Gamma+\nabla}\, c \}(C^{-1})_{\beta\zeta}( \Lambda+\Gamma+\nabla)\\
    &\hspace{2cm}\left\{ \theta_\epsilon \,{}_{X}\, \left\{ \theta_{\eta} \,{}_{Y}\, \theta_{\zeta} \right\} \right\}\left(\Big\vert_{X=\Lambda+\nabla}\left(C^{-1}\right)_{\epsilon\delta}(\Lambda+\nabla)\left\{ a \,{}_{\Lambda}\, \theta_{\delta} \right\}\right)\left(\Big\vert_{Y=\Gamma+\nabla}\left(C^{-1}\right)_{\eta\alpha}(\Gamma+\nabla)\left\{ b \,{}_{\Gamma}\, \theta_{\alpha} \right\}\right)\\
    &=\displaystyle\sum_{\alpha, \beta, \delta, \epsilon, \zeta, \eta\in  I}(-1)^{\tilde{a}\tilde{c}+\tilde{a}\eta+\tilde{a}\delta+\tilde{a}+\tilde{b}\tilde{c}+\tilde{b}\alpha+\tilde{c}\beta+\tilde{c}+\beta\zeta+\alpha\eta+\delta\epsilon+\epsilon\eta+\epsilon\zeta+\epsilon+\zeta\eta+\zeta}\{ \theta_\beta \,{}_{\Lambda+\Gamma+\nabla}\, c \}(C^{-1})_{\beta\zeta}( \Lambda+\Gamma+\nabla)\\
    &\hspace{2cm}\left\{ \theta_\epsilon \,{}_{X}\, \left\{ \theta_{\eta} \,{}_{Y}\, \theta_{\zeta} \right\} \right\}\left(\Big\vert_{X=\Lambda+\nabla}\left(C^{-1}\right)_{\epsilon\delta}(\Lambda+\nabla)\left\{ a \,{}_{\Lambda}\, \theta_{\delta} \right\}\right)\left(\Big\vert_{Y=\Gamma+\nabla}\left(C^{-1}\right)_{\eta\alpha}(\Gamma+\nabla)\left\{ b \,{}_{\Gamma}\, \theta_{\alpha} \right\}\right).
    \end{align*}
    
   Again, by  Lemma \ref{Lemma:Inverse operator-SUSY}, the other two terms in the triple can be expressed as follows:
    \begin{align*}          
    &(\ref{Jacobi 2-3}.2)\\
    &=\displaystyle\sum_{\alpha, \beta, \delta, \epsilon, \zeta, \eta\in  I}(-1)^{\tilde{a}(\tilde{c}+\delta+\tilde{b}+\eta)+\tilde{b}(\tilde{c}+\alpha+1)+\tilde{c}(\beta+1)+\beta\zeta+\alpha\eta+\eta(\zeta+1)+\delta\epsilon+\zeta\epsilon+\zeta}\{ \theta_\beta \,{}_{\Lambda+\Gamma+\nabla}\, c \}(C^{-1})_{\beta\zeta}( \Lambda+\Gamma+\nabla)\\
    &\hspace{2cm}\left\{ \theta_\eta \,{}_{Y}\, \left\{ \theta_{\epsilon} \,{}_{X}\, \theta_{\zeta} \right\} \right\}\left(\Big\vert_{X=\Lambda+\nabla}\left(C^{-1}\right)_{\epsilon\delta}(\Lambda+\nabla)\left\{ a \,{}_{\Lambda}\, \theta_{\delta} \right\}\right)\left(\Big\vert_{Y=\Gamma+\nabla}\left(C^{-1}\right)_{\eta\alpha}(\Gamma+\nabla)\left\{ b \,{}_{\Gamma}\, \theta_{\alpha} \right\}\right),
  \end{align*}
  \begin{align*}
    &(\ref{Jacobi 3-3}.2)\\
    &=\displaystyle\sum_{\alpha, \beta, \delta, \epsilon, \zeta, \eta\in  I}(-1)^{\tilde{a}\delta+\tilde{a}\tilde{c}+\tilde{b}\tilde{c}+\tilde{c}+\zeta+\tilde{c}\beta+\beta\zeta+\zeta\eta+\zeta\epsilon+\delta\epsilon+\epsilon\eta+\tilde{a}\eta+\tilde{b}\alpha+\alpha\eta}\{ \theta_\beta \,{}_{\Lambda+\Gamma+\nabla}\, c \}(C^{-1})_{\beta\zeta}( \Lambda+\Gamma+\nabla)\\
    &\hspace{2cm}\left\{ \theta_\eta \,{}_{Y}\, \left\{ \theta_{\epsilon} \,{}_{X}\, \theta_{\zeta} \right\} \right\}\left(\Big\vert_{X=\Lambda+\nabla}\left(C^{-1}\right)_{\epsilon\delta}(\Lambda+\nabla)\left\{ a \,{}_{\Lambda}\, \theta_{\delta} \right\}\right)\left(\Big\vert_{Y=\Gamma+\nabla}\left(C^{-1}\right)_{\eta\alpha}(\Gamma+\nabla)\left\{ b \,{}_{\Gamma}\, \theta_{\alpha} \right\}\right).
  \end{align*}
Hence,
  \begin{align*}
    &(\ref{Jacobi 1-3}.2)+(-1)^{\tilde{a}\tilde{b}+\tilde{a}+\tilde{b}}(\ref{Jacobi 2-3}.2)+(-1)^{\tilde{a}}(\ref{Jacobi 3-3}.2)\\
    &=\displaystyle\sum_{\alpha,\beta,\delta, \epsilon, \zeta, \eta \in  I}(-1)^{\tilde{a}\tilde{c}+\tilde{a}\eta+\tilde{a}\delta+\tilde{a}+\tilde{b}\tilde{c}+\tilde{b}\alpha+\tilde{c}\beta+\tilde{c}+\beta\zeta+\alpha\eta+\delta\epsilon+\epsilon\eta+\epsilon\zeta+\epsilon+\zeta\eta+\zeta}\{ \theta_\beta \,{}_{\Lambda+\Gamma+\nabla}\, c \}(C^{-1})_{\beta\zeta}( \Lambda+\Gamma+\nabla)\\
    &\hspace{0.5cm}  \left(\{ \theta_{\epsilon} {}_{\Lambda+\nabla} \{ \theta_\eta \,{}_{\Gamma+\nabla}\, \theta_{\zeta}\}\}+(-1)^{\epsilon\eta+\epsilon+\eta}\{\theta_\eta \,{}_{\Gamma+\nabla} \{ \theta_{\epsilon} \,{}_{\Lambda+\nabla}\, \theta_{\zeta}\}\}+(-1)^{\epsilon}\{\{ \theta_{\epsilon} \,{}_{\Lambda+\nabla}\, \theta_{\eta}\} \,_{\Lambda+\Gamma+\nabla}\, \theta_{\zeta}\}\right)\\
    &\hspace{0.5cm} \left(\Big\vert_{X=\Lambda+\nabla}\left(C^{-1}\right)_{\epsilon\delta}(\Lambda+\nabla)\left\{ a \,{}_{\Lambda}\, \theta_{\delta} \right\}\right)\left(\Big\vert_{Y=\Gamma+\nabla}\left(C^{-1}\right)_{\eta\alpha}(\Gamma+\nabla)\left\{ b \,{}_{\Gamma}\, \theta_{\alpha} \right\}\right)=0.
  \end{align*} 
  
  \vskip 3mm
  
Finally, the three terms in the last triple are 
  \begin{align*}
    (\ref{Jacobi 1-4}.2)&=\displaystyle\sum_{\alpha, \beta, \delta, \epsilon\in  I}(-1)^{\tilde{b}\tilde{c}+\tilde{b}\alpha+\tilde{c}\beta+\alpha\beta+\tilde{a}\tilde{c}+\tilde{a}\alpha+\tilde{c}+\alpha}(-1)^{(\tilde{b}+\alpha+1)\tilde{a}+(\tilde{b}+\alpha+1)\epsilon+\tilde{a}\delta+\delta\epsilon}\\
    &\left\{ \theta_{\beta} \,{}_{\Lambda+\Gamma+\nabla}\, c \right\}(C^{-1})_{\beta\alpha}( \Lambda+\Gamma+\nabla)\{ \theta_{\epsilon} \,{}_{\Lambda+\nabla}\, \{ b \,{}_{\Gamma}\, \theta_{\alpha} \} \}(C^{-1})_{\epsilon\delta}(\Lambda+\nabla)\{ a \,{}_{\Lambda}\, \theta_{\delta} \},   \\       
    (\ref{Jacobi 2-3}.1)&=\displaystyle\sum_{\alpha, \beta, \delta, \epsilon\in  I}(-1)^{\tilde{a}\tilde{c}+\tilde{a}\delta+\tilde{c}\beta+\beta\delta+\tilde{b}\tilde{c}+\tilde{b}\beta+\tilde{b}+\tilde{c}+\beta}(-1)^{(\beta+\epsilon)(\delta+\alpha)+(\tilde{b}+1)(\beta+\alpha+1)+1}\\
    &\hspace{0.5cm}\left\{ \theta_{\beta} \,{}_{\Lambda+\Gamma+\nabla}\, c \right\}(C^{-1})_{\beta\alpha}( \Lambda+\Gamma+\nabla)\{ b \,{}_{\Gamma}\, \{ \theta_{\epsilon} \,{}_{\Lambda+\nabla}\, \theta_{\alpha} \} \}(C^{-1})_{\epsilon\delta}(\Lambda+\nabla)\{ a \,{}_{\Lambda}\, \theta_{\delta} \},
  \end{align*}
  \begin{align*}
    (\ref{Jacobi 3-2}.2)&=\displaystyle\sum_{\alpha, \beta, \delta, \epsilon\in  I}(-1)^{\tilde{a}\tilde{b}+\tilde{a}\delta+\tilde{b}\epsilon+\epsilon\delta+\tilde{a}\tilde{c}+\tilde{c}\epsilon}(-1)^{\tilde{c}(\epsilon+\tilde{b}+1)+\tilde{c}\beta+(\epsilon+\tilde{b}+1)\alpha+\alpha\beta}\\
    &\hspace{0.5cm}\left\{ \theta_{\beta} \,{}_{\Lambda+\Gamma}\, c \right\}(C^{-1})_{\beta\alpha}( \Lambda+\Gamma+\nabla)\{ \{ \theta_{\epsilon} \,{}_{\Lambda+\nabla}\, b \} \,{}_{\Lambda+\Gamma+\nabla}\, \theta_{\alpha} \}(C^{-1})_{\epsilon\delta}(\Lambda+\nabla)\{ a \,{}_{\Lambda}\, \theta_{\delta} \}.
  \end{align*}
Hence, 
  \begin{align*}
    &(\ref{Jacobi 1-4}.2)+(-1)^{\tilde{a}\tilde{b}+\tilde{a}+\tilde{b}}(\ref{Jacobi 2-3}.1)+(-1)^{\tilde{a}}(\ref{Jacobi 3-2}.2)\\
    &=\sum_{\alpha,\beta,\delta, \epsilon \in  I}(-1)^{bc+b\alpha+c\beta+\alpha\beta+ac+c+\alpha+ab+a+b\epsilon+\alpha\epsilon+\epsilon+a\delta+\delta\epsilon}\left\{ \theta_{\beta} \,{}_{\Lambda+\Gamma}\, c \right\}(C^{-1})_{\beta\alpha}( \Lambda+\Gamma+\nabla)\\
    &\hspace{0.5cm}  \left(\{ \theta_{\epsilon} \,{}_{\Lambda+\nabla} \{ b \,{}_\Gamma\, \theta_{\alpha}\}\}+(-1)^{\tilde{b}\epsilon+\tilde{b}+\epsilon}\{b \,_\Gamma \{ \theta_{\epsilon} \,{}_{\Lambda+\nabla}\, \theta_{\alpha}\}\}+(-1)^{\epsilon}\{\{ \theta_{\epsilon} \,{}_{\Lambda+\nabla}\, b\} \,_{\Lambda+\Gamma+\nabla}\, \theta_{\alpha}\}\right)\\
    &\hspace{0.5cm}  (C^{-1})_{\epsilon\delta}(\Lambda+\nabla)\left\{ a \,{}_{\Lambda}\, \theta_{\delta} \right\}=0.   
  \end{align*} 
Therefore, we showed the sum of each triple is zero so that the Jacobi identity holds.
  \end{proof}

\begin{thm} \label{Theorem:SUSY Dirac, part2}
Let $\mathscr{P}$ be a SUSY PVA and  $C(\Lambda)$ be the matrix in \eqref{3.3}, which is assumed to be invertible.
  \begin{itemize}
   \item[\textrm{(a)}] All the elements in $\theta_{ I}$ are central with respect to $\{\cdot {}_\Lambda \cdot\}^{D}$  in \eqref{3.5}.
   \item[\textrm{(b)}] Let $\left<\theta_{ I} \right>$ be the differential algebra  ideal of $\mathscr{P}$ generated by $\theta_{ I}$. Then the bracket $\left\{ \cdot {}_\Lambda \cdot \right\}^{D}$ induces a well-defined SUSY PVA bracket on $\mathscr{P}/\left<\theta_{ I} \right>$.
   \end{itemize}
\end{thm}

\begin{proof}

   For $i\in I$ and $a\in \mathscr{P}$, we have  
   \[ \{a{}_{\Lambda} \theta_i\}^{D}=\{a{}_{\Lambda}\theta_i\}-\sum_{\alpha, \beta \in I}(-1)^{(\tilde{a}+\beta)(i+\alpha)}\{\theta_{\beta}\,{}_{\Lambda+\nabla}\,\theta_i\}_{\rightarrow}\left(C^{-1}\right)_{\beta\alpha}(\Lambda+\nabla)\{a\,{}_{\Lambda}\,\theta_{\alpha}\}.\]
Since 
  \begin{equation}
    \displaystyle\sum_{\beta\in I}(-1)^{(i+\beta+i)(\beta+\alpha)}\left\{ \theta_\beta \,{}_{\Lambda+\nabla}\, \theta_i \right\}_{\rightarrow}(C^{-1})_{\beta\alpha}(\Lambda)=\delta_{i\alpha},
  \end{equation}
 the equality $\{a {}_\Lambda \theta_{i}\}^{D}=0$ holds, and $\{\theta_i{}_\Lambda a\}^{D}=0$ also holds by the skewsymmetry of $\{\cdot {}_\Lambda \cdot\}^{D}$.  Hence (a) is proved and (b) is a direct consequence of (a).
\end{proof}

\begin{section}{ W-superalgebras and modified Dirac reduction of PVSAs} \label{Sec:Structure of W-superalgebras}
    In this section, inspired by the definition of Dirac reduced bracket, we define a modified Dirac reduced bracket on a quotient space of PVSA. We also describe Poisson structures of  W-superalgebras as modified Dirac reductions of affine PVSAs. For the further properties of W-superalgebras, we refer to \cite{Suh18, Suh20}.
    
      \begin{subsection}{Modified Dirac reduction of PVSAs} \label{subsec: Dirac modification} \hfill
    
        Let $\mathcal{P}$ be a differential algebra freely generated by a finite homogeneous set $\theta_{J}:=\{\theta_j\,|\,j\in J\}$ with an even derivation $\partial$, i.e., $\mathcal{P}\simeq \CC[\partial^n \theta_j\,|\,j\in J, n\in \ZZ_{\geq 0}]$. 
        To simplify notations, we assume that the index set $J$ is a subset of $\ZZ$ and $j\in J$ is even (resp. odd) if $\theta_j$ is even (resp. odd). 
    Consider   the differential algebra ideal $\mathcal{I}$ of $\mathcal{P}$ generated by $ \theta_I: = \{ \theta_j \, | \, j \in I \subset J\}$ and let $\widetilde{\mathcal{P}}:= \mathcal{P} /\mathcal{I}$ be the quotient algebra. Then we can identify 
  \begin{equation} \label{tilde of P}
     \widetilde{\mathcal{P}}  \simeq \CC[\partial^n \theta_j\,|\,j\in J\setminus I , n\in \ZZ_{+}]
  \end{equation}
    as differential algebras.

    From now on we simply write as $\text{Mat}_{I}=\text{Mat}_{(r|s)}$ and $\text{Id}_{I}=\text{Id}_{(r|s)}$,          
    where $r$ (resp. $s$) denotes the number of even (resp. odd) elements in $\theta_I$. Note that $\widetilde{\mathcal{M}}(\lambda):=\text{Mat}_I\otimes\widetilde{\mathcal{P}}(\!(\lambda^{-1})\!)$ is a unital associative algebra with respect to the product $\circ$ in \eqref{eq:matrix product PVA}. Then $\text{Id}_{I}\otimes 1\in \widetilde{\mathcal{M}}(\lambda)$   for the identity matrix $\text{Id}_{I}\in\text{Mat}_I$ is the unity and an element $\widetilde{A}(\lambda)\in \widetilde{\mathcal{M}}(\lambda)$ is called \textit{invertible} if there exists $\widetilde{A}^{-1}(\lambda)\in \widetilde{\mathcal{M}}(\lambda)$ such that
      \begin{equation} \label{eq: inverse in quotient space}
      \widetilde{A}(\lambda) \circ \widetilde{A}^{-1}(\lambda)= \widetilde{A}^{-1}(\lambda) \circ \widetilde{A}(\lambda)= \text{Id}_{I}\otimes 1.
      \end{equation}  
      
      Now we suppose $\mathcal{P}$ is a PVSA endowed with the Poisson bracket $\{ \cdot {}_\lambda \cdot\}$ and consider
      \begin{equation} \label{eq: matrix for quotient PVA}
         \widetilde{C}(\lambda):= \sum_{i,j\in I} e_{ij} \otimes  \pi\{ \theta_{j}\, {}_{\lambda}\, \theta_{i} \}\in \widetilde{\mathcal{M}}(\lambda),
      \end{equation}
    where $\pi : \mathcal{P}(\!(\lambda^{-1})\!)\to \widetilde{\mathcal{P}}(\!(\lambda^{-1})\!)$ is  the canonical quotient map.    
    \begin{defn}
      Assume that $\widetilde{C}(\lambda)$ in \eqref{eq: matrix for quotient PVA} is invertible and $\widetilde{C}^{-1}(\lambda)=\left(\widetilde{C}_{ij}^{-1}(\lambda)\right)_{i,j\in I}$ is its inverse. Then the \textit{modified Dirac reduced bracket  $\pi\{\cdot {}_{\lambda}\cdot\}^D : \widetilde{\mathcal{P}}\times \widetilde{\mathcal{P}} \to \widetilde{\mathcal{P}}(\!(\lambda^{-1})\!)$ on  $\widetilde{\mathcal{P}}$ associated with $\theta_I$} is the map defined by 
      \begin{equation} \label{eq: DR bracket quotient PVA}
        \pi\{a{}_\lambda b\}^D=\pi\{a{}_{\lambda}b\}
        -\sum_{i,j\in I}(-1)^{(\tilde{a}+i)(\tilde{b}+j)+i+j}\pi\{\theta_i{}_{\lambda+\partial}b\}_{\rightarrow}\,(\widetilde{C}^{-1})_{ij}(\lambda+\partial)\,\pi\{a{}_{\lambda}\theta_j\}
        \end{equation}
        for $a,b \in \widetilde{\mathcal{P}}$. By \eqref{tilde of P}, the RHS of \eqref{eq: DR bracket quotient PVA} is computed regarding $a,b\in \CC[\partial^n \theta_j\,|j\in J\setminus I, n\in \ZZ_{+}]$.
    \end{defn}
    
    Note that the bracket $\pi\{\, \cdot \, {}_\lambda \, \cdot \, \}^D$ in  \eqref{eq: DR bracket quotient PVA} is not necessarily a PVSA $\lambda$-bracket. Especially, the Jacobi identity is not guaranteed 
     since $\pi$ is not a LCA homomorphism. In other words,  
    $  \pi\{a{}_{\lambda}\pi\{b{}_{\mu}c\}\}$ may not be same as $\pi\{a{}_{\lambda}\{b{}_{\mu}c\}\}$
    for  $a,b,c \in \widetilde{\mathcal{P}}$ so that the Jacobi identity of $\{\, \cdot \, {}_\lambda \, \cdot \, \}$ does not imply that of $\pi\{\, \cdot \, {}_\lambda \, \cdot \, \}^D$.
    However, other axioms of PVSA hold for $\pi\{\, \cdot \, {}_\lambda \, \cdot \, \}^D$ by the corresponding axioms for  $\{\, \cdot \, {}_\lambda \, \cdot \, \}$.
    
    \begin{prop}
     If the modified Dirac reduced bracket in \eqref{eq: DR bracket quotient PVA} satisfies the Jacobi identity then $\widetilde{\mathcal{P}}$ is a PVSA.
    \end{prop}
    
    \begin{proof}
    We need to show that \eqref{eq: DR bracket quotient PVA} satisfies the sesquilinearity, skew-symmetry and Leibniz rule. Since $\pi$ is differential algebra homomorphism, these properties can be proved by the similar proof to Theorem \ref{Theorem:super Dirac, part1}.
    \end{proof}

    Now we close this section with the following lemma, which is used to describe 
     the PVSA structures of W-superalgebras using modified Dirac reductions.
    \begin{lem} \label{lem:super modified dirac, basis change}
      Let $\theta'_I:=\{ \theta'_j \, | \, j\in I\}$ be  another basis of a vector superspace $V$ spanned by $\theta_I$. Suppose $(-1)^{\tilde{\theta}'_i}=(-1)^i=(-1)^{\tilde{\theta}_i}$. 
      If  $\widetilde{C}(\lambda)$ in \eqref{eq: matrix for quotient PVA} is invertible, then  $\overline{C}(\lambda)=\big(\pi\{\theta_j{}_{\lambda}\theta'_i\}\big)_{i,j\in I}$ is invertible and 
     the bracket \eqref{eq: DR bracket quotient PVA} can be written as
      \begin{equation} \label{eq: DR bracket basis change}
        \pi\{a{}_\lambda b\}^D=\pi\{a{}_{\lambda}b\}
          -\sum_{i,j\in I}(-1)^{(\tilde{a}+i)(\tilde{b}+j)+i+j}\pi\{\theta_i{}_{\lambda+\partial}b\}_{\rightarrow}\,\overline{C}^{-1}_{ij}(\lambda+\partial)\,\pi\{a{}_{\lambda}\theta'_j\}.
      \end{equation}
    \end{lem}
    \begin{proof}
      Take the invertible matrix $K=\sum_{i,j\in I} e_{ij}\otimes K_{ij} \in \widetilde{M}(\lambda)$ with $K_{ij}\in \CC$ satisfying 
      \begin{equation} \label{eq: basis change_1}
        \sum_{i\in I}e_i\otimes \theta'_i = K\, \Big(\sum_{j\in I} e_j\otimes \theta_j\Big).
      \end{equation}
    Observe that
     $K\, \widetilde{C}(\lambda)=\overline{C}(\lambda)$ and $\widetilde{C}(\lambda)$ is invertible. Hence the matrix $\overline{C}(\lambda)$ is also invertible and 
     \begin{equation}\label{eq: basis change_2}
       \overline{C}^{-1}(\lambda+\partial)=\widetilde{C}^{-1}(\lambda+\partial)K^{-1}.
     \end{equation}
    By \eqref{eq: basis change_1} and \eqref{eq: basis change_2}, we have
      \begin{equation}\label{eq: basis change_3}
      \sum_{j\in I} \big(\widetilde{C}^{-1}(\lambda+\partial) K^{-1} \big) \big( K (e_j \otimes \pi \{a{}_\lambda \theta_j\})\big) =   \sum_{j\in I} \overline{C}^{-1}(\lambda+\partial) (e_j \otimes \otimes \pi \{a{}_\lambda \theta'_j\}).
      \end{equation}
    By \eqref{eq: basis change_3}, we obtain
        \begin{equation}\label{eq: basis change_4}
        \begin{aligned}
        &  \sum_{i,j\in I} (-1)^{\tilde{a}\tilde{b}+\tilde{a}j}\left(e_i^T\otimes \pi\{\theta_i{}_{\lambda+\partial}b\}_{\rightarrow}\right)  \widetilde{C}^{-1}(\lambda+\partial)\left(e_j\otimes \pi\{a{}_{\lambda}\theta_j\}\right)\\
        & = \sum_{i,j\in I} (-1)^{\tilde{a}\tilde{b}+\tilde{a}j}\left(e_i^T\otimes \pi\{\theta_i{}_{\lambda+\partial}b\}_{\rightarrow}\right)  \overline{C}^{-1}(\lambda+\partial)\left(e_j\otimes \pi\{a{}_{\lambda}\theta'_j\}\right).
        \end{aligned}
      \end{equation}
    The last terms of \eqref{eq: DR bracket quotient PVA} and \eqref{eq: DR bracket basis change} equal to the LHS and RHS of \eqref{eq: basis change_4}, respectively. Hence we proved the lemma.
     \end{proof}

    \end{subsection}
    
    \begin{subsection}{Structures of W-superalgebras} \label{subsec: W-superalgebras} \hfill

      Let $\g$ be a finite simple Lie superalgebra with a nondegenerate supersymmetric invariant even bilinear form $(\cdot|\cdot)$. Suppose $F\in \g$ is an nilpotent element in an $\mathfrak{sl}_2$-triple $\{E,H,F\}$. Assume that the bilinear form is normalized by $(E|F)=1$. 
      Consider the eigenspace decomposition
    \begin{equation} \label{eq:LSA grading}
    \g=\oplus_{i\in\frac{\ZZ}{2}}\g(i),
    \end{equation}
    where $\g(i)=\{g\in \g\,|\,[\frac{H}{2},g]=i g\}$. Then $F\in \g(-1)$ and $E \in \g(1)$. Let us write
    \begin{equation} \label{eq: notation for g_?}
      \g_{\geq j}=\oplus_{i\geq j}\g(i),\quad \g_{\leq j}=\oplus_{i\leq j}\g(i),\quad \g_{>j}=\oplus_{i>j} \g(i),\quad \g_{<j}=\oplus_{i<j} \g(i)
    \end{equation}
    for each $j\in \frac{\ZZ}{2}$. In particular, we denote
    \begin{equation} \label{eq: notation2 for g_?}
      \mathfrak{m}:=\g_{\geq 1},\quad \mathfrak{n}:=\g_{>0}, \quad \mathfrak{p}:=\g_{<1}.
    \end{equation}
    On the other hand,  by the $\mathfrak{sl}_2$-representation theory, the Lie algebra $\g$ can be decomposed into 
    \begin{equation}
      \g=\g^F\oplus [E,\g]
      \end{equation}
     for $\g^F:=\{x\in \g\,|\, [F,x]=0\}\subset \g_{\leq 0}$. We denote by $a^{\sharp}$ the projection of $a\in \g$ to $\g^F$.

    Recall that $\mathcal{V}^k(\g)$ is the affine PVSA of level $k$ as defined in  Example \ref{ex: affine PVSA}. For a subspace $\mathfrak{a}$ of $\g$, the differential algebra 
    \begin{equation} \label{eq: differential algebra_P}
     \mathcal{P}(\mathfrak{a}):=S(\CC[\partial]\otimes \mathfrak{a})
    \end{equation}
    can be viewed as a subalgebra of $\mathcal{V}^k(\g)$ as differential algebras.

    \begin{defn} \label{def: def of W-alg}
    Let $\mathcal{J}_F$ be the differential algebra ideal of $\mathcal{V}^k(\g)$ generated by $\{m-(F|m)\,|\,m\in \mathfrak{m}\}$ and let $\rho: \mathcal{V}^k(\g)[\lambda]\rightarrow \mathcal{V}^k(\g)/\mathcal{J}_F[\lambda]$ be the canonical projection map. The {\it W-superalgebra of level $k$ associated with $\g$ and $F\in \g$} is defined by
    \begin{equation} 
      \mathcal{W}^k(\g, F):=\{a\in \mathcal{P}(\mathfrak{p})\,|\,\rho\{n{}_{\lambda} a\}=0 \text{ for any }n\in \mathfrak{n}\},
    \end{equation}
    where $\{\cdot {}_{\lambda} \cdot\}$ is the $\lambda$-bracket of $\mathcal{V}^k(\g)$.
    \end{defn}
    It is well known that $\mathcal{W}^k(\g, F)$ is a PVSA whose $\lambda$-bracket is induced from that of $\mathcal{V}^k(\g)$. More precisely, since there is a canonical differential algebra isomorphism $ \iota: \mathcal{V}^k(\g)/\mathcal{J}_F[\lambda] \to \mathcal{P(\mathfrak{p})}[\lambda]$, the bracket $\iota (\rho \{ \cdot {}_\lambda \cdot \})$ defined on $\mathcal{P}(\mathfrak{p})$ induces a PVSA $\lambda$-bracket on $\mathcal{W}^k(\g, F)$. In the rest of this section, we discuss  $\lambda$-brackets between  generators of the W-superalgebra described in Proposition \ref{prop: W-algebra generator}. \\

    Let us introduce a $\frac{\ZZ}{2}$-grading $\Delta$ on $\mathcal{V}^k(\g)$, called the \textit{conformal weight}. For $a\in  \g(i)$,  the conformal weight $\Delta_a$ of $a$ is defined by 
     \begin{equation}
       \Delta_{a}:=1-i.
     \end{equation}
    Let
    \begin{equation} \label{eq: property of conf wt(nonSUSY)}
      \Delta_{\partial A}=\Delta_{A}+1,\quad \Delta_{AB}=\Delta_A+\Delta_B,
    \end{equation}
     where $\Delta_A$ and $\Delta_B$ are conformal weights of homogeneous elements $A$ and $B$.
     The conformal weight $\Delta$ naturally induces the $\frac{\ZZ}{2}$-grading on the W-algebra $\mathcal{W}^k(\g, F)$, which we also denote by $\Delta$.
     Furthermore, we can find a homogeneous generating set of $\mathcal{W}^k(\g, F)$ which satisfy the properties in the following proposition. 
    
    \begin{prop} \cite{Suh18}\label{prop: W-algebra generator}
    Let $\{q_i\,|\,i\in J^F\}$ be a basis of $\g^F$, homogeneous with respect to both conformal weight and parity. Then there exists a unique subset $\{\omega_i\,|\, i\in J^F\}$ of homogeneous elements in $\mathcal{W}^k(\g, F)$ satisfying the following properties:
    \begin{itemize}
    \item[\rm(i)] $\mathcal{W}^k(\g,F)\simeq\CC[\partial^n \omega_i \,|\, i\in J^F, n\in \ZZ_{+}]$ as differential algebras,
    \item[\rm(ii)]  $\Delta_{\omega_i}=\Delta_{q_i}$ for $i\in J^F$,
    \item[\rm(iii)]  \label{prop: W-algebra generator2}
       $\omega_i$ for $i\in J^F$  is decomposed into 
      \begin{equation}
       \omega_i=q_i+\sum_{n\geq 1}\gamma_i^n 
      \end{equation}
      for $\gamma_i^n\in \mathcal{P}(\g^F)\otimes \left(\CC[\partial]\otimes [E, \g_{<0}]\right)^{\otimes n}$.
    \end{itemize}
    \end{prop}
    
    \begin{remark} \label{rmk: reducing omega nonSUSY}
     A generating set $\{\, \omega'_i \,  |\,  i\in J^F\}$ satisfying (i) and (ii) can be found by Drinfeld-Sokolov reduction (see \cite{Suh18}). Starting from  $\omega'_i$, we can find $\omega_i$ by induction on the conformal weights as follows. For $i\in J^F$, denote by $\mathfrak{p}[\Delta_i]$ the subspace of $\mathfrak{p}$ spanned by homogeneous elements whose conformal weights are strictly less than  $\Delta_{q_i}$. Then 
    \[ \omega'_i - q_i  =  \sum_{n \in \ZZ_+}{\gamma'}_i^{n}\in \mathcal{P}(\mathfrak{p}[\Delta_i]) \] 
    for some ${\gamma'}_i^{n}\in \mathcal{P}(\g^F)\otimes \left(\CC[\partial]\otimes [E, \g_{<0}]\right)^{\otimes n}$. The element ${\gamma'}_i^{0} \in \mathcal{P}(\g^F)$ can be written as a differential polynomial in $q_\alpha$'s where $\alpha \in J^F$ and $\Delta_{q_\alpha}< \Delta_{q_i}$. Consider the element ${\omega'}_i^0$ which is obtained from ${\gamma'}_i^{0}$ by  substituting $q_\alpha$'s with $\omega_\alpha$'s. Then 
    \[ \omega_i :=  \omega'_i - {\omega'}_i^0\]
    is the generator satisfying (i), (ii), (iii) in Proposition \ref{prop: W-algebra generator}.
    \end{remark}

    In addition, the uniqueness of  $\{\omega_i\, | \, i\in J^F\}$ in Proposition \ref{prop: W-algebra generator} follows from the fact that
    \begin{equation}
      \mathcal{W}^k(\g, F)\cap \sum_{n\geq 1} \mathcal{P}(\g^F)\otimes \left(\CC[\partial]\otimes [E, \g_{<0}]\right)^{\otimes n}=0.
    \end{equation}
    It is a direct consequence of the property  (i) in Proposition \ref{prop: W-algebra generator}. By the uniqueness, the map $\omega$ in Definition \ref{defn:omega} is completely determined.
    
    \begin{defn} \label{defn:omega}
    Recall the basis $\{q_i \, | \, i\in J^F\}$ of $\g^F$ and the generating set $\{\omega_i \, | \in i\in J^F\}$ of $\mathcal{W}^k(\g, F)$ in Proposition \ref{prop: W-algebra generator}.
     The injective linear map $\omega$ is defined by 
    \begin{equation} \label{eq: omega classical}
      \omega : \g^F \rightarrow \mathcal{W}^k(\g, F), \quad q_i \mapsto \omega_i.
    \end{equation}
    \end{defn}
    
    
    In order to describe the $\lambda$-bracket of $\mathcal{W}^k(\g, F)$, let us introduce some notations. 
    Let $\{q^i\,|\,i\in J^F\}$ be the basis of $\g^E:= \{ a \in \g| [E,a]=0\}$ which is dual to $\{q_i\,|\,i\in J^F\}$ in Proposition \ref{prop: W-algebra generator}, i.e., $(q^i|q_j)=\delta_{ij}$. For $\alpha_i \in \ZZ/2$ such that $q^i\in \g(\alpha_i)$, we denote 
     \begin{equation} \label{dual basis of g}
      q^i_m :=(\text{ad}F)^m q^i,   \quad   q_i^m :=k_{i,m}(\text{ad}E)^m q_i.
      \end{equation}
    Here, we have $0\leq m \leq 2\alpha_i$ and $k_{i,m}$'s are nonzero constant chosen to satisfy $(q^i_m|q_j^n)=\delta_{mn}\delta_{ij}$. For the explicit choices of $k_{i,m}$, see Section 3 of \cite{Suh20}.

    
      \begin{thm} {\cite{Suh20}} \label{Thm: suh nonSUSY}
    Recall the map $\omega$ in Definition \ref{defn:omega} and take $a\in \g(-t_1)\cap\g^F$ and $b\in \g(-t_2)\cap \g^F$. Identify the index set $J^F$ with a subset of $\ZZ$ so that $(-1)^{\tilde{q}_i} = (-1)^i$ for any $i\in J^F$.
    Then
      \begin{equation*}
        \begin{aligned}
        &\{\omega(a){}_{\lambda}\omega(b)\}=\omega([a,b])+k\lambda(a|b)\\
        &-\sum_{p\in \ZZ_{+}}\sum_{\substack{-t_2-1\prec (j_0,n_0)\prec \cdots\\ \cdots \prec(j_p,n_p)\prec t_1}}(-1)^{p+\tilde{a}\tilde{b}+\tilde{a}j_p+j_p+\tilde{b}j_0}\left(\omega([q^{j_0}_{n_0},b]^{\sharp})+(q^{j_0}_{n_0}|b)k(\lambda+\partial)\right)\\
        &\left[\prod_{t=1}^p (-1)^{j_{t-1}+j_{t-1}j_t}\left(\omega([q^{j_t}_{n_t},q^{n_{t-1}+1}_{j_{t-1}}]^{\sharp})+(q^{j_t}_{n_t}|q_{j_{t-1}}^{n_{t-1}+1})k(\lambda+\partial)\right)\right]\left(\omega([a,q_{j_p}^{n_p+1}]^{\sharp})+(a|q_{j_p}^{n_p+1})k\lambda\right)
        \end{aligned}
      \end{equation*}
      where
      \begin{itemize}
    \item   the pairs $(j_t,n_t)$ for $t=0,1,\cdots, p$ are elements of $I=\{(j,n)\,|\,j\in J^F, n=0,1, \cdots, 2\alpha_j-1\}$,
    \item   the partial order $\prec$ on $\frac{\ZZ}{2}\cup I$ is given by
    \begin{gather*}
      (j,n)\prec t \,  \Longleftrightarrow\, n-\alpha_j+1 \leq t, \quad  t\prec (j,n) \,\Longleftrightarrow\, t+1 \leq n-\alpha_j,\\
      (j,n)\prec (i,m)\, \Longleftrightarrow\, n-\alpha_j+1\leq m-\alpha_i.
    \end{gather*}
    \end{itemize}
      \end{thm}


    \end{subsection}


    \begin{subsection}{W-superalgebras via modified Dirac reduction}\ \hfill
    
    In this section, we continue using the notations defined in Section \ref{subsec: W-superalgebras}.
    Consider the subset 
    \begin{equation} \label{eq: constraints super}
      \theta_I=\{q^i_m-(F|q^i_m)\,|\,(i,m)\in I\} \subset \mathcal{V}^k(\g)
    \end{equation}
     for $I=\{(i,m)\,|\, i\in J^F, 0\leq m < 2\alpha_i\}$. The set \eqref{eq: constraints super} would play the role of $\theta_I$ in Section \ref{subsec: Dirac modification}. Recall in Theorem \ref{Thm: suh nonSUSY}  that the index set $J^F$ was considered as a subset of $\ZZ$ and assume $(-1)^{\tilde{q}_i}= (-1)^i$ for $i\in J^F$.  Hence, $(-1)^{\tilde{q}^i_m}= (-1)^{\tilde{q}_i^{m+1}}= (-1)^i$ for any $(i,m)\in I$.
     
     Let $\mathcal{I}$ be the differential algebra ideal of $\mathcal{V}^k(\g)$ generated by $\theta_I$ and denote the canonical projection map by 
     \[ \pi : \mathcal{V}^k(\g)\rightarrow  \mathcal{V}^k(\g)/{\mathcal{I}}=: \widetilde{ \mathcal{V}}^k(\g).\] Note that both $\{q^i_m\,|\,(i,m)\in I\}$ and $\{q_j^{n+1}\,|\,(j,n)\in I\}$ are the bases of $[E,\g]$. Therefore, by Lemma  \ref{lem:super modified dirac, basis change}, one can express the modified Dirac reduced bracket on $\widetilde{\mathcal{V}}^k(\g)$ associated with $\theta_I$ using 
     \begin{equation} \label{eq: C(lambda)}
      C(\lambda)=\sum_{(i,m),(j,n)\in I}e_{(j,n),(i,m)}\otimes \pi\{q^i_m{}_{\lambda}q_j^{n+1}\}\in \text{Mat}_{I}\otimes \widetilde{\mathcal{V}}^k(\g)(\!(\lambda^{-1})\!).
      \end{equation}

    
    Since the quotient space $\widetilde{\mathcal{V}}^k(\g)$ is isomorphic to $S(\CC[\partial]\otimes\g^F)$ as superalgebras, we would identify the elements of $\widetilde{\mathcal{V}}^k(\g)$ with the corresponding elements of $S(\CC[\partial]\otimes\g^F)$. Under this identification, the projection map $\pi: \mathcal{V}^k(\g) \rightarrow \widetilde{\mathcal{V}}^k(\g)\cong S(\CC[\partial]\otimes\g^F)$ is written as
    \begin{equation} \label{eq: pi}
    \pi(a)=a^{\sharp}+(F|a)
    \end{equation}
    for all $a\in \g$.
    
    \begin{lem} \label{Lem: C(lambda) entry}
    Recall the elements in \eqref{dual basis of g} and let $a\in \g(-t)\cap \g^F.$ For the map $\pi$ in  \eqref{eq: pi},  we have 
    \begin{itemize}
     \item[\rm{(1)}] $\pi\{q^i_m{}_{\lambda}q_i^{m+1}\}=1$ for any $(i,m)\in I$,
   \item[\rm{(2)}]  $\pi\{q^i_m{}_{\lambda}q_j^{n+1}\}=0$ if $n-\alpha_j+1>m-\alpha_i$,
   \item[\rm{(3)}]  $\pi\{a{}_{\lambda}q_j^{n+1}\}=0$ if $n-\alpha_j+1>t$,
   \item[\rm{(4)}]  $\pi\{q^j_n{}_{\lambda}a\}=0$ if $n-\alpha_j<-t$.
    \end{itemize}
    \end{lem}
    
    \begin{proof}
    Let us show (2). By \eqref{eq: pi}, we have 
    \begin{equation}
      \begin{aligned}
      \pi\{q^i_m{}_{\lambda}q_j^{n+1}\}&=[q^i_m, q_j^{n+1}]^{\sharp}+(q^i_{m+1}|q_j^{n+1})+(q^i_m|q_j^{n+1})k\lambda \\
      &=[q^i_m, q_j^{n+1}]^{\sharp}+\delta_{i,j}\delta_{m,n}+\delta_{i,j}\delta_{m,n+1}k\lambda.
      \end{aligned}
    \end{equation}
    Note that if the condition of (2) holds, then neither $(i,m)=(j,n)$ nor $(i,m)=(j,n+1)$. Also, $[q^i_m, q_j^{n+1}]\in \g((n-\alpha_j+1)-(m-\alpha_i))\subset \g_{>0}$ under the condition. Therefore, its projection to $\g^F$ vanishes. Hence we proved (2). Similarly, the others can be proved by direct computations.
    \end{proof}
    Then, by Lemma \ref{Lem: C(lambda) entry}, we have 
    \begin{equation} \label{eq: C(lambda) decompose}
    C(\lambda)=\sum_{(i,m)\in I}e_{(i,m),(i,m)}\otimes 1+\sum_{(j,n)\prec (i.m)} e_{(j,n),(i,m)} \otimes \pi\{q^i_m{}_\lambda q_j^{n+1}\},
    \end{equation}
    where $\prec$ is the partial order defined in Theorem \ref{Thm: suh nonSUSY}.

    \begin{prop} \label{Prop: inverse of C}
      The element $C(\lambda)$ in \eqref{eq: C(lambda)} is invertible and its inverse is
      \begin{equation}\label{eq:inverse of C}
        \begin{aligned}
      & (C^{-1})(\lambda)=\textup{Id}_I+\sum_{p\in \ZZ_{\geq 1}}\ \  \sum_{ \substack{(j_0,n_0)\prec(j_1,n_1) \prec \\ \prec   \cdots \prec(j_p,n_p)} }(-1)^{p+j_0+j_0j_p}(-1)^{j_{p-1}+j_{p-1}j_p}\\
      & \qquad e_{(j_0,n_0),(j_p,n_p)}\otimes
      \Big( \prod_{t=1}^{p-1} (-1)^{j_{t-1}+j_{t-1}j_t}\pi\{q^{j_t}_{n_t}\,{}_{\lambda+\partial}\,q_{j_{t-1}}^{n_{t-1}+1}\}_{\rightarrow}\Big)\pi\{q^{j_p}_{n_p}\,{}_{\lambda}\,q_{j_{p-1}}^{n_{p-1}+1}\},
        \end{aligned}
      \end{equation}
      where $(j_0, n_0), \cdots, (j_p, n_p) \in I$. 
      \end{prop}
      
      \begin{proof}
      By the equation \eqref{eq: C(lambda) decompose}, we can write
      \begin{equation} \label{eq: C=I+T}
      C(\lambda)= \text{Id}_I+T(\lambda)
      \end{equation}
      for the identity $\text{Id}_I$ in $\widetilde{\mathcal{M}}(\lambda)$ and $T(\lambda)=\sum_{(j,n)\prec (i.m)} e_{(j,n),(i,m)} \otimes \pi\{q^i_m{}_\lambda q_j^{n+1}\}.$ Then $C(\lambda)$ is invertible and its inverse is   \begin{equation}\label{eq:inverse of C, proof}
      \begin{aligned}
      (C^{-1})(\lambda) & = \text{Id}-T(\lambda)+T(\lambda) \circ T(\lambda)-T(\lambda) \circ T(\lambda) \circ T(\lambda)+ \cdots \\
      &    =\text{Id}-\sum_{p\in \ZZ_+} (-1)^p \,(T(\lambda+\partial))^pT(\lambda).
      \end{aligned}
      \end{equation}
      Equivalently, \[(C^{-1})(\lambda+\partial) =\sum_{p\in \ZZ_+} (-1)^p \,(T(\lambda+\partial))^p.\]
    Let us use the induction on $p$ to compute each term of $(C^{-1})(\lambda+\partial)$. By induction hypothesis,
      \allowdisplaybreaks
       \begin{equation*}
        \begin{aligned}
        &T(\lambda+\partial)^{p+1}=T(\lambda+\partial)^p\left(T(\lambda+\partial)\right)\\
        &= T(\lambda+\partial)^p \ \ \bigg(\sum_{(j_p,n_p)\prec(j_{p+1},n_{p+1})}e_{(j_p,n_p),(j_{p+1},n_{p+1})}\otimes \pi\{q^{j_{p+1}}_{n_{p+1}}\,{}_{\lambda+\partial}\,q_{j_p}^{n_p+1}\}_{\rightarrow}\bigg)\\
        &=\sum_{(j_0,n_0)\prec \cdots \prec(j_{p+1},n_{p+1})}(-1)^{(j_p+j_{p+1})(j_0+j_p)}(-1)^{j_0+j_0j_p}\\
        &\hskip 2.5cm e_{(j_0,n_0),(j_{p+1},n_{p+1})}\otimes (-1)^{j_p+j_pj_{p+1}}\prod_{t=1}^{p+1}(-1)^{j_{t-1}+j_{t-1}j_t}\pi\{q^{j_t}_{n_t}\,{}_{\lambda+\partial}\,q_{j_{t-1}}^{n_{t-1}+1}\}_{\rightarrow} \\
        &=\sum_{(j_0,n_0)\prec \cdots \prec(j_{p+1},n_{p+1})}(-1)^{j_0+j_0j_{p+1}}e_{(j_0,n_0),(j_{p+1},n_{p+1})}\otimes
        \prod_{t=1}^{p+1}(-1)^{j_{t-1}+j_{t-1}j_t}\pi\{q^{j_t}_{n_t}\,{}_{\lambda+\partial}\,q_{j_{t-1}}^{n_{t-1}+1}\}_{\rightarrow}.
        \end{aligned}
      \end{equation*}
      Thus, we have
      \begin{equation}
        \begin{aligned}
      \left(T(\lambda+\partial)\right)^p   = & \sum_{(j_0,n_0)\prec \cdots \prec(j_p,n_p)} (-1)^{j_0+j_0j_p}\\&e_{(j_0,n_0),(j_p,n_p)}\otimes
       \prod_{t=1}^p(-1)^{j_{t-1}+j_{t-1}j_t}\pi\{q^{j_t}_{n_t}\,{}_{\lambda+\partial}\,q_{j_{t-1}}^{n_{t-1}+1}\}_{\rightarrow}
        \end{aligned}
      \end{equation}
        which implies \eqref{eq:inverse of C}.
      \end{proof}
    
      The modified Dirac reduced bracket on $\widetilde{\mathcal{V}}^k(\g)$ immediately follows from Lemma \ref{Lem: C(lambda) entry} and Proposition \ref{Prop: inverse of C}.
    
      \begin{thm} \label{Thm: nonSUSY DR calculated}
      Recall the subset $\theta_I$ of $\mathcal{V}^k(\g)$ in \eqref{eq: constraints super}. If we denote the differential algebra ideal of $\mathcal{V}^k(\g)$ generated by $\theta_I$ by $\mathcal{I}$ and the canonical projection map by $\pi: \mathcal{V}^k(\g)\rightarrow \widetilde{\mathcal{V}}^k(\g) :=  \mathcal{V}^k(\g)/\mathcal{I}$, then the modified Dirac reduced bracket $\pi\{\cdot \,{}_{\lambda}\,\cdot\}^D$ on $\widetilde{\mathcal{V}}^k(\g)$ associated with $\theta_I$ can be written as 
      \begin{equation}
        \begin{aligned}
        \pi\{a{}_{\lambda}b\}^D=\pi\{a{}_{\lambda}b\} -\sum_{p\in \ZZ_{\geq 0}} & \sum_{\substack{-t_2-1\prec (j_0,n_0)\prec \cdots\\ \cdots \prec(j_p,n_p)\prec t_1}}(-1)^{p+\tilde{a}\tilde{b}+\tilde{a}j_p+j_p+\tilde{b}j_0}\pi\{q^{j_0}_{n_0}{}_{\lambda+\partial}b\}_{\rightarrow}\\
        &\left[\prod_{t=1}^p (-1)^{j_{t-1}+j_{t-1}j_t}\pi\{q^{j_t}_{n_t}{}_{\lambda+\partial}q_{j_{t-1}}^{n_{t-1}+1}\}_{\rightarrow}\right]\pi\{a{}_{\lambda}q_{j_p}^{n_p+1}\}
        \end{aligned}
      \end{equation}
      for any $a,b\in \g^F$.
      \end{thm}
    
      Note that the bracket $\{\cdot {}_{\lambda} \cdot\}$ in Theorem \ref{Thm: nonSUSY DR calculated} is the $\lambda$-bracket for $\mathcal{V}^k(\g)$ given in Example \ref{ex: affine PVSA}.
    Now, one can deduce the following theorem.
    
      \begin{thm} \label{Thm: nonSUSY main}
        The map $\omega: \g^F \rightarrow \mathcal{W}^k(\g,F)$ in \eqref{eq: omega classical} can be uniquely extended to the differential algebra isomorphism
        \[\omega: \widetilde{\mathcal{V}}^k(\g) \rightarrow \mathcal{W}^k(\g, F).\] 
         Moreover, the map $\omega$ naturally induces the PVSA $\lambda$-bracket of $\widetilde{\mathcal{V}}^k(\g) $, which is identical to the modified Dirac reduced bracket $\pi\{\cdot\,{}_{\lambda}\, \cdot\}^D$ in Theorem \ref{Thm: nonSUSY DR calculated}. 
     In other words, $\omega$ is a PVSA isomorphism between  $\left(\widetilde{\mathcal{V}}^k(\g), \pi\{\cdot{}_{\lambda}\cdot\}^D\right)$ and $\left(\mathcal{W}^k(\g, F), \{\cdot{}_{\lambda}\cdot\}\right)$.
      \end{thm}
      
      \begin{proof}
       The unique extension of $\omega$ and its bijectivity are ensured by Proposition \ref{prop: W-algebra generator}. Therefore, it is enough to check that $\omega$ is a PVSA homomorphism.
    Observe that
      \begin{equation}
        \pi\{q^{j_t}_{n_t}{}_{\lambda+\partial}q_{j_{t-1}}^{n_{t-1}+1}\}_{\rightarrow}=[q^{j_t}_{n_t},q_{j_{t-1}}^{n_{t-1}+1}]^{\sharp}+(q^{j_t}_{n_t}|q_{j_{t-1}}^{n_{t-1}+1})k(\lambda+\partial)
      \end{equation}
      for $(j_{t-1},n_{t-1})\prec (j_t, n_t)$.
      Thus, the desired statement immediately follows from Theorem \ref{Thm: suh nonSUSY} and Theorem \ref{Thm: nonSUSY DR calculated}.
      \end{proof}
    
      \begin{remark}
        For the subset  $\theta_I\subset \mathcal{V}^k(\g)$ in \eqref{eq: constraints super}, we cannot get invertible matrix if we consider
        \begin{equation} \label{eq: C(lambda) no projection}
          C(\lambda)=\sum_{(i,m),(j,n)\in I}e_{(j,n),(i,m)}\otimes \{q^i_m{}_{\lambda}q_j^{n+1}\}\in \text{Mat}_{I}\otimes\mathcal{V}^k(\g)\,(\!(\lambda^{-1})\!)
      \end{equation}
      instead of $\eqref{eq: C(lambda)}$.
       For example, let $\g=\mathfrak{sl}_2(\CC)$ and denote $F=e_{21}, E=e_{12}, H=e_{11}-e_{22}$. Then $\{F\}$ and $\{E\}$ are dual bases of $\g^F$ and $\g^E$, respectively, and  $C(\lambda)$ in \eqref{eq: C(lambda) no projection} is
      \begin{equation*}
        \begin{aligned}
        C(\lambda)&=
        \begin{pmatrix}
          \left\{E\ {}_{\lambda} -\frac{1}{2}H\right\} & \left\{E\ {}_{\lambda}-\frac{1}{2}E\right\} \\
          \left\{-H\ {}_{\lambda}-\frac{1}{2}H\right\} & \left\{-H\ {}_{\lambda}-\frac{1}{2}E\right\}
        \end{pmatrix}=
        \begin{pmatrix}
              E    & 0 \\
          k\lambda & E
        \end{pmatrix},
      \end{aligned}
      \end{equation*}
      which is not invertible.
      Therefore, we cannot use the Dirac reduction in Section \ref{Subsec:Dirac PVA}.
      \end{remark}
    
      We close this section by computing the modified Dirac reduced bracket in Theorem \ref{Thm: nonSUSY DR calculated} when $\g=\mathfrak{osp}(1|2)$. 
      
        \begin{ex} \label{ex: osp nonSUSY}
        Consider the Lie superalgebra $\g:=\mathfrak{osp}(1|2)\subset \mathfrak{gl}(2|1)$ equipped with the supertrace form $(\cdot\,|\,\cdot)$ given by 
        \begin{equation*}
          (A|B)=\text{str}(AB).
        \end{equation*}
        Denote the index set by $J=\{1, 2 , \bar{1}\}$, where $1, 2$ are even indices and $\bar{1}$ is an odd index. Let
        \begin{equation*}
          E=e_{12},\quad e=e_{1\bar{1}}+e_{\bar{1}2},\quad H=e_{11}-e_{22},\quad f=e_{\bar{1}1}-e_{2\bar{1}},\quad F=e_{21}.
        \end{equation*}
    Take the basis $\{q_a:=F, \, q_b:=f\}$ of $\g^F$ and its dual basis $\{q^a:=E,\, q^b:=\frac{1}{2}e\}$ of $\g^E$. 
        Order the index set $I$ for the matrix $C(\lambda)$ in \eqref{eq: C(lambda)} as
        \[I:=\{(a,0), (a,1), (b,0)\}.\]
       Then $C(\lambda)$ is 
        \begin{equation*}
          C(\lambda)=
          \begin{pmatrix}
            1 & \pi\{q^a_1{}_{\lambda}q_a^1\} & 0\\
            0 &               1               & 0\\
            0 &               0               & 1
          \end{pmatrix}
          =
          \begin{pmatrix}
            1 & k\lambda & 0\\
            0 &     1    & 0\\
            0 &     0    & 1
          \end{pmatrix}
        \end{equation*}
        and its inverse is 
        \begin{equation*}
          C^{-1}(\lambda)=
          \begin{pmatrix}
            1 & -k\lambda & 0\\
            0 &     1    & 0\\
            0 &     0    & 1
          \end{pmatrix}.
        \end{equation*}
        Therefore, for any two elements of $\CC[F, f]$, we can compute the bracket $\pi\{\cdot {}_{\lambda} \cdot\}^D$ between them. In particular,
        \begin{equation*}
          \pi\{F{}_{\lambda}f\}^D=-\left(f\left(-\frac{1}{2}k\lambda\right)+k(\lambda+\partial)(-f)\right)=\frac{3}{2}k\lambda f+k\partial f.
        \end{equation*}
      \end{ex}
    \end{subsection}
    \end{section}

\begin{section}{SUSY W-algebras and modified Dirac reduction of SUSY PVAs } \label{Sec:Structure of SUSY W-algebra}
        In this section,  we define a modified Dirac reduced bracket on a quotient space of SUSY PVA. We also show that SUSY W-algebras can be described by modified Dirac reductions of SUSY affine PVAs. For further properties of SUSY W-algebras, we refer to \cite{MRS21,Suh20}.
        
        \begin{subsection} {Modified Dirac reduction of SUSY PVAs} \label{subsec: SUSY DR for quotient space} \hfill

        Let $\mathscr{P}$ be a differential algebra freely generated by a finite homogeneous set $\theta_{J}:=\{\theta_j\,|\,j\in J\}$ with an odd derivation $D$, i.e., $\mathscr{P}\simeq \CC[D^n \theta_j\,|\,j\in J, n\in \ZZ_{+}]$. Assume that the index set $J$ is a subset of $\ZZ$ and $j\in J$ is even (resp. odd) if $\theta_j$ is even (resp. odd). 
        Consider   the differential algebra ideal $\mathscr{I}$ of $\mathscr{P}$ generated by $ \theta_I: = \{ \theta_j \, | \, j \in I \subset J\}$ and let $\widetilde{\mathscr{P}}:= \mathscr{P} /\mathscr{I}$ be the quotient algebra. Then we can identify 
        \begin{equation} \label{eq:tilde of SUSY P}
         \widetilde{\mathscr{P}}  \simeq \CC[D^n \theta_j\,|\,j\in J\setminus I , n\in \ZZ_{+}]
        \end{equation}
        as differential algebras.
        The space $\widetilde{\mathscr{M}}(\Lambda):=\text{Mat}_{I}\otimes \widetilde{\mathscr{P}}(\!(\Lambda^{-1})\!)$ is an associative algebra with unity $\text{Id}_I \otimes 1$ for the product defined in \eqref{3.1}. An invertible element in $\widetilde{\mathscr{M}}(\Lambda)$ is also defined as in \eqref{eq: SUSY invertible element}.  Now we suppose that $\mathscr{P}$ is a SUSY PVA with the $\Lambda$-bracket $\{ \cdot {}_\Lambda \cdot \}$ and consider the following odd element
         \begin{equation} \label{eq: C(Lambda) SUSY}
           \widetilde{C}(\Lambda):=\sum_{i,j\in I}e_{ij}\otimes \pi\{\theta_j{}_{\Lambda}\theta_i\}\in \widetilde{\mathscr{M}}(\Lambda)
         \end{equation}
        where $\pi: \mathscr{P}(\!( \Lambda^{-1})\!) \to \widetilde{\mathscr{P}}(\!(\Lambda^{-1})\!)$ is the canonical quotient map of differential algebras.
        
        \begin{defn} 
          Assume that $\widetilde{C}(\Lambda)$ in \eqref{eq: C(Lambda) SUSY} is invertible and let $\widetilde{C}^{-1}(\Lambda)\in \widetilde{\mathscr{M}}(\Lambda)$ be its inverse. Then the \textit{modified Dirac reduced bracket on $ \widetilde{\mathscr{P}}$ associated with $\theta_I$} is a bilinear map
          \begin{equation}
            \pi\{\cdot{}_{\Lambda}\cdot\}^D:  \widetilde{\mathscr{P}}\times  \widetilde{\mathscr{P}}\, \longrightarrow \,  \widetilde{\mathscr{P}}(\!(\Lambda^{-1})\!)
          \end{equation}
          given by
          \begin{equation} \label{eq: SUSY DR bracket for quotient}
            \pi\{a{}_{\Lambda} b\}^{D}=\pi\{a{}_{\Lambda}b\}-\sum_{i,j \in I}(-1)^{(\tilde{a}+j)(\tilde{b}+i)}\pi\{\theta_{j}\,{}_{\Lambda+\nabla}\,b\}_{\rightarrow} \widetilde{C}^{-1}_{ji}(\Lambda+\nabla)\pi\{a\,{}_{\Lambda}\,\theta_{i}\},
          \end{equation}
          where $\widetilde{C}^{-1}_{ij}(\Lambda)$ denotes the $ij$-entry of $\widetilde{C}^{-1}(\Lambda)$ and $a, b$ are elements in  $\widetilde{\mathscr{P}}.$ Regard $a,b\in \CC[D^n\theta_j\,|j\in J\setminus I, n\in \ZZ_+]$, when we compute the RHS.
        \end{defn}
        
        In general, the modified bracket \eqref{eq: SUSY DR bracket for quotient} does not give a SUSY PVA structure on $\widetilde{\mathscr{P}}$.  More precisely, the modified $\Lambda$-bracket satisfies the sesquilinearity, skew-symmetry and Leibniz rule, but Jacobi identity for SUSY PVAs. That is because the quotient map $\pi$ is only a differential algebra homomorphism, not a SUSY LCA algebra homomorphism.
        
        \vskip 2mm
        
        Now we introduce the SUSY analogue of Lemma \ref{lem:super modified dirac, basis change}.
        \begin{lem} \label{lem:SUSY modified dirac, basis change}
          Let $\theta'_I:=\{ \theta'_j \, | \, j\in I\}$ be  another basis of a vector superspace $V$ spanned by $\theta_I$ and suppose $(-1)^{\tilde{\theta}'_i}=(-1)^i=(-1)^{\tilde{\theta}_i}$. 
          If  $\widetilde{C}(\Lambda)$ in \eqref{eq: C(Lambda) SUSY} is invertible, then  $\overline{C}(\Lambda)=\big(\pi\{\theta_j{}_{\Lambda}\theta'_i\}\big)_{i,j\in I}$ is invertible and 
         the bracket \eqref{eq: SUSY DR bracket for quotient} can be written as
          \begin{equation} \label{eq: SUSY DR bracket basis change}
            \pi\{a{}_\Lambda b\}^D=\pi\{a{}_{\Lambda}b\}
              -\sum_{i,j\in I}(-1)^{(\tilde{a}+j)(\tilde{b}+i)}\pi\{\theta_j{}_{\Lambda+\nabla}b\}_{\rightarrow}\,\overline{C}^{-1}_{ji}(\Lambda+\nabla)\,\pi\{a{}_{\Lambda}\theta'_i\}.
          \end{equation}
        \end{lem}
        \begin{proof}
        A similar proof to  Lemma \ref{lem:super modified dirac, basis change} works.
        \end{proof}

        \end{subsection}
        
        \begin{subsection}{Structures of SUSY W-algebras} \label{subsec: SUSY W-algebras} \hfill
        
          Let $\g$ be a finite simple Lie superalgebra and $\mathfrak{s}=\text{Span}_{\CC}\{E,e,H=2x,f,F\}$ be a subalgebra isomorphic to $\mathfrak{osp}(1|2)$.  Here, $\{E, H, F\}$ is an $\mathfrak{sl}_2$-triple and $e,f$ are odd elements of $\g$ satisfying $[e,e]=2E$,  $[f,f]=-2F$, $[H,f]=-f$ and $[H,e]=e$. Assume that $\g$ is equipped with a nondegenerate supersymmetric invariant even bilinear form $(\cdot\,|\,\cdot)$ such that 
        $(E|F)=2(x|x)=1.$
        We can consider the eigenspace decomposition with respect to $\ad x$:
        \begin{equation} \label{eq: LSA grading SUSY}
          \g=\oplus_{i\in \frac{\ZZ}{2}}\g(i).
        \end{equation}
        In this section, we use the notation \eqref{eq: notation for g_?} and denote
        \begin{equation} \label{eq:SUSY n and p}
          \mathfrak{n}:= \g_{>0}, \quad \mathfrak{p}:=\g_{\leq 0}.
        \end{equation}
        Beware of that $\mathfrak{p}$ is different from the $\mathfrak{p}$ in Section \ref{subsec: W-superalgebras}. Here, $\mathfrak{p}$ is the complement of $\mathfrak{n}$ with respect to \eqref{eq: LSA grading SUSY}.
        On the other hand, by the $\mathfrak{osp}(1|2)$-representation theory, $\g$ can be decomposed into
        \begin{equation} \label{eq:decompose via e,f}
          \g=\g^f\oplus[e,\g],
        \end{equation}
        where $\g^f= \{ a\in \g\, | \, [f,a]=0\}\subset \g_{\leq 0}$. For each $a\in \g$, denote by $a^{\sharp}$ its projection to $\g^f$ with respect to \eqref{eq:decompose via e,f}.
        
        \vskip 2mm
          Recall that $\mathscr{V}^k(\bar{\g})$ is the affine SUSY PVA (see Example \ref{Ex:affine SUSY PVA}). For a subspace $\mathfrak{a}$ of $\g$, the 
        differential algebra $\mathscr{P}(\bar{\mathfrak{a}}):=S\left(\CC[D]\otimes \bar{\mathfrak{a}}\right)$ can be viewed as a subalgebra of $\mathscr{V}^k(\bar{\g})$ as differential algebras.
        
        \begin{defn}
         Let $\mathscr{J}_f$ be the differential algebra ideal of $\mathscr{V}^k(\bar{\g})$ generated by $\{\bar{n}-(f|n)\,|\,n\in \mathfrak{n}\}$ and let $\rho: \mathscr{V}^k(\bar{\g})[\Lambda]\rightarrow \mathscr{V}^k(\bar{\g})/\mathscr{J}_f[\Lambda]$ be the canonical projection. Then for $k\in \CC$, the {\it SUSY W-algebra of level $k$ associated with $\g$ and $f$} is defined by
          \begin{equation}
            \mathcal{W}^k(\bar{\g},f):=\{a\in \mathscr{P}(\bar{\mathfrak{p}})\,|\,\rho\{\bar{n}{}_{\Lambda}a\}=0 \text{ for any } n\in \mathfrak{n} \}
          \end{equation}
        where $\{\cdot {}_{\Lambda} \cdot\}$ is the $\Lambda$-bracket of $\mathscr{V}^k(\bar{\g})$.
        \end{defn}

        It is well-known that the SUSY W-algebra $\mathcal{W}^k(\bar{\g}, f)$ is a SUSY PVA with $\Lambda$-bracket naturally induced from that of $\mathscr{V}^k(\bar{\g})$. More precisely, for the canonical differential algebra isomorphism $\iota: \mathscr{V}^k(\bar{\g})/\mathscr{J}_f[\Lambda] \to \mathscr{P}(\bar{\mathfrak{p}})[\Lambda]$, the bracket $\iota (\rho\{\cdot {}_\Lambda \cdot \})$ defined on $\mathscr{P}(\bar{\mathfrak{p}})$ induces a SUSY PVA $\Lambda$-bracket on $\mathcal{W}^k(\bar{\g}, f)$.
        
        Let us introduce a $\frac{\ZZ}{2}$-grading $\Delta$ on $\mathscr{V}^k(\bar{\g})$ called \textit{conformal weight}. The conformal weight $\Delta_{\bar{a}}$ of $\bar{a}$ for $a\in \g(i)$  is defined by 
        \begin{equation}
          \Delta_{\bar{a}} := \frac{1}{2}-i,
        \end{equation} 
        and for homogeneous elements $A,B\in \mathscr{V}^k(\bar{\g})$, the conformal weights of $DA$ and $AB$ are 
        \begin{equation}
          \Delta_{D A}=\Delta_A+\frac{1}{2}, \quad \Delta_{AB}=\Delta_{A}+\Delta_{B}.
        \end{equation}
        The conformal weight $\Delta$ naturally induces the $\frac{\ZZ}{2}$-grading on the SUSY W-algebra $\mathcal{W}^k(\bar{\g}, f)$. Furthermore, there is a set of free generators of $\mathcal{W}^k(\bar{\g},f)$, which satisfies the properties in the following proposition.

        \begin{prop} \cite{Suh20}  \label{prop: SUSY W-alg properties}
          Let $\{r_i\,|\,i\in J^f\}$ be a basis of $\g^f$, homogeneous with respect to both conformal weight and parity. Then there exists a unique subset $\{\overline{\omega}_i\,|\,i\in J^f\}$ of homogeneous elements in $\mathcal{W}^k(\bar{\g},f)$ satisfying the follows properties:
          \begin{itemize}
            \item[\rm{(i)}] $\mathcal{W}^k(\bar{\g},f)\simeq \CC[D^n\overline{\omega}_i\,|\,i\in J^f, n\in \ZZ_{+}]$ as differential algebras,
            \item[\rm{(ii)}]  $\Delta_{\overline{\omega}_i}=\Delta_{\overline{r}_i}$ for each $i\in J^f$,
            \item[\rm{(iii)}]  $\overline{\omega}_i$ for $i\in J^f$ is decomposed into
        $\overline{\omega}_i=\overline{r}_i+\sum_{n\geq 1} \overline{\gamma}_i^n$\,
            for \, $\overline{\gamma}_i^n\in \mathscr{P}(\bar{\g}^f)\otimes\left(\CC[D]\otimes \overline{[e, \g_{\leq -\frac{1}{2}}]}\right)^{\otimes n}$.
          \end{itemize}
        \end{prop}
        
        The set $\{\overline{\omega}_i\,|\,i\in J^f\}$ in Proposition \ref{prop: SUSY W-alg properties} can be found by the SUSY analogue of Remark \ref{rmk: reducing omega nonSUSY} and the uniqueness of such set follows from the fact that 
        \begin{equation}
          \mathcal{W}^k(\bar{\g}, f)\cap \sum_{n\geq 1} \mathscr{P}(\bar{\g}^f)\otimes \left(\CC[D]\otimes \overline{[e, \g_{\leq -\frac{1}{2}}]}\right)^{\otimes n}=0.
        \end{equation}
        By the uniqueness, the linear map $\overline{\omega}$ in Definition \ref{Def: SUSY omega map} is determined completely.
        \begin{defn} \label{Def: SUSY omega map}
        Recall the basis $\{ r_i \, | \, i\in J^f\}$ of $\g^f$ and the set of generators $\{\overline{\omega}_i \, |\, i\in J^f\}$ in Proposition \ref{prop: SUSY W-alg properties}. The injective map $\overline{\omega}$ is defined by
        \begin{equation} \label{eq: SUSY omega map}
          \overline{\omega} : \g^f \rightarrow \mathcal{W}^k(\bar{\g},f),\quad r_i \mapsto \overline{\omega}_i.
        \end{equation}
        \end{defn}
        
        In order to see the $\Lambda$-bracket of $\mathcal{W}^k(\bar{\g}, f)$, we introduce the basis $\{r^j|j\in J^f\}$ of $\g^e:= \{a\in \g | [e,a]=0\}$ which is dual to $\{r_j \, |\, j\in J^f\}$  in Proposition \ref{prop: SUSY W-alg properties}, i.e., $(r^i|r_j)=\delta_{ij}$.  For $\beta_i \in \ZZ/2$ such that $r^i\in \g(\beta_i)$, we denote 
        \begin{equation} \label{eq: basis of g, r} 
         r^i_m :=(\text{ad}f)^m r^i,   \quad   r_i^m :=\bar{k}_{i,m}(\text{ad}e)^m r_i.
        \end{equation}
        Here, we have $0\leq m \leq 4\beta_i$ and $\bar{k}_{i,m}$'s are nonzero constants chosen to satisfy $(r^i_m|r_j^n)=\delta_{mn}\delta_{ij}$.  For the explicit choice of $\bar{k}_{i,m}$, we refer to Section 4 of \cite{Suh20}.

        \begin{thm} [\cite{Suh20}] \label{thm:SUSY W-algebra bracket}
          Recall the map $  \overline{\omega} : \g^f \rightarrow \mathcal{W}^k(\bar{\g},f)$ in \eqref{eq: SUSY omega map}. Then for any $a\in \g(-t_1)\cap \g^f$ and $b\in \g(-t_2)\cap \g^f$, we have
          \begin{equation}
            \begin{aligned}
              &\{ \overline{\omega}(a){}_{\Lambda} \overline{\omega}(b)\}=(-1)^{\tilde{a}}\left( \overline{\omega}([a,b])+k\chi(a|b)\right)\\
              &-\sum_{p\in \ZZ_{\geq 0}}\sum_{\substack{-t_2-\frac{1}{2}\llcurly (j_0,n_0)\llcurly \cdots\\ \cdots \llcurly(j_p,n_p)\llcurly t_1}}(-1)^{p+(\tilde{a}+1)(\tilde{b}+1)+(\tilde{a}+1)(j_p+n_p)+j_p+n_p+(\tilde{b}+1)(j_0+n_0+1)} \\
              & (-1)^{j_0+n_0}\left( \overline{\omega}([r^{j_0}_{n_0},b]^{\sharp})+(r^{j_0}_{n_0}|b)k(\chi+D)\right)\\
              &\left[\prod_{t=1}^p (-1)^{(j_{t-1}+n_{t-1})(j_t+n_t)}(-1)^{j_t+n_t}\left( \overline{\omega}([r^{j_t}_{n_t},r^{n_{t-1}+1}_{j_{t-1}}]^{\sharp})+(r^{j_t}_{n_t}|r_{j_{t-1}}^{n_{t-1}+1})k(\chi+D)\right)\right]\\
              &(-1)^{\tilde{a}+1}\left( \overline{\omega}([a,r_{j_p}^{n_p+1}]^{\sharp})+(a|r_{j_p}^{n_p+1})k\chi\right),
              \end{aligned}
          \end{equation}
          where 
            \begin{itemize}
        \item   the pairs $(j_t,n_t)$ for $t=0,1,\cdots, p$ are elements of $I=\{(j,n)\,|\,j\in J^f, j=0,1,\cdots, 4\beta_j-1\}$, 
        \item   the partial order $\llcurly$ on $\frac{\ZZ}{2}\cup I$ is given by
        \end{itemize}
          \begin{gather*}
            (j,n)\llcurly t \,  \Longleftrightarrow\, \frac{n}{2}-\beta_j+\frac{1}{2} \leq t, \quad    t\llcurly (j,n) \,\Longleftrightarrow\, t+\frac{1}{2} \leq \frac{n}{2}-\beta_j,\\
            (j,n)\llcurly (i,m) \,\Longleftrightarrow\, \frac{n}{2}-\beta_j+\frac{1}{2} \leq \frac{m}{2}-\beta_i.
            \end{gather*}\qed
        \end{thm}
        
        \end{subsection}
        
        \begin{subsection}{SUSY W-algebras via modified Dirac reduction}\hfill
        
        In this section, we use notations defined in Section \ref{subsec: SUSY W-algebras}. Recall the basis $\{r_i\, | \, i\in J^f\}$ of $\g^f$ in Proposition \ref{prop: SUSY W-alg properties} and bases of $\g$ in \eqref{eq: basis of g, r}. To simplify notations, we assume that $J^f$ is a subset of $\ZZ$ such that $(-1)^{\tilde{r}_i} = (-1)^i$ for any $i\in J^f$. 
        
         Let us consider the index set 
        \begin{equation}
        I= \{ (i,m) | \,  i\in J^f, \, m=0,1, \cdots, 4\beta_i-1\},
        \end{equation}
        where $\beta_i \in \frac{\ZZ}{2}$ satisfies $r_i \in \g(\beta_i)$. Then $\{r^i_m|(i,m)\in I\}$ is a basis of $[e,\g]$. Let $p(i,m):=i+m+1$ so that $p(i,m)=p(\bar{r}^i_m)$. In addition,  
        if we denote  \[(i,m)':= (i, 4\beta_i-m)\] for $i\in J^f$ and $m\in \ZZ$, then $\{ r_i^{m}| (i,m)'\in I\}$
        is another basis of $[e, \g]$. 
        Take a subset
         \begin{equation} \label{eq: SUSY constraints}
        \theta_I=\{\overline{r}^i_m-(f|r^i_m)\,|\,(i,m)\in I\}
         \end{equation}
         of $\mathscr{V}^k(\bar{\g})$ and denote the differential algebra ideal generated by $\theta_{I}$ in $\mathscr{V}^k(\bar{\g})$ by $\mathscr{I}$. Then the quotient space $\widetilde{\mathscr{V}}^k(\g):=\mathscr{V}^k(\bar{\g})/\mathscr{I}$ is isomorphic to $S(\CC[D]\otimes\bar{\g}^f)$ as differential algebras. The canonical projection map is defined by 
        \begin{equation} \label{eq: pi SUSY}
        \pi: \mathscr{V}^k(\bar{\g})\longrightarrow\widetilde{\mathscr{V}}^k(\g)\cong S(\CC[D]\otimes\bar{\g}^f), \quad \bar{a}\mapsto \bar{a}^\sharp+(f|a) \text{  for } \bar{a}\in \bar{\g}.
        \end{equation}
        Consider another basis of the vector superspace spanned by $\theta_I$: 
        \begin{equation}
          \theta'_{I}=\{\overline{r}_i^{m}-(f|r_i^{m})\,|\,(i,m)' \in I \}.
        \end{equation}
        By Lemma \ref{lem:SUSY modified dirac, basis change}, the modified Dirac reduced bracket associated with $\theta_I$ can be obtained by the odd matrix
        \begin{equation} \label{eq: C(Lambda)}
        \widetilde{C}(\Lambda)=\sum_{(i,m), (j,n+1)'\in I}e_{(j,n+1)',(i,m)}\otimes \pi\{\bar{r}^i_m{}_{\Lambda}\bar{r}_j^{n+1}\}\in \text{Mat}_I \otimes \widetilde{\mathscr{V}}^k(\bar{\g})(\!(\Lambda^{-1})\!).
        \end{equation}
        By direct computations, we get the following lemma.
        
        \begin{lem} \label{Lem: C(Lambda) entry}
        For $\{r^i_m\}, \{r_j^{n+1}\}$ in \eqref{eq: basis of g, r} and $a\in \g(-t)\cap \g^f$, the following identities hold: 
        \begin{itemize}
        \item[\textrm{(1)}] $\pi\{\bar{r}^i_m{}_{\Lambda}\bar{r}_i^{m+1}\}=(-1)^{i+m}$,
        \item[\textrm{(2)}] $\pi\{\bar{r}^i_m{}_{\Lambda}\bar{r}_j^{n+1}\}=0$ if $\frac{n}{2}-\beta_j+\frac{1}{2} > \frac{m}{2}-\beta_i,$
        \item[\textrm{(3)}] $\pi\{\bar{a}{}_{\Lambda}\bar{r}_j^{n+1}\}=0$ if $\frac{n}{2}-\beta_j+\frac{1}{2}>t$,
        \item[\textrm{(4)}] $\pi\{\bar{r}^j_n{}_{\Lambda}\bar{a}\}=-$ if $-t>\frac{n}{2}-\beta_j$.
        \end{itemize}
        \end{lem}
        
        By Lemma \ref{Lem: C(Lambda) entry}, we can write \eqref{eq: C(Lambda)} as
        \begin{equation} \label{eq: C(chi) rewrite}
        \widetilde{C}(\Lambda)=\sum_{(i,m)\in I}e_{(i,m+1)',(i,m)}\otimes(-1)^{i+m}+
        \sum_{(j,n)\llcurly(i,m)}e_{(j,n+1)',(i,m)}\otimes \pi\{\bar{r}^i_m{}_{\Lambda}\bar{r}_j^{n+1}\}
        \end{equation}
        for the partial order $\llcurly$ defined in Theorem \ref{thm:SUSY W-algebra bracket}
        
        \begin{prop} \label{Prop: SUSY inverse}
        The matrix $\widetilde{C}(\Lambda)$ in \eqref{eq: C(Lambda)} is invertible. Moreover, its inverse  $\widetilde{C}^{-1}(\Lambda)$ is given by
        \begin{equation}
          \begin{aligned}
        & \widetilde{C}^{-1}(\Lambda)=\sum_{(i,m)\in I}e_{(i,m), (i,m+1)'}\otimes(-1)^{i+m}\\
        & +\sum_{p\in \ZZ_{\geq 1}}\sum_{(j_0,n_0)\llcurly  \cdots \llcurly(j_p,n_p)}(-1)^{p+(j_0+n_0)(j_p+n_p)}e_{(j_0,n_0),(j_p,n_p+1)'}\otimes
         \widetilde{C}^{-1}_{(j_0, \cdots, j_p), (n_0, \cdots, n_p)},
          \end{aligned}
        \end{equation}
         where 
         \[ \widetilde{C}^{-1}_{(j_0, \cdots, j_p), (n_0, \cdots, n_p)}=  \Big(\prod_{t=1}^{p-1} R_t(\Lambda+\nabla)_{\rightarrow}\Big) R_p(\Lambda)\]
         and \[ R_s(\Lambda)=  (-1)^{(j_{s-1}+n_{s-1})(j_s+n_s)}\pi\{\bar{r}^{j_s}_{n_s}\,{}_{\Lambda}\,\bar{r}_{j_{s-1}}^{n_{s-1}+1}\}\]
         for $s=1,2, \cdots, p$.
         \end{prop}
        \begin{proof}
          Multiply the constant matrix $K:=\sum_{(i,m)\in I}e_{(i,m),(i,m+1)'}\otimes (-1)^{i+m}$ on the right of $\widetilde{C}(\Lambda)$, so that one can obtain
          \begin{equation}
           \widetilde{C}(\Lambda)\cdot K=\sum_{(i,m)\in I}e_{(i,m+1)',(i,m+1)'}\otimes 1+
          \sum_{(j,n)\llcurly(i,m)}e_{(j,n+1)',(i,m+1)'}\otimes (-1)^{i+m}\pi\{\bar{r}^i_m{}_{\Lambda}\bar{r}_j^{n+1}\}.
          \end{equation}
          Hence, if we write
          \begin{equation}
          \widetilde{T}(\Lambda)=\sum_{(j,n)\llcurly(i,m)}e_{(j,n+1)',(i,m+1)'}\otimes (-1)^{i+m}\pi\{\bar{r}^i_m{}_{\Lambda}\bar{r}_j^{n+1}\}
          \end{equation}
        then $\widetilde{C}(\Lambda)\cdot K$ can be expressed as
          \begin{equation}
          \widetilde{C}(\Lambda)\cdot K=\text{Id}+\widetilde{T}(\Lambda)
          \end{equation}
          for identity matrix $\text{Id}$. Since it is of the same form with \eqref{eq: C=I+T}, we can find its inverse with analogous calculation to Proposition \ref{Prop: inverse of C}. Denote the inverse by $\widehat{C}^{-1}$, i.e.,
          \begin{equation}
            \widehat{C}^{-1}(\Lambda+\nabla) \widetilde{C}(\Lambda)\cdot K=\text{Id}.
          \end{equation}
           Then by multiplying $K$ on the left side of $\widehat{C}^{-1}$, one get the inverse of $\widetilde{C}(\Lambda)$.
          \end{proof}
          
          Using Proposition \ref{Prop: SUSY inverse}, we can find a formula for the modified Dirac reduced bracket on $\widetilde{\mathscr{V}}^k(\g)$.
        
          \begin{thm} \label{thm: SUSY Dirac formula calculated}
            Consider the subset $\theta_{I}$ in \eqref{eq: SUSY constraints} and the differential algebra ideal $\mathscr{I}$ of $\mathscr{V}^k(\bar{\g})$ generated by $\theta_I$. Denote the canonical projection map by $\pi: \mathscr{V}^k(\bar{\g})\rightarrow \widetilde{\mathscr{V}}^k(\g):= \mathscr{V}^k(\bar{\g})/\mathscr{I}$. Then the modified Dirac reduced bracket $\pi\{\cdot {}_{\Lambda}\cdot\}^D$ on $\widetilde{\mathscr{V}}^k(\g)$ associated with $\theta_I$ can be written as 
            \begin{equation}
              \begin{aligned}
                \pi\{\bar{a}{}_{\Lambda}\bar{b}\}^D&=\pi\{\bar{a}{}_{\Lambda}\bar{b}\}-\sum_{p\in \ZZ_{\geq 0}}\sum_{\substack{-t_2-\frac{1}{2}\llcurly (j_0,n_0)\llcurly(j_1,n_1)\llcurly \cdots \\ \cdots \llcurly(j_p,n_p)\llcurly t_1}}(-1)^{p+(\tilde{a}+1)(\tilde{b}+1)+(\tilde{a})(j_p+n_p)+(\tilde{b}+1)(j_0+n_0+1)}\\
                &\pi\{\bar{r}^{j_0}_{n_0}{}_{\Lambda+\nabla}\bar{b}\}_{\rightarrow}\left[\prod_{t=1}^p (-1)^{(j_{t-1}+n_{t-1})(j_t+n_t)}\pi\{\bar{r}^{j_t}_{n_t}{}_{\Lambda+\nabla}\bar{r}_{j_{t-1}}^{n_{t-1}+1}\}_{\rightarrow}\right]\pi\{\bar{a}{}_{\Lambda}\bar{r}_{j_p}^{n_p+1}\}
                \end{aligned}
            \end{equation} 
            for any $a.b\in \g^f$, where the bracket $\{\cdot {}_{\Lambda} \cdot\}$ is the $\Lambda$-bracket for $\mathscr{V}^k(\bar{\g})$.
          \end{thm}
        
        Comparing Theorem \ref{thm:SUSY W-algebra bracket} and Theorem \ref{thm: SUSY Dirac formula calculated}, we obtain the following theorem.

        \begin{thm} \label{Thm: SUSY main}
          The map $\overline{\omega}: \g^f\rightarrow \mathcal{W}^{k}(\bar{\g}, f)$ in \eqref{eq: SUSY omega map} can be uniquely extended to the differential algebra isomorphism
          \begin{equation}
            \overline{\omega} : \widetilde{\mathscr{V}}^k(\g) \rightarrow \mathcal{W}^k(\bar{\g}, f).
          \end{equation}
          Then $\overline{\omega}$ naturally induces the SUSY PVA $\Lambda$-bracket on $ \widetilde{\mathscr{V}}^k(\g)$ which coincides with the modified Dirac reduced bracket $\pi\{\cdot {}_{\Lambda} \cdot\}^D$ in Theorem \ref{thm: SUSY Dirac formula calculated}. In other words, $\overline{\omega}$ is a SUSY PVA isomorphism between $(\widetilde{\mathscr{V}}^k(\g), \pi\{\cdot{}_{\Lambda}\cdot\}^D)$ and $(\mathcal{W}^k(\bar{\g},f), \{\cdot{}_{\Lambda} \cdot\})$.
        \end{thm}
        \begin{proof}
          The proof is similar to that of Theorem \ref{Thm: nonSUSY main}.
        \end{proof}
        
        In the end, we see how the bracket $\pi\{\cdot {}_{\Lambda}\cdot\}^D$ in Theorem \ref{thm: SUSY Dirac formula calculated} is computed in the example.
        \begin{ex}
        Let $\g=\mathfrak{osp}(1|2)$ be equipped with a supertrace form $(\cdot\,|\,\cdot)$. As in Example \ref{ex: osp nonSUSY}, denote the elements of $\g$ by $E, e, H, f$ and $F$. Take the basis  $\{r_a:=F\}$ of $\g^f$ and its dual basis $\{r^a:=E\}$ of $\g^e$. Then 
        \begin{gather*}
          r^a_0=r^a=E,\quad r^a_1=e,\quad r^a_2=H,\quad r^a_3=f,\quad r^a_4=-2F\\
          r_a^4=-\frac{1}{2}E,\quad r_a^3=-\frac{1}{2}e,\quad r_a^2=\frac{1}{2}H,\quad r_a^1=\frac{1}{2}f,\quad r_a^0=F.
        \end{gather*}
        The index set $I$ consists of the four elements:
        \begin{align*}
          (a, 4)'=(a,0),\quad (a,3)'= (a,1)\quad (a, 2)'=(a,2),\quad (a,1)'= (a,3),
        \end{align*}
        and the partial order in the index set $I$ is given by 
        \[(a,0)\llcurly (a,1)\llcurly (a,2) \llcurly (a,3).\]
        Then the odd matrix $\widetilde{C}(\Lambda)$ in \eqref{eq: C(Lambda)} is 
        \begin{equation*}
          \widetilde{C}(\Lambda)=
          \begin{pmatrix}
            0  &  0  &  0  &  -1 \\
            0  &  0  &  1  & \pi\{\bar{r}^a_3{}_{\Lambda} \bar{r}_a^3\} \\
            0  &  -1 &  \pi\{\bar{r}^a_2{}_{\Lambda}\bar{r}_a^2\} & \pi\{\bar{r}^a_3{}_{\Lambda}\bar{r}_a^2\} \\
            1  & \pi\{\bar{r}^a_1{}_{\Lambda}\bar{r}_a^1\} & \pi\{\bar{r}^a_2{}_{\Lambda}\bar{r}_a^1\} & \pi\{\bar{r}^a_3{}_{\Lambda}\bar{r}_a^1\} 
          \end{pmatrix}
          =
          \begin{pmatrix}
            0 & 0 & 0 & -1\\
            0 & 0 & 1 & -k\chi \\
            0 & -1 & k\chi & 0  \\
            1 & -k\chi & 0 & \overline{F}
          \end{pmatrix}
        \end{equation*}
        and its inverse is
        \begin{equation*}
          \widetilde{C}^{-1}(\Lambda)=
          \begin{pmatrix}
            -k^3\chi\lambda-\overline{F} & -k^2\lambda & k\chi & 1 \\
            k^2\lambda                       & -k\chi             & -1        & 0 \\
            k\chi                                  & 1                      & 0         & 0 \\
            -1                                          & 0                      & 0         & 0
          \end{pmatrix}.
        \end{equation*}
         Using this, we obtain the bracket $\pi\{\cdot {}_{\Lambda} \cdot \}^D$ on $\CC[D^nF\,|\,n\in \ZZ_+]$. In particular,
        \begin{align*}
          &\pi\{\overline{F}{}_{\Lambda}\overline{F}\}^D\\
          &=-k(\chi+D)k(\chi+D)\overline{F}-k(\chi+D)\left(-k^3(\chi+D)(\lambda+\partial)-\overline{F}\right)\left(\frac{1}{2}k\chi\right)+2\overline{F}k(\chi+D)\left(-\frac{1}{2}k\chi\right) \\
          &=-\frac{1}{2}k^5\lambda^2 \chi-\frac{3}{2}k^2\lambda \overline{F}-\frac{1}{2}k^2\chi D\overline{F}-k^2 \partial \overline{F}.
        \end{align*}
        \end{ex}
        
        \end{subsection}
        \end{section}

\end{document}